\documentclass[12pt,american]{article}
\usepackage[T1]{fontenc}
\usepackage[utf8]{inputenc}
\usepackage{color}
\usepackage{babel}
\usepackage{mathtools}
\usepackage{amsmath}
\usepackage{amsthm}
\usepackage{amssymb}
\usepackage{graphicx}
\usepackage[dvipsnames]{xcolor} % for colors
\usepackage[normalem]{ulem} % for \sout{...}
\usepackage[pdfusetitle,
 bookmarks=true,bookmarksnumbered=false,bookmarksopen=false,
 breaklinks=false,pdfborder={0 0 0},pdfborderstyle={},backref=false,colorlinks=true]
 {hyperref}
\hypersetup{
 linkcolor=magenta, urlcolor=blue, citecolor=blue}

\makeatletter

\providecommand{\tabularnewline}{\\}

\theoremstyle{plain}
\newtheorem{thm}{\protect\theoremname}
\theoremstyle{definition}
\newtheorem{defn}[thm]{\protect\definitionname}
\theoremstyle{remark}
\newtheorem{rem}[thm]{\protect\remarkname}
\theoremstyle{plain}
\newtheorem{prop}[thm]{\protect\propositionname}
\theoremstyle{plain}
\newtheorem{cor}[thm]{\protect\corollaryname}
\theoremstyle{plain}
\newtheorem{lem}[thm]{\protect\lemmaname}
\theoremstyle{plain}
\newtheorem{lyxalgorithm}[thm]{\protect\algorithmname}

\DeclareMathOperator{\Tr}{Tr}

\DeclareMathOperator{\id}{id}

\DeclareMathOperator{\imagem}{im}

\DeclareMathOperator{\adaptivedisc}{ad}
\DeclareMathOperator{\supp}{supp}

\allowdisplaybreaks

\numberwithin{equation}{section}
\usepackage{fullpage}
\usepackage{xfrac}

\makeatother

\providecommand{\algorithmname}{Algorithm}
\providecommand{\corollaryname}{Corollary}
\providecommand{\definitionname}{Definition}
\providecommand{\lemmaname}{Lemma}
\providecommand{\propositionname}{Proposition}
\providecommand{\remarkname}{Remark}
\providecommand{\theoremname}{Theorem}

\begin{document}
\title{Query complexities of quantum channel discrimination and estimation:
A unified approach}
\author{Zixin Huang\thanks{School of Science, STEM College, RMIT University, Melbourne, VIC 3000 Australia}
\and Johannes Jakob Meyer\thanks{Dahlem Center for Complex Quantum Systems, Freie Universit\"at Berlin, 14195 Berlin, Germany}
\and Theshani Nuradha\thanks{Department of Mathematics and
Illinois Quantum Information Science and Technology (IQUIST) Center,
University of Illinois Urbana-Champaign, Urbana, IL 61801, USA} \and Mark M. Wilde\thanks{School of Electrical and Computer Engineering,
Cornell University, Ithaca, New York 14850, USA}}
\maketitle
\begin{abstract}
The goal of quantum channel discrimination and estimation is to determine
the identity of an unknown channel 
from a discrete or continuous set, respectively.
The query complexity of these tasks is equal to the minimum
number of times one must call an unknown channel 
to identify it 
within a desired threshold on the
error probability. In this paper, we establish lower bounds on the
query complexities of channel discrimination and estimation, in both
the parallel and adaptive access models. 
We do so by establishing
new or applying known upper bounds on the squared Bures distance and
symmetric logarithmic derivative Fisher information of channels. 
Phrasing our statements and proofs in terms of isometric extensions of quantum channels allows us to give conceptually simple proofs for both novel and known bounds. We also provide alternative proofs for several established results in an effort to present a consistent and unified framework for quantum channel discrimination and estimation, which we believe will be helpful in addressing future questions in the field. 
\end{abstract}
\tableofcontents{}

\section{Introduction}

\subsection{Background}

One of the main aims of quantum information theory~\cite{Wilde2017,Hayashi2017,Watrous2018,Holevo2019,Khatri2024}
is to understand the distinction between the classical and quantum
theories of information. Non-classical effects like superposition
and entanglement lead to possibilities exceeding what can be achieved
in classical information theory alone, with protocols like teleportation
\cite{Bennett1993}, super-dense coding~\cite{Bennett1992}, and quantum
key distribution~\cite{Bennett1984,Ekert1991} being key examples.

Beyond these celebrated protocols, one can also find non-classical
features in the tasks of quantum channel discrimination~\cite{Sacchi2005,Chiribella2008}
and estimation~\cite{Fujiwara2001,Fujiwara2008}. These tasks involve
an unknown channel being selected from a predetermined set,
and an agent can query the unknown channel multiple times in order
to guess or estimate its identity. There are at least two access models
for querying the unknown channel, which include the parallel and adaptive
access models, as depicted in Figure~\ref{fig:parallel-adaptive-ch-disc-est}.
In channel discrimination, the unknown channel is selected from a
finite set, whereas in channel estimation, the unknown channel is
selected from a continuous set. As a particularly striking distinction
between classical and quantum capabilities in these settings, even
if two unitary channels are not perfectly distinguishable with a single
query, they always become perfectly distinguishable with a finite
number of queries~\cite{Acin2001}. Additionally, if a unitary channel
is selected from a smooth family, one can estimate it
with a precision that goes beyond what one can achieve in the classical
setting, dubbed the Heisenberg limit~\cite{Giovannetti2004}.

The plethora of rich quantum phenomena that can occur in channel discrimination
and estimation has motivated many researchers to pursue and understand
this setting for over two decades, so that channel discrimination
\cite{Acin2001,Duan2007,Chiribella2008,Duan2009,Piani2009,Hayashi2009,Matthews2010,Harrow2010,Cooney2016,Wang2019,Wilde2020,Katariya2021a,Salek2022}
and channel estimation~\cite{Fujiwara2001,Sarovar2006,Ji2008,Fujiwara2008,Escher2011,Hayashi2011,DemkowiczDobrzanski2012,Kolodynski2013,Demkowicz2014,Yuan2017,Katariya2021,Zhou2021,Kurdzialek2023}
now constitute their own subfields within quantum information science
research. Additionally, it is understood that there are strong links
between the two domains~\cite{Yuan2017,Katariya2021,Meyer2025}, given
that, roughly speaking, channel estimation can be viewed as discriminating
nearby channels.

\begin{figure}
\begin{centering}
\includegraphics[width=0.65\textwidth]{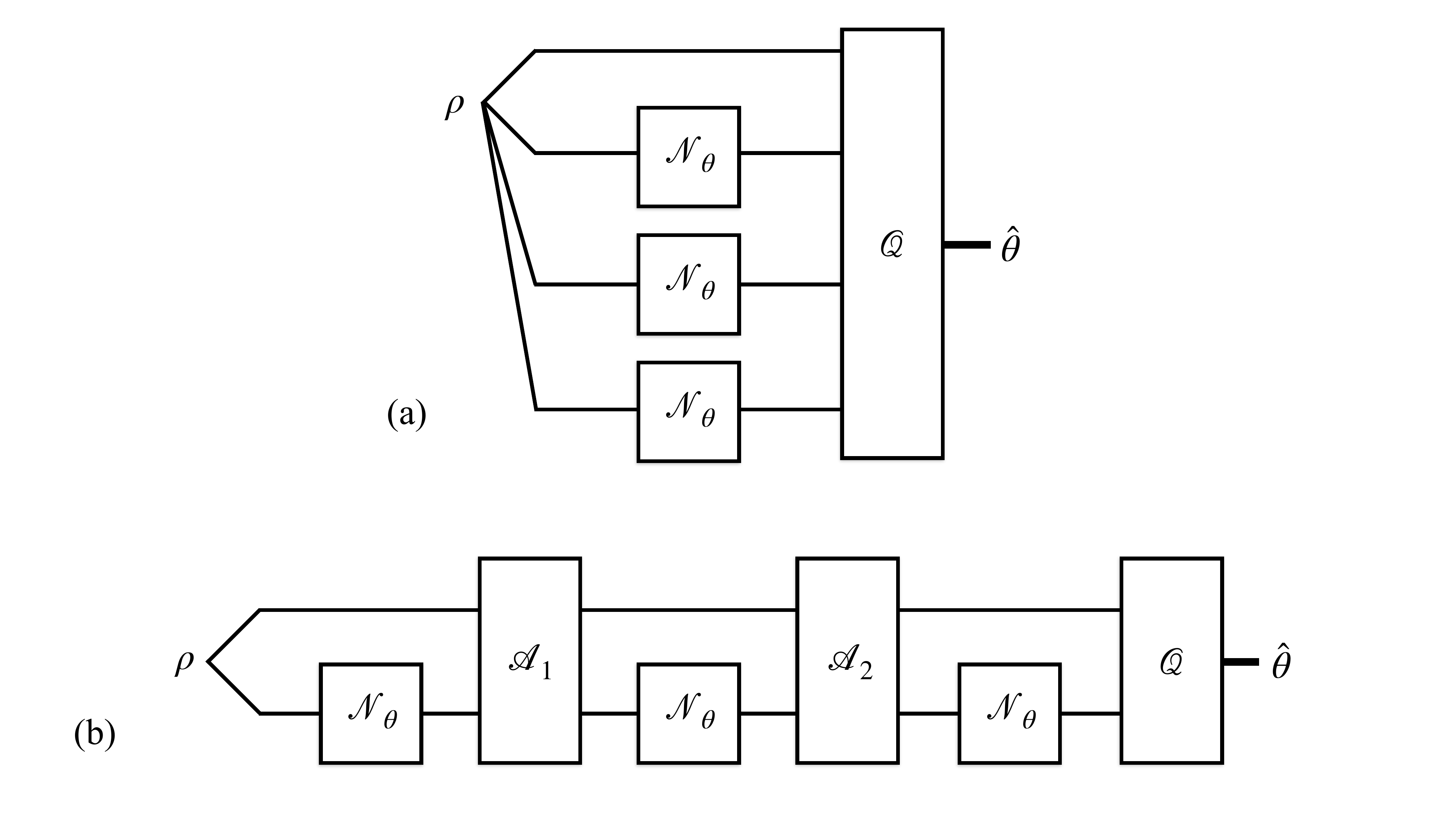}
\par\end{centering}
\caption{Depiction of (a) parallel and (b) adaptive strategies for quantum
channel discrimination and estimation. In both cases, the unknown channel is queried $n$ times, where $n=3$ in the figure. In discrimination, $\theta \in \{ 1,2\}$ is from a discrete set, and the goal is for the guess $\hat{\theta}$ to equal $\theta$. In estimation, $\theta \in \Theta \subseteq \mathbb{R}$ is from a continuous set, and the goal is for the guess $\hat{\theta}$ to approximate $\theta$ within some tolerance. 
In each protocol, $\mathcal{Q}$ is a measurement performed in order
to determine the value of $\theta$. In the case of channel discrimination,
$\mathcal{Q}$ is a measurement with two outcomes, while in the case
of channel estimation, $\mathcal{Q}$ is a measurement with outcomes
in $\Theta$. 
}
\label{fig:parallel-adaptive-ch-disc-est}
\end{figure}

\subsection{Summary of results}

In this paper, our main original contributions consist of lower bounds
on the error probabilities and query complexities of channel discrimination
and estimation. These lower bounds establish fundamental limitations
for these tasks, i.e., impossibility statements, and thus, they can
be understood as physical limitations placed on any protocol for channel
discrimination or estimation, akin to the second law of thermodynamics,
the uncertainty principle, or channel capacity theorems. More precisely,
the query complexities of quantum channel discrimination and estimation
are equal to the number of queries of an unknown channel that are
required to determine or estimate its identity, respectively, up to
a desired error probability. One can think of query complexity, roughly,
as the amount of time one has to wait in order to obtain a desired
error probability, and as such, this notion links information theory
and complexity theory in a nontrivial way, with query complexity often
being used to demonstrate optimality of various algorithms~\cite{Kawachi2019,Huang2022}.

Although our results here build on many prior results in the domains
of channel discrimination and estimation~\cite{Fujiwara2008,DemkowiczDobrzanski2012,Kolodynski2013,Demkowicz2014,Yuan2017,Zhou2021,Kurdzialek2023},
the lower bounds on query complexities are original to our paper.
We should note that the notion of query complexity of channel estimation
was introduced recently
in~\cite[Section~XI]{Meyer2025} (see also~\cite{Haah2023} for work on query complexity of unitary channel estimation); however,
the main emphasis of that paper was on establishing a framework for
quantum estimation theory in the finite-sample regime, while leaving
basic questions regarding channel estimation largely open. Thus, to
the best of our knowledge, our lower bounds on the query complexity
of channel estimation are the first of their kind in the literature.
Furthermore, a number of papers have appeared by now on query complexity
of quantum channel discrimination~\cite{Chiribella2013,Kawachi2019,Ito2021,Huang2022,Rossi2022,Li2025,Nuradha2025},
but none of them have adopted the particular approach taken here in
establishing lower bounds.

On the way to establishing these results, we often revisit well known
findings in the literature on channel discrimination and estimation.
The general approach that we adopt here is to be detailed and comprehensive,
even if this amounts to revisiting existing statements in the literature.
We believe that this approach has intrinsic value, leading to a consistent
and transparent framework and presentation style that should be widely
accessible to an audience familiar with the foundations of quantum
information theory. We should also note that the technical approach
we adopt here is  different from the main approach usually adopted in papers on channel
discrimination and estimation: indeed, we conduct all of our mathematical
developments for these tasks in terms of isometric extensions of channels.
In comparison, all prior papers closely related to ours use the Kraus
representation of quantum channels for their mathematical developments
\cite{Fujiwara2008,DemkowiczDobrzanski2012,Kolodynski2013,Demkowicz2014,Yuan2017,Zhou2021,Kurdzialek2023}.
We think an advantage of this approach is that various proofs become
easier to follow and have common points, relying on basic properties
of isometries and the spectral norm. Additionally, it seems more natural
to work with an isometric extension as the key object representing
a channel, rather than with Kraus operators, similar to how one manipulates
vectors and matrices in linear algebra, rather than manipulating the
individual components of the vectors and matrices.

To summarize our results succinctly, we first either recall or establish
upper bounds on the squared Bures distance %and 
in both the parallel
and adaptive settings of channel discrimination
%, respectively 
(see
Theorems~\ref{thm:parallel-bures-dist-channels} and~\ref{thm:seq-bures-distance-channels}),
and we do the same for the symmetric logarithmic derivative (SLD)
Fisher information for channel estimation (see Theorems~\ref{thm:upper-bound-SLD-Fisher-parallel}
and~\ref{thm:upper-bound-SLD-Fisher-adaptive}). After doing so, we
then establish lower bounds on the error probabilities of channel
discrimination and estimation in terms of the squared Bures distance
and SLD Fisher information (see Corollaries~\ref{cor:error-prob-bounds-ch-disc},
\ref{cor:minimax-error-ch-est}, and~\ref{cor:minimax-error-query-comp-SLD-Fisher}),
following a line of thinking adopted in~\cite{Cheng2025}. By applying
the definitions of query complexities for channel estimation and discrimination,
we then convert these lower bounds on error probabilities to lower
bounds on the query complexities for these tasks (see Corollaries
\ref{cor:query-comp-lower-bnd-ch-disc},~\ref{cor:query-comp-ch-est},
and~\ref{cor:minimax-error-query-comp-SLD-Fisher}). Finally, we provide
various numerical methods for efficiently computing lower bounds on
the query complexities of channel discrimination and estimation, in
both the parallel and adaptive access models.

Similar to the developments of~\cite{Yuan2017}, our paper also presents
a unified approach to establishing lower bounds on the query complexities
of channel discrimination and estimation, making use of a key fact
hinted at in~\cite[Section~XI]{Meyer2025} and made more concrete
in Section~\ref{subsec:Connecting-channel-est-disc}: lower bounds
on the query complexities of channel discrimination imply lower bounds
on the query complexities of channel estimation. Indeed, channel estimation
can be viewed as attempting to discriminate two nearby channels, and
so one would expect the ability to estimate should be limited by the
ability to discriminate.

\medskip{}

\textit{Note on independent work}---While finalizing our manuscript,
we noticed that the concurrent and independent paper~\cite{Sieniawski2025}
appeared. Similar to the theme of our paper and that of~\cite{Yuan2017},
Ref.~\cite{Sieniawski2025} presents a unified approach to quantum
channel discrimination and estimation. More specifically, the bound
presented in Eqs.~(18)--(20) therein are similar in spirit to our
bound from Corollary~\ref{cor:error-prob-bounds-ch-disc}.

\subsection{Paper organization}

The rest of our paper is organized as follows:
\begin{itemize}
\item Section~\ref{sec:Preliminaries} reviews basic concepts and establishes
some notation used in the rest of the paper.
\item Section~\ref{sec:Quantum-channel-disc-est} reviews the parallel and
adaptive settings of channel discrimination and estimation, defining
error probabilities and query complexities in these settings.
\item Section~\ref{sec:Fidelity-and-Bures} reviews basic aspects of fidelity
and Bures distance of states and channels. Although the fidelity and
Bures distance were defined some time ago now~\cite{Bures1969,Uhlmann1976}
and have been studied abundantly in the literature, the particular
way in which we present these quantities is useful for proving the
bounds on channel discrimination and estimation, as presented in later
sections.
\item Section~\ref{sec:SLD-Fisher-information-states-chs} reviews the SLD
Fisher information of states and channels. In particular, Theorems
\ref{thm:SLD-Fisher-states} and~\ref{thm:SLD-Fisher-channels} by
and large date back to~\cite[Theorems~1 and 4]{Fujiwara2008}. However,
some of the alternative expressions for SLD Fisher information of
states, as presented in Theorem~\ref{thm:SLD-Fisher-states}, appear
to be novel, and our proof of Theorem~\ref{thm:SLD-Fisher-states}
is distinct from the proof of~\cite[Theorem~1]{Fujiwara2008}. A similar
statement applies to Theorem~\ref{thm:SLD-Fisher-channels}.
\item Section~\ref{sec:Upper-bounds-Bures} presents upper bounds on the
squared Bures distance of channels in both the parallel and adaptive
settings of channel discrimination. The upper bound for the parallel
setting (Theorem~\ref{thm:parallel-bures-dist-channels}) was already
established in~\cite[Eq.~(23)]{Yuan2017}, but the upper bound for
the adaptive setting (Theorem~\ref{thm:seq-bures-distance-channels})
appears to be novel. Regardless, our proofs for these theorems are
distinct from that given for~\cite[Eq.~(23)]{Yuan2017} and are consistent
with the aforementioned theme: they are established in terms of isometric
extensions of quantum channels and rely on basic properties of isometries
and the spectral norm.
\item Section~\ref{sec:Upper-bounds-SLD-Fisher} reviews upper bounds on
the SLD Fisher information of channels in the parallel and adaptive
settings of channel estimation. The upper bound in the parallel setting
(Theorem~\ref{thm:upper-bound-SLD-Fisher-parallel}) was already established
in~\cite[Section~4.1]{Fujiwara2008}, and the upper bound in the adaptive
setting (Theorem~\ref{thm:upper-bound-SLD-Fisher-adaptive}) was established
somewhat recently in~\cite[Eq.~(11)]{Kurdzialek2023}, improving upon
an earlier bound from~\cite[Eq.~(9)]{Demkowicz2014}. Interestingly,
our proofs for Theorems~\ref{thm:upper-bound-SLD-Fisher-parallel}
and~\ref{thm:upper-bound-SLD-Fisher-adaptive} mirror those given
for Theorems~\ref{thm:parallel-bures-dist-channels} and~\ref{thm:seq-bures-distance-channels},
respectively, making use of basic properties of isometries and spectral
norms, again consistent with our overall theme. Furthermore, our proofs
for Theorems~\ref{thm:parallel-bures-dist-channels} and~\ref{thm:seq-bures-distance-channels}
can be understood as replacing derivatives with finite differences,
when compared to the proofs for Theorems~\ref{thm:upper-bound-SLD-Fisher-parallel}
and~\ref{thm:upper-bound-SLD-Fisher-adaptive}, consistent with our
other theme of having a unified approach to channel discrimination
and estimation.
\item Section~\ref{sec:Applications-to-query-comp} presents our main results:
lower bounds on the error probabilities and query complexities of
channel discrimination and estimation. Ultimately, with the framework
in place at this point, the results of this section follow somewhat
directly from definitions, reasoning similar to that in~\cite{Cheng2025,Nuradha2025},
and the bounds from Theorems~\ref{thm:parallel-bures-dist-channels},
\ref{thm:seq-bures-distance-channels},~\ref{thm:upper-bound-SLD-Fisher-parallel},
and~\ref{thm:upper-bound-SLD-Fisher-adaptive}.
\item Section~\ref{sec:Optimization-of-channel-disc-est} is the final technical
section of our paper, presenting methods for optimizing bounds on
the error probabilities of channel discrimination and estimation by
means of semi-definite optimization, or in some cases, semi-definite
optimization combined with a grid search. These methods are appealing
because they lead to globally optimal numerical solutions of the bounds.
\item We finally conclude in Section~\ref{sec:Conclusion} by providing
a summary of our results and suggestions for future research.
\end{itemize}

\section{Preliminaries}

\label{sec:Preliminaries}In this section, we establish notation and
concepts used throughout the rest of our paper. To begin with, we
use the following notation to denote various sets of interest that
arise throughout:
\begin{center}
\begin{tabular}{|c|c|c|}
\hline 
Symbol & Meaning & Definition\tabularnewline
\hline 
$\mathbb{L}$ & Set of linear operators & \tabularnewline
\hline 
$\mathbb{U}$ & Set of unitary operators & $\mathbb{U}\coloneqq\left\{ U\in\mathbb{L}:U^{\dag}U=UU^{\dag}=I\right\} $\tabularnewline
\hline 
$\mathbb{H}$ & Set of Hermitian operators & $\mathbb{H}\coloneqq\left\{ A\in\mathbb{L}:A=A^{\dag}\right\} $\tabularnewline
\hline 
$\mathbb{D}$ & Set of density operators & $\mathbb{D}\coloneqq\left\{ \rho\in\mathbb{H}:\rho\geq0,\ \Tr[\rho]=1\right\} $\tabularnewline
\hline 
$\mathbb{P}$ & Set of state vectors & $\mathbb{P}\coloneqq\left\{ |\psi\rangle:\left\Vert |\psi\rangle\right\Vert =1\right\} $\tabularnewline
\hline 
$\mathbb{B}$ & Set of contractions & $\mathbb{B}\coloneqq\left\{ W \in \mathbb{L} :\left\Vert W\right\Vert \leq 1\right\} $\tabularnewline
\hline 
\end{tabular}
\par\end{center}

\noindent In using this notation, we leave the underlying Hilbert
space and its dimension implicit, with the idea being that it should
be clear from the context. We also denote an operator $P$ as being
positive semi-definite by means of the notation $P\geq0$, and we
denote it as being positive definite by means of the notation $P>0$.
In defining the set $\mathbb{P}$ of (pure) state vectors, let us
note that $\left\Vert \cdot\right\Vert $ denotes the standard Euclidean
norm for vectors, and we have left the underlying Hilbert space implicit;
however, throughout the paper, we always have $|\psi\rangle\in\mathbb{C}^{d}$
for some $d\in\mathbb{N}$.

An isometry $V$ is an operator satisfying $V^{\dag}V=I$. For an
operator $A$, we denote the spectral norm by $\left\Vert A\right\Vert \coloneqq\sup_{|\psi\rangle\in\mathbb{P}}\left\Vert A|\psi\rangle\right\Vert $.
It obeys the following properties for operators $A$ and $B$ and
an isometry $V$:
\begin{align}
\left\Vert A\right\Vert  & =\left\Vert A^{\dag}\right\Vert =\left\Vert A^{\dag}A\right\Vert ^{\frac{1}{2}},\\
\left\Vert AB\right\Vert  & \leq\left\Vert A\right\Vert \left\Vert B\right\Vert ,\\
\left\Vert A\otimes B\right\Vert  & =\left\Vert A\right\Vert \left\Vert B\right\Vert ,\\
\left\Vert A+B\right\Vert  & \leq\left\Vert A\right\Vert +\left\Vert B\right\Vert ,\\
\left\Vert V\right\Vert  & =1.
\end{align}

We also make extensive use of the Schatten norms, defined for $p\geq1$
as
\begin{equation}
\left\Vert A\right\Vert _{p}\coloneqq\left(\Tr\!\left[\left|A\right|^{p}\right]\right)^{\frac{1}{p}},
\end{equation}
where $\left|A\right|\coloneqq\sqrt{A^{\dag}A}$. Special cases include
the trace norm $\left\Vert A\right\Vert _{1}=\Tr\!\left[\left|A\right|\right]$,
the Hilbert--Schmidt norm $\left\Vert A\right\Vert _{2}=\sqrt{\Tr\!\left[A^{\dag}A\right]}$,
and the spectral norm $\left\Vert A\right\Vert =\left\Vert A\right\Vert _{\infty}$.

When dealing with the SLD Fisher information, we take partial derivatives,
which we denote with the shorthand $\partial_{\theta}\equiv\frac{\partial}{\partial\theta}$.

Finally, we make extensive use of the maximally entangled vector
\begin{equation}
|\Gamma\rangle\coloneqq\sum_{i}|i\rangle\otimes|i\rangle,\label{eq:max-ent-vec-def}
\end{equation}
where $\left\{ |i\rangle\right\} _{i}$ is an orthonormal basis, and
its associated identities:
\begin{align}
\left(A\otimes I\right)|\Gamma\rangle & =\left(I\otimes A^{T}\right)|\Gamma\rangle,\label{eq:transpose-trick}\\
\langle\Gamma|\left(A\otimes I\right)|\Gamma\rangle & =\Tr\!\left[A\right],\label{eq:max-ent-reduce-to-trace}\\
\left\Vert A\right\Vert _{2} & =\left\Vert \left(A\otimes I\right)|\Gamma\rangle\right\Vert ,\label{eq:hilbert-schmidt-euclidean-norm}
\end{align}
which hold for a square operator $A$.

\section{Quantum channel discrimination and estimation}

\label{sec:Quantum-channel-disc-est}

\subsection{Quantum channel discrimination}

In the setting of symmetric binary quantum channel discrimination,
a channel $\mathcal{N}_{1}$ is selected with probability $p\in\left(0,1\right)$
and a channel $\mathcal{N}_{2}$ is selected with probability $q\equiv1-p$.
The discriminator is allowed to query the unknown channel $n\in\mathbb{N}$
times in an attempt to decide which channel was selected. There are
at least two ways for the discriminator to access the channel, called
the parallel and the adaptive access models, as depicted in Figure
\ref{fig:parallel-adaptive-ch-disc-est}.

\subsubsection{Parallel and adaptive settings of channel discrimination}

In the parallel setting of channel discrimination, the discriminator
prepares a state $\rho_{RA_{1}\cdots A_{n}}$, sends the systems $A_{1}\cdots A_{n}$
into $\mathcal{N}_{i}^{\otimes n}$, and performs a measurement on
the reference system $R$ and the channel output systems $B_{1}\cdots B_{n}$.
The minimum error probability is then given by~\cite[Theorem~5.9]{Khatri2024}
\begin{equation}
p_{e,\|}(p,\mathcal{N}_{1},q,\mathcal{N}_{2},n)\coloneqq\frac{1}{2}\left(1-\left\Vert p\mathcal{N}_{1}^{\otimes n}-q\mathcal{N}_{2}^{\otimes n}\right\Vert _{\diamond}\right),\label{eq:error-prob-parallel-ch-disc}
\end{equation}
where the diamond norm of a Hermiticity preserving map $\mathcal{M}$
can be defined (for our purposes here) as follows:
\begin{equation}
\left\Vert \mathcal{M}\right\Vert _{\diamond}\coloneqq\sup_{\psi_{RA}\in\mathbb{P}}\left\Vert \left(\id_{R}\otimes\mathcal{M}\right)\left(\psi_{RA}\right)\right\Vert _{1},
\end{equation}
with the reference system $R$ isomorphic to the channel input system
$A$. 

In the adaptive setting of channel discrimination, the discriminator
prepares an initial state $\rho_{RA}$, feeds system $A$ into the
first query of the unknown channel $\mathcal{N}_{i}$, then acts on
the reference system $R$ and the channel output system with a channel
$\mathcal{A}_{1}$, and iterates this process. Overall, a general
adaptive $n$-query protocol consists of a tuple $\left(\mathcal{A}_{i}\right)_{i=1}^{n-1}$
of channels. Applying the protocol to $n$ queries of the channel
$\mathcal{N}_{i}$ leads to the following channel $\mathcal{P}_{i}^{(n)}$
for $i\in\left\{ 1,2\right\} $:
\begin{equation}
\mathcal{P}_{i}^{(n)}\coloneqq\left(\id\otimes\mathcal{N}_{i}\right)\circ\mathcal{A}_{n-1}\circ\cdots\circ\mathcal{A}_{2}\circ\left(\id\otimes\mathcal{N}_{i}\right)\circ\mathcal{A}_{1}\circ\left(\id\otimes\mathcal{N}_{i}\right).\label{eq:adaptive-protocol-def}
\end{equation}
The minimum error probability in this case is given by
\begin{equation}
p_{e,\adaptivedisc}(p,\mathcal{N}_{1},q,\mathcal{N}_{2},n)\coloneqq\inf_{\left(\mathcal{A}_{i}\right)_{i=1}^{n-1}}\frac{1}{2}\left(1-\left\Vert p\mathcal{P}_{1}^{(n)}-q\mathcal{P}_{2}^{(n)}\right\Vert _{\diamond}\right).\label{eq:error-prob-adaptive-ch-disc}
\end{equation}

Adopting the shorthand
\begin{align}
p_{e,\|} & \equiv p_{e,\|}(p,\mathcal{N}_{1},q,\mathcal{N}_{2},n),\\
p_{e,\adaptivedisc} & \equiv p_{e,\adaptivedisc}(p,\mathcal{N}_{1},q,\mathcal{N}_{2},n),
\end{align}
let us note that the following inequalities hold:
\begin{equation}
p_{e,\adaptivedisc}\leq p_{e,\|}\leq\min\!\left\{ p,q\right\} \leq\frac{1}{2}.\label{eq:basic-ineq-ch-disc}
\end{equation}
The first inequality follows because a parallel channel discrimination
strategy is a special case of an adaptive channel discrimination strategy.
The second inequality is a consequence of~\cite[Theorem~1]{Audenaert2007},
as remarked in~\cite[Eqs.~(D1)--(D3)]{Cheng2025}, and the fact that
$\min\!\left\{ p,q\right\} =\min_{s\in\left[0,1\right]}p^{s}q^{1-s}$. Alternatively, one can understand this inequality as resulting from discarding the channel outputs and simply guessing the channel that has the higher prior probability $p$ or $q$. Discarding the channel in either~\eqref{eq:error-prob-parallel-ch-disc} or~\eqref{eq:error-prob-adaptive-ch-disc} leads to the following expression for the minimum error probability $\frac{1}{2}\left(1-\left\vert p-q\right\vert\right) $, which is equal to $\min\{p,q\}$. 
As such, the following inequalities hold for $p\in\left(0,1\right)$:
\begin{equation}
0\leq\frac{p_{e,\adaptivedisc}\left(1-p_{e,\adaptivedisc}\right)}{pq}\leq\frac{p_{e,\|}\left(1-p_{e,\|}\right)}{pq}\leq1,\label{eq:alt-channel-disc-metric}
\end{equation}
because the function $x\mapsto x\left(1-x\right)$ is monotone increasing
on the interval $\left[0,\frac{1}{2}\right]$. We can also understand
the quantities in~\eqref{eq:alt-channel-disc-metric} as being alternative
measures of the performance of a channel discrimination protocol,
being equal to zero if the channels are perfectly distinguishable
and equal to their maximum value of one if the channels are not distinguishable
at all, so that the best strategy in this latter case is simply to
guess the channel with the higher prior probability $p$ or $q$. 

\subsubsection{Query complexities of channel discrimination}

We can ask about the minimum number of queries needed to achieve an
error probability not exceeding a desired threshold. The resulting
quantity is known as query complexity and is of chief interest in
our paper.
\begin{defn}
\label{def:query-complexity-ch-disc-def}For $\varepsilon\in\left[0,1\right]$,
we define two different notions of query complexity, based on the
parallel and adaptive settings of channel discrimination:
\begin{align}
n_{\|}^{\star}(p,\mathcal{N}_{1},q,\mathcal{N}_{2},\varepsilon) & \coloneqq\inf\left\{ n\in\mathbb{N}:p_{e,\|}(p,\mathcal{N}_{1},q,\mathcal{N}_{2},n)\leq\varepsilon\right\} ,\label{eq:query-complexity-parallel-def}\\
n_{\adaptivedisc}^{\star}(p,\mathcal{N}_{1},q,\mathcal{N}_{2},\varepsilon) & \coloneqq\inf\left\{ n\in\mathbb{N}:p_{e,\adaptivedisc}(p,\mathcal{N}_{1},q,\mathcal{N}_{2},n)\leq\varepsilon\right\} .\label{eq:query-complexity-adaptive-def}
\end{align}
The quantities above are understood to take the value $+\infty$ if the set on the right-hand side is empty.
\end{defn}

Intuitively, query complexities give a sense of how long one has to
wait in order to achieve a desired error probability in the task of
channel discrimination. However, this characterization is somewhat
rough, because the definition of query complexity assumes the ability
to prepare arbitrary states and perform arbitrary measurements in
the parallel setting and it assumes the same ability, in addition
to the ability to perform arbitrary channels in the tuple $\left(\mathcal{A}_{i}\right)_{i=1}^{n-1}$,
in the adaptive setting, thus ignoring the computational complexity
of implementing these operations.

Proposition 7 of~\cite{Nuradha2025} identifies conditions under which
the query complexity of channel discrimination is trivial, i.e., either
equal to one, in which case a desired threshold constraint on the
error probability can be met with just a single query, or $+\infty$,
in which case it is impossible to distinguish the channels by using
a finite number of queries. In this context, see also~\cite[Remark~2]{Cheng2025}.
Here we recall conditions under which the query complexity is equal
to one, while noting from~\eqref{eq:basic-ineq-ch-disc} that two
of the conditions from~\cite{Nuradha2025} and~\cite[Remark~2]{Cheng2025}
can actually be merged into a single, simpler condition:
\begin{rem}[Trivial cases]
\label{rem:trivial-cases-ch-disc-query-comp}If $\varepsilon\geq\min\!\left\{ p,q\right\} $
or if there exists a pure state $\psi_{RA}$ such that
\begin{equation}
\supp\!\left(\left(\id\otimes\mathcal{N}_{1}\right)\left(\psi_{RA}\right)\right)\cap\supp\!\left(\left(\id\otimes\mathcal{N}_{2}\right)\left(\psi_{RA}\right)\right)=\emptyset,
\end{equation}
then
\begin{equation}
n_{\adaptivedisc}^{\star}(p,\mathcal{N}_{1},q,\mathcal{N}_{2},\varepsilon)=n_{\|}^{\star}(p,\mathcal{N}_{1},q,\mathcal{N}_{2},\varepsilon)=1.
\end{equation}
\end{rem}

One of the main goals of our paper is to establish lower bounds on
the query complexities in~\eqref{eq:query-complexity-parallel-def}
and~\eqref{eq:query-complexity-adaptive-def}.

\subsection{Quantum channel estimation}

In the setting of quantum channel estimation, the goal is to estimate
a parameter value $\theta\in\Theta$ by querying an unknown channel
$\mathcal{N}_{\theta}$ selected from a parameterized family $\left(\mathcal{N}_{\theta}\right)_{\theta\in\Theta}$
of channels, where $\Theta\in\mathbb{R}$ is an open set. The estimator
is allowed to query the unknown channel $n$ times in order to estimate
the value of $\theta$. As with channel discrimination, there are
at least two different access models, which include the parallel and
adaptive access models, as depicted in Figure~\ref{fig:parallel-adaptive-ch-disc-est}.
In what follows, we detail these settings and recall definitions of
the minimax error probabilities and query complexities of channel
estimation, following the framework of~\cite[Section~XI]{Meyer2025}.
Another main goal of our paper is to establish lower bounds on the
query complexities of channel estimation, in both the parallel and
adaptive settings.

\subsubsection{Parallel setting}

\label{subsec:Parallel-setting-estimation}Let us begin by considering
the parallel setting of channel estimation. An $n$-query parallel
estimation protocol begins with the estimator preparing a state $\rho_{RA_{1}\cdots A_{n}}$,
feeding the systems $A_{1}\cdots A_{n}$ of this state through the
unknown channel $\mathcal{N}_{\theta}^{\otimes n}$, and performing
a measurement 
$\smash{\mathcal{Q}^{\left(n\right)}\coloneqq(Q_{\hat{\theta}}^{\left(n\right)})_{\hat{\theta}\in\Theta}}$ 
on the reference system $R$ and the channel output systems $B_{1}\cdots B_{n}$. We assume that $\mathcal{Q}^{\left(n\right)}$ forms a positive operator-valued measure (POVM) on $\Theta$ with respect to the Lebesgue measure.
The probability of outputting the estimate $\hat{\theta}$, given
that $\theta$ is the true parameter value, is as follows:
\begin{equation}
\Tr\!\left[Q_{\hat{\theta}}^{\left(n\right)}\left(\id_{R}\otimes\mathcal{N}_{\theta}^{\otimes n}\right)\left(\rho_{RA_{1}\cdots A_{n}}\right)\right],
\end{equation}
and the goal is for $\hat{\theta}$ to be as close to $\theta$ as
possible.

Given a tolerance $\delta>0$, the success probability, for identifying
an unknown parameter $\theta\in\Theta$ in the family $\mathcal{F}\equiv\left(\mathcal{N}_{\theta}\right)_{\theta\in\Theta}$
by means of a parallel estimation protocol $\mathcal{P}^{\left(n\right)}\equiv\left(\rho_{RA_{1}\cdots A_{n}},\mathcal{Q}^{\left(n\right)}\right)$,
is defined as
\begin{equation}
p_{s,\|}(\delta,\mathcal{F},\mathcal{P}^{\left(n\right)},\theta)\coloneqq\int d\hat{\theta}\:\Tr\!\left[Q_{\hat{\theta}}^{\left(n\right)}\left(\id_{R}\otimes\mathcal{N}_{\theta}^{\otimes n}\right)\left(\rho_{RA_{1}\cdots A_{n}}\right)\right]w_{\delta}(\hat{\theta}-\theta),
\end{equation}
where $w_{\delta}(\theta')$ is a window function, defined as
\begin{equation}
w_{\delta}(\theta')\coloneqq\begin{cases}
1 & :\left|\theta'\right|\leq\delta\\
0 & :\left|\theta'\right|>\delta
\end{cases}.\label{eq:window-function-def}
\end{equation}
Intuitively, the success probability $p_{s,\|}(\delta,\mathcal{F},\mathcal{P}^{\left(n\right)},\theta)$
quantifies the probability that the measurement $\mathcal{Q}$ outputs
an estimate $\hat{\theta}$ that is $\delta$-close to $\theta$.
The error probability is then defined as
\begin{align}
p_{e,\|}(\delta,\mathcal{F},\mathcal{P}^{\left(n\right)},\theta) & \coloneqq1-p_{s,\|}(\delta,\mathcal{F},\mathcal{P}^{\left(n\right)},\theta)\\
 & =\int_{\hat{\theta}:\left|\hat{\theta}-\theta\right|>\delta}d\hat{\theta}\:\Tr\!\left[Q_{\hat{\theta}}^{\left(n\right)}\left(\id_{R}\otimes\mathcal{N}_{\theta}^{\otimes n}\right)\left(\rho_{RA_{1}\cdots A_{n}}\right)\right].
\end{align}

The worst-case error probability, for identifying an unknown parameter
$\theta\in\Theta$ in the family $\mathcal{F}$ by using a parallel
protocol $\mathcal{P}^{\left(n\right)}$, is defined as
\begin{align}
p_{e,\|}(\delta,\mathcal{F},\mathcal{P}^{\left(n\right)}) & \coloneqq\sup_{\theta\in\Theta}p_{e,\|}(\delta,\mathcal{F},\mathcal{P}^{\left(n\right)},\theta).
\end{align}
Minimizing the worst-case error probability over every possible parallel
protocol $\mathcal{P}^{\left(n\right)}$ leads to the following key
figure of merit:
\begin{defn}[Minimax error probability for parallel channel estimation]
\label{def:minimax-parallel-est}Given a tolerance $\delta>0$, a
parameterized family $\mathcal{F}\equiv\left(\mathcal{N}_{\theta}\right)_{\theta\in\Theta}$
of quantum channels, and $n\in\mathbb{N}$, the minimax error probability
is defined as
\begin{align}
p_{e,\|}(\delta,\mathcal{F},n) & \coloneqq\inf_{\mathcal{P}^{\left(n\right)}}p_{e,\|}(\delta,\mathcal{F},\mathcal{P}^{\left(n\right)})\\
 & =\inf_{\substack{\rho_{RA_{1}\cdots A_{n}},\\
\left(Q_{\hat{\theta}}^{\left(n\right)}\right)_{\hat{\theta}\in\Theta}
}
}\sup_{\theta\in\Theta}\left\{ \int_{\hat{\theta}:\left|\hat{\theta}-\theta\right|>\delta}d\hat{\theta}\:\Tr\!\left[Q_{\hat{\theta}}^{\left(n\right)}\left(\id_{R}\otimes\mathcal{N}_{\theta}^{\otimes n}\right)\left(\rho_{RA_{1}\cdots A_{n}}\right)\right]\right\} .
\end{align}
\end{defn}

\begin{defn}[Query complexity of parallel channel estimation]
\label{def:query-comp-parallel-est}For all $\varepsilon\in\left[0,1\right]$,
$\delta>0$, and every parameterized family $\mathcal{F}\equiv\left(\mathcal{N}_{\theta}\right)_{\theta\in\Theta}$
of quantum channels, the minimax query complexity of parallel channel
estimation is defined as
\begin{equation}
n_{\|}^{\star}(\varepsilon,\delta,\mathcal{F})\coloneqq\inf\left\{ n\in\mathbb{N}:p_{e,\|}(\delta,\mathcal{F},n)\leq\varepsilon\right\} .
\end{equation}
The quantity above is understood to equal $+\infty$ if the set on the right-hand side is empty.
\end{defn}

As in the case of channel discrimination, the query complexity of
channel estimation intuitively captures how long one has to wait,
or how many queries one has to make, in order to reach a desired error
probability in the task of channel estimation.

\subsubsection{Adaptive setting}

\label{subsec:Adaptive-setting-estimation}We now establish similar
definitions for the adaptive setting of channel estimation. An $n$-query
adaptive channel estimation protocol begins with the estimator preparing
an initial state $\rho_{RA}$, feeding system $A$ into the first
query of the unknown channel $\mathcal{N}_{\theta}$, then acting
on the reference system $R$ and the channel output system with a
channel $\mathcal{A}_{1}$, and iterating this process. Overall, a
general adaptive $n$-query protocol consists of a tuple $\smash{\left(\mathcal{A}_{i}\right)_{i=1}^{n-1}}$
of channels. Applying the protocol to $n$ queries of the unknown
channel $\mathcal{N}_{\theta}$ leads to the following channel $\mathcal{A}_{\theta}^{(n)}$
for $\theta\in\Theta$:
\begin{equation}
\mathcal{A}_{\theta}^{(n)}\coloneqq\left(\id\otimes\mathcal{N}_{\theta}\right)\circ\mathcal{A}_{n-1}\circ\cdots\circ\mathcal{A}_{2}\circ\left(\id\otimes\mathcal{N}_{\theta}\right)\circ\mathcal{A}_{1}\circ\left(\id\otimes\mathcal{N}_{\theta}\right).\label{eq:adaptive-protocol-def-1}
\end{equation}
At the end, a measurement $\smash{\mathcal{Q}^{\left(n\right)}\coloneqq(Q_{\hat{\theta}}^{\left(n\right)})_{\hat{\theta}\in\Theta}}$ 
is performed on the final state. The probability of outputting the
estimate $\hat{\theta}$, given that $\theta$ is the true parameter
value, is as follows:
\begin{equation}
\Tr\!\left[Q_{\hat{\theta}}^{\left(n\right)}\mathcal{A}_{\theta}^{(n)}\!\left(\rho_{RA}\right)\right],
\end{equation}
and the goal is for $\hat{\theta}$ to be as close to $\theta$ as
possible.

Given a tolerance $\delta>0$, the error probability, for identifying
an unknown parameter $\theta\in\Theta$ in the family $\mathcal{F}\equiv\left(\mathcal{N}_{\theta}\right)_{\theta\in\Theta}$
by means of an adaptive estimation protocol $\mathcal{E}^{\left(n\right)}\equiv\left(\rho_{RA},\left(\mathcal{A}_{i}\right)_{i=1}^{n-1},\mathcal{Q}^{\left(n\right)}\right)$,
is defined as
\begin{align}
p_{e,\adaptivedisc}(\delta,\mathcal{F},\mathcal{E}^{\left(n\right)},\theta) & \coloneqq1-\int d\hat{\theta}\:\Tr\!\left[Q_{\hat{\theta}}^{\left(n\right)}\mathcal{A}_{\theta}^{(n)}\!\left(\rho_{RA}\right)\right]w_{\delta}(\hat{\theta}-\theta)\\
 & =\int_{\hat{\theta}:\left|\hat{\theta}-\theta\right|>\delta}d\hat{\theta}\:\Tr\!\left[Q_{\hat{\theta}}^{\left(n\right)}\mathcal{A}_{\theta}^{(n)}\!\left(\rho_{RA}\right)\right],
\end{align}
where $w_{\delta}(\theta')$ is the window function defined in~\eqref{eq:window-function-def}. 

The worst-case error probability, for identifying an unknown parameter
$\theta\in\Theta$ in the family $\mathcal{F}$ by using an adaptive
protocol $\mathcal{E}^{\left(n\right)}$, is defined as
\begin{align}
p_{e,\adaptivedisc}(\delta,\mathcal{F},\mathcal{E}^{\left(n\right)}) & \coloneqq\sup_{\theta\in\Theta}p_{e,\adaptivedisc}(\delta,\mathcal{F},\mathcal{E}^{\left(n\right)},\theta).
\end{align}
Minimizing the worst-case error probability over every possible adaptive
protocol $\mathcal{E}^{\left(n\right)}$ leads to the following key
figure of merit:
\begin{defn}[Minimax error probability for adaptive channel estimation]
\label{def:minimax-adaptive-est}Given a tolerance $\delta>0$, a
parameterized family $\mathcal{F}\equiv\left(\mathcal{N}_{\theta}\right)_{\theta\in\Theta}$
of quantum channels, and $n\in\mathbb{N}$, the minimax error probability
is defined as
\begin{align}
p_{e,\adaptivedisc}(\delta,\mathcal{F},n) & \coloneqq\inf_{\mathcal{E}^{\left(n\right)}}p_{e,\adaptivedisc}(\delta,\mathcal{F},\mathcal{E}^{\left(n\right)})\\
 & =\inf_{\substack{\rho_{RA},\left(\mathcal{A}_{i}\right)_{i=1}^{n-1},\\
\left(Q_{\hat{\theta}}^{\left(n\right)}\right)_{\hat{\theta}\in\Theta}
}
}\sup_{\theta\in\Theta}\left\{ \int_{\hat{\theta}:\left|\hat{\theta}-\theta\right|>\delta}d\hat{\theta}\:\Tr\!\left[Q_{\hat{\theta}}^{\left(n\right)}\mathcal{A}_{\theta}^{(n)}\!\left(\rho_{RA}\right)\right]\right\} .
\end{align}
\end{defn}

\begin{defn}[Query complexity of adaptive channel estimation]
\label{def:query-comp-adaptive-est}For all $\varepsilon\in\left[0,1\right]$,
$\delta>0$, and every parameterized family $\mathcal{F}\equiv\left(\mathcal{N}_{\theta}\right)_{\theta\in\Theta}$
of quantum channels, the minimax query complexity of adaptive channel
estimation is defined as
\begin{equation}
n_{\adaptivedisc}^{\star}(\varepsilon,\delta,\mathcal{F})\coloneqq\inf\left\{ n\in\mathbb{N}:p_{e,\adaptivedisc}(\delta,\mathcal{F},n)\leq\varepsilon\right\} .
\end{equation}
The quantity above is understood to equal $+\infty$ if the set on the right-hand side is empty.
\end{defn}

\subsection{Connecting channel estimation and discrimination}

\label{subsec:Connecting-channel-est-disc}Ref.~\cite[Corollary~7]{Meyer2025}
established a non-trivial link between channel estimation and discrimination,
particularly relevant for the minimax error formulation discussed
in Sections~\ref{subsec:Parallel-setting-estimation} and~\ref{subsec:Adaptive-setting-estimation}.
One of the implications of~\cite[Corollary~7]{Meyer2025}, along with
the same reasoning that leads to~\cite[Corollary~1]{Meyer2025}, is
as follows:
\begin{lem}
\label{lem:est-to-disc-minimax-lower-bnd}
    For all $\delta > 0$ and $n\in \mathbb{N}$ and for every channel family $\mathcal{F}$, the following bounds hold:
\begin{align}
p_{e,\|}(\delta,\mathcal{F},n) & \geq\sup_{\substack{\theta,\theta'\in\Theta:\\
\left|\theta-\theta'\right|>2\delta
}
}p_{e,\|}\!\left(\sfrac12,\mathcal{N}_{\theta},\sfrac12,\mathcal{N}_{\theta'},n\right),\label{eq:parallel-est-minimax-lower-bnd}\\
p_{e,\adaptivedisc}(\delta,\mathcal{F},n) & \geq\sup_{\substack{\theta,\theta'\in\Theta:\\
\left|\theta-\theta'\right|>2\delta
}
}p_{e,\adaptivedisc}\!\left(\sfrac12,\mathcal{N}_{\theta},\sfrac12,\mathcal{N}_{\theta'},n\right).\label{eq:adaptive-est-minimax-lower-bnd}
\end{align}
\end{lem}

\begin{proof}
    We provide a simple proof of~\eqref{eq:adaptive-est-minimax-lower-bnd} for completeness, which follows Le Cam's two-point method \cite{LeCam1973,Yu1997}. The proof of~\eqref{eq:parallel-est-minimax-lower-bnd} is analogous.
     Define
the measurement operator $Q_{\left[\theta-\delta,\theta+\delta\right]}^{\left(n\right)}$
as follows:
\begin{equation}
Q_{\left[\theta-\delta,\theta+\delta\right]}^{\left(n\right)}\coloneqq\int_{\left|\hat{\theta}-\theta\right|\leq\delta}d\hat{\theta}\:Q_{\hat{\theta}}^{\left(n\right)}.
\end{equation}
Fix $\theta_{1},\theta_{2}\in\Theta$
such that $\left|\theta_{1}-\theta_{2}\right|>2\delta$.
Then
\begin{align}
 & p_{e,\adaptivedisc}(\delta,\mathcal{F},n)\nonumber \\
 & =\inf_{\mathcal{E}^{\left(n\right)}}\sup_{\theta\in\Theta}\left\{ \int_{\hat{\theta}:\left|\hat{\theta}-\theta\right|>\delta}d\hat{\theta}\:\Tr\!\left[Q_{\hat{\theta}}^{\left(n\right)}\mathcal{A}_{\theta}^{(n)}\!\left(\rho_{RA}\right)\right]\right\} \\
 & =\inf_{\mathcal{E}^{\left(n\right)}}\sup_{\theta\in\Theta}\left\{ \Tr\!\left[\left(I-Q_{\left[\theta-\delta,\theta+\delta\right]}^{\left(n\right)}\right)\mathcal{A}_{\theta}^{(n)}\!\left(\rho_{RA}\right)\right]\right\} \\
 & \geq\inf_{\mathcal{E}^{\left(n\right)}}\sup_{\substack{\theta_{1},\theta_{2}\in \Theta:\\|\theta_1-\theta_2|>2\delta }}\left\{ \Tr\!\left[\left(I-Q_{\left[\theta_{1}-\delta,\theta_{1}+\delta\right]}^{\left(n\right)}\right)\mathcal{A}_{\theta_{1}}^{(n)}\!\left(\rho_{RA}\right)\right],\Tr\!\left[\left(I-Q_{\left[\theta_{2}-\delta,\theta_{2}+\delta\right]}^{\left(n\right)}\right)\mathcal{A}_{\theta_{2}}^{(n)}\!\left(\rho_{RA}\right)\right]\right\} \\
 & \geq\inf_{\mathcal{E}^{\left(n\right)}}\sup_{\substack{\theta_{1},\theta_{2}\in \Theta:\\|\theta_1-\theta_2|>2\delta }}\left\{ \Tr\!\left[\left(I-Q_{\left[\theta_{1}-\delta,\theta_{1}+\delta\right]}^{\left(n\right)}\right)\mathcal{A}_{\theta_{1}}^{(n)}\!\left(\rho_{RA}\right)\right],\Tr\!\left[Q_{\left[\theta_{1}-\delta,\theta_{1}+\delta\right]}^{\left(n\right)}\mathcal{A}_{\theta_{2}}^{(n)}\!\left(\rho_{RA}\right)\right]\right\} \\
 & \geq\inf_{\mathcal{E}^{\left(n\right)}}\left\{ \frac{1}{2}\Tr\!\left[\left(I-Q_{\left[\theta_{1}-\delta,\theta_{1}+\delta\right]}^{\left(n\right)}\right)\mathcal{A}_{\theta_{1}}^{(n)}\!\left(\rho_{RA}\right)\right]+\frac{1}{2}\Tr\!\left[Q_{\left[\theta_{1}-\delta,\theta_{1}+\delta\right]}^{\left(n\right)}\mathcal{A}_{\theta_{2}}^{(n)}\!\left(\rho_{RA}\right)\right]\right\} \\
 & =p_{e,\adaptivedisc}\!\left(\sfrac12,\mathcal{N}_{\theta_{1}},\sfrac12,\mathcal{N}_{\theta_{2}},n\right).
\end{align}
The second inequality follows because $I-Q_{\left[\theta_{2}-\delta,\theta_{2}+\delta\right]}^{\left(n\right)}\geq Q_{\left[\theta_{1}-\delta,\theta_{1}+\delta\right]}^{\left(n\right)}$, given that $\left|\theta_{1}-\theta_{2}\right|>2\delta$.
Indeed, since $|\theta_1-\theta_2|>2\delta$, the intervals
$
[\theta_1-\delta,\theta_1+\delta]$ and 
$
[\theta_2-\delta,\theta_2+\delta]
$
are disjoint. Because $\big(Q_{\hat\theta}^{(n)}\big)_{\hat\theta\in\Theta}$
forms a POVM, it is countably additive on disjoint measurable sets.
Hence,
\begin{equation}
    Q_{[\theta_1-\delta,\theta_1+\delta]}^{(n)}
+
Q_{[\theta_2-\delta,\theta_2+\delta]}^{(n)}
=
Q_{[\theta_1-\delta,\theta_1+\delta]\cup
[\theta_2-\delta,\theta_2+\delta]}^{(n)}
\le I,
\end{equation}
where $I$ denotes the identity operator on the final output space.
Therefore,
\begin{equation}
    Q_{[\theta_1-\delta,\theta_1+\delta]}^{(n)}
\le
I-
Q_{[\theta_2-\delta,\theta_2+\delta]}^{(n)},
\end{equation}
which proves the claimed operator inequality. The last equality follows because the  optimization ranges over
all adaptive discrimination strategies between
$\mathcal{N}_{\theta_1}$ and $\mathcal{N}_{\theta_2}$ with equal
prior probabilities, given that  every adaptive estimation protocol allows for an arbitrary
final POVM on the output system. Hence, the expression coincides with the
minimum Bayesian error probability for adaptive discrimination
between these two channels.
Since the inequality $p_{e,\adaptivedisc}(\delta,\mathcal{F},n)\geq p_{e,\adaptivedisc}\!\left(\sfrac12,\mathcal{N}_{\theta_{1}},\sfrac12,\mathcal{N}_{\theta_{2}},n\right)$ holds for all $\theta_{1},\theta_{2}\in\Theta$
such that $\left|\theta_{1}-\theta_{2}\right|>2\delta$, we conclude
\eqref{eq:adaptive-est-minimax-lower-bnd}.
\end{proof}

Lemma~\ref{lem:est-to-disc-minimax-lower-bnd} states that the minimax error probabilities for channel estimation, with
tolerance $\delta$, are not smaller than the channel discrimination
error probabilities for two channels with parameter values that are
more than $2\delta$-far from each other. Stated informally, channel
estimation is just as difficult as channel discrimination of nearby
channels. As such, limitations on channel discrimination provide limitations
on channel estimation.

By combining~\eqref{eq:parallel-est-minimax-lower-bnd}--\eqref{eq:adaptive-est-minimax-lower-bnd}
with Definitions~\ref{def:query-complexity-ch-disc-def},~\ref{def:minimax-parallel-est},
and~\ref{def:query-comp-parallel-est}, we arrive at the following
lower bounds on the query complexities of channel estimation in terms
of query complexities of channel discrimination, for all $\varepsilon\in\left[0,1\right]$:
\begin{align}
n_{\|}^{\star}(\varepsilon,\delta,\mathcal{F}) & \geq\sup_{\substack{
\theta,\theta'\in\Theta:\\\left|\theta-\theta'\right|>2\delta
}
}n_{\|}^{\star}\!\left(\sfrac12,\mathcal{N}_{\theta},\sfrac12,\mathcal{N}_{\theta'},\varepsilon\right),\label{eq:parallel-est-query-comp-lower-bnd}\\
n_{\adaptivedisc}^{\star}(\varepsilon,\delta,\mathcal{F}) & \geq\sup_{\substack{
\theta,\theta'\in\Theta:\\\left|\theta-\theta'\right|>2\delta
}
}n_{\adaptivedisc}^{\star}\!\left(\sfrac12,\mathcal{N}_{\theta},\sfrac12,\mathcal{N}_{\theta'},\varepsilon\right).\label{eq:adaptive-est-query-comp-lower-bnd}
\end{align}

Section~\ref{subsec:Lower-bounds-query-compl-ch-estimation} delves
into more detail of our lower bounds on the minimax error probabilities
and query complexities of channel estimation. 

\section{Fidelity and Bures distance of states and channels}\label{sec:Fidelity-and-Bures}

\subsection{Fidelity and Bures distance of states}

In this section, we develop some aspects of fidelity and Bures distance
of quantum states. Although these quantities were introduced some
time ago~\cite{Bures1969,Uhlmann1976} and have been extensively studied
in the literature, we develop some subtle aspects of these quantities
that are typically not considered, yet are necessary in order for
us to obtain the tightest query complexity bounds possible in the
channel discrimination application, while using available techniques.

To summarize the subtle aspects, the fidelity has a well known characterization
in terms of Uhlmann's theorem~\cite{Uhlmann1976}, in which it can
be written as the squared overlap of two canonical purifications with
an optimization over all unitaries acting on the purifying systems
of the purifications. However, we can alternatively optimize this
same quantity instead over all contractions acting on the purifying
system, which is useful in the context of minimax theorems, given
that the set of contractions is convex, while the set of unitaries
is not. Thus, it is useful to have a form of the fidelity that we
can optimize over all contractions. Furthermore, by introducing an
extra qubit for the reference system, it is possible to dilate a contraction
to a unitary acting on a larger system (called Halmos dilation~\cite{Halmos1950}),
which implies that the optimization over contractions acting on the
original reference systems is actually equivalent to optimizing over
unitaries acting on a larger reference system and vice versa. 

After developing these points for the fidelity, we then briefly state
how they apply to the Bures distance of states, given the strong connections
between Bures distance and fidelity.

\subsubsection{Fidelity of states}

Recall that the fidelity of states $\rho$ and $\sigma$ is defined
as follows~\cite{Uhlmann1976}:
\begin{equation}
F(\rho,\sigma)\coloneqq\left\Vert \sqrt{\rho}\sqrt{\sigma}\right\Vert _{1}^{2}.
\end{equation}
Let us define canonical purifications of $\rho$ and $\sigma$ as
follows:
\begin{align}
|\psi^{\rho}\rangle & \coloneqq\left(I\otimes\sqrt{\rho}\right)|\Gamma\rangle,\label{eq:canonical-pure-rho}\\
|\psi^{\sigma}\rangle & \coloneqq\left(I\otimes\sqrt{\sigma}\right)|\Gamma\rangle,\label{eq:canonical-pure-sigma}
\end{align}
where $|\Gamma\rangle$ is defined in~\eqref{eq:max-ent-vec-def}.

In Proposition~\ref{prop:fidelity-different-forms} below, we state
three different formulas for the root fidelity.
\begin{prop}
\label{prop:fidelity-different-forms}For $d$-dimensional states
$\rho$ and $\sigma$ with canonical purifications in~\eqref{eq:canonical-pure-rho}
and~\eqref{eq:canonical-pure-sigma}, respectively, the following
equalities hold:
\begin{align}
\sqrt{F}(\rho,\sigma) & =\sup_{W\in\mathbb{B}}\Re\!\left[\langle\psi^{\rho}|\left(W\otimes I\right)|\psi^{\sigma}\rangle\right]\label{eq:fid-states-1}\\
 & =\sup_{U\in\mathbb{U}}\Re\!\left[\langle\psi^{\rho}|\left(U\otimes I\right)|\psi^{\sigma}\rangle\right]\label{eq:fid-states-2}\\
 & =\sup_{U_{W}\in\mathbb{U}}\Re\!\left[\langle0|\langle\psi^{\rho}|\left(U_{W}\otimes I\right)|0\rangle|\psi^{\sigma}\rangle\right].\label{eq:halmos-dilation-fidelity}
\end{align}
In~\eqref{eq:halmos-dilation-fidelity}, $|0\rangle|\psi^{\sigma}\rangle,|0\rangle|\psi^{\rho}\rangle\in\mathbb{C}^{2}\otimes\mathbb{C}^{d}\otimes\mathbb{C}^{d}$
and $U_{W}$ acts nontrivially on the first two systems of $\mathbb{C}^{2}\otimes\mathbb{C}^{d}\otimes\mathbb{C}^{d}$.
\end{prop}

\begin{proof}
Appealing to the variational form of the trace norm of an operator
$A$ as
\begin{equation}
\left\Vert A\right\Vert _{1}=\sup_{B\in\mathbb{L}}\left\{ \left|\Tr\!\left[B^{\dag}A\right]\right|:\left\Vert B\right\Vert \leq1\right\} ,
\end{equation}
we can write the root fidelity $\sqrt{F}(\rho,\sigma)$ in a variational
way as follows:
\begin{align}
\left\Vert \sqrt{\rho}\sqrt{\sigma}\right\Vert _{1} & =\sup_{W\in\mathbb{B}}\left|\Tr\!\left[\sqrt{\rho}\sqrt{\sigma}W\right]\right|\\
 & =\sup_{W\in\mathbb{B}}\Re\!\left[\Tr\!\left[\sqrt{\rho}\sqrt{\sigma}W\right]\right]\\
 & =\sup_{W\in\mathbb{B}}\Re\!\left[\langle\Gamma|\left(I\otimes\sqrt{\rho}\sqrt{\sigma}W\right)|\Gamma\rangle\right]\\
 & =\sup_{W\in\mathbb{B}}\Re\!\left[\langle\Gamma|\left(W^{T}\otimes\sqrt{\rho}\sqrt{\sigma}\right)|\Gamma\rangle\right]\\
 & =\sup_{W\in\mathbb{B}}\Re\!\left[\langle\Gamma|\left(I\otimes\sqrt{\rho}\right)\left(W^{T}\otimes I\right)\left(I\otimes\sqrt{\sigma}\right)|\Gamma\rangle\right]\\
 & =\sup_{W\in\mathbb{B}}\Re\!\left[\langle\psi^{\rho}|\left(W^{T}\otimes I\right)|\psi^{\sigma}\rangle\right]\\
 & =\sup_{W\in\mathbb{B}}\Re\!\left[\langle\psi^{\rho}|\left(W\otimes I\right)|\psi^{\sigma}\rangle\right].\label{eq:fidelity-states-contractions}
\end{align}
This proves~\eqref{eq:fid-states-1}. The second equality follows
because, for $z\in\mathbb{C}$, we can write the polar form $z=re^{i\phi}$,
where $r\geq0$ and $\phi\in\left[0,2\pi\right]$. Then $\left|z\right|=r$
and $\Re\!\left[z\right]=r\cos\phi$, and by tuning a global phase
in $W$, we can deduce that the second equality holds. The third equality
follows from~\eqref{eq:max-ent-reduce-to-trace} and the fourth from
\eqref{eq:transpose-trick}. 

By noting that the trace norm of a square operator $A$ has the following
variational characterization as
\begin{equation}
\left\Vert A\right\Vert _{1}=\sup_{U\in\mathbb{U}}\left|\Tr\!\left[AU\right]\right|,
\end{equation}
we can repeat the analysis above to conclude that
\begin{equation}
\left\Vert \sqrt{\rho}\sqrt{\sigma}\right\Vert _{1}=\sup_{U\in\mathbb{U}}\Re\left[\langle\psi^{\rho}|\left(U\otimes I\right)|\psi^{\sigma}\rangle\right],\label{eq:fidelity-states-unitaries}
\end{equation}
thus establishing~\eqref{eq:fid-states-2}.

Additionally, we can embed the contraction $W$ into a unitary matrix
in a standard way (called Halmos dilation~\cite{Halmos1950}) as follows:
\begin{align}
U_{W} & \coloneqq\begin{bmatrix}W & \sqrt{I-WW^{\dag}}\\
\sqrt{I-W^{\dag}W} & -W^{\dag}
\end{bmatrix}\label{eq:halmos-dilation}\\
 & =|0\rangle\!\langle0|\otimes W+|0\rangle\!\langle1|\otimes\sqrt{I-WW^{\dag}}+|1\rangle\!\langle0|\otimes\sqrt{I-W^{\dag}W}-|1\rangle\!\langle1|\otimes W^{\dag},
\end{align}
and observe that
\begin{equation}
\left(\langle0|\otimes I\right)U_{W}\left(|0\rangle\otimes I\right)=W,
\end{equation}
which implies that
\begin{equation}
\langle\psi^{\rho}|\left(W\otimes I\right)|\psi^{\sigma}\rangle=\langle0|\langle\psi^{\rho}|\left(U_{W}\otimes I\right)|0\rangle|\psi^{\sigma}\rangle.\label{eq:unify-contraction-unitary-fidelity}
\end{equation}
As stated above, $|0\rangle|\psi^{\sigma}\rangle,|0\rangle|\psi^{\rho}\rangle\in\mathbb{C}^{2}\otimes\mathbb{C}^{d}\otimes\mathbb{C}^{d}$
and $U_{W}$ acts nontrivially on the first two systems of $\mathbb{C}^{2}\otimes\mathbb{C}^{d}\otimes\mathbb{C}^{d}$.
Using~\eqref{eq:unify-contraction-unitary-fidelity}, we can unify
the two formulas for fidelity in~\eqref{eq:fidelity-states-contractions}
and~\eqref{eq:fidelity-states-unitaries} as follows:
\begin{align}
\left\Vert \sqrt{\rho}\sqrt{\sigma}\right\Vert _{1} & =\sup_{U_{W}\in\mathbb{U}}\Re\!\left[\langle0|\langle\psi^{\rho}|\left(U_{W}\otimes I\right)|0\rangle|\psi^{\sigma}\rangle\right]\label{eq:unify-fidelity-opt-unitary-contractions}\\
 & =\sup_{W\in\mathbb{B}}\Re\!\left[\langle\psi^{\rho}|\left(W\otimes I\right)|\psi^{\sigma}\rangle\right].\label{eq:unify-fidelity-opt-unitary-contractions-2nd}
\end{align}
In~\eqref{eq:unify-fidelity-opt-unitary-contractions}, we have used
that the optimization can be taken over every unitary $U_{W}$, and
not just unitaries of the particular Halmos form in~\eqref{eq:halmos-dilation},
because $\left(\langle0|\otimes I\right)U_{W}\left(|0\rangle\otimes I\right)$
is a contraction for every unitary $U_{W}$.
\end{proof}

\subsubsection{Bures distance of states}

Recall that the Bures distance of states $\rho$ and $\sigma$ is
defined as~\cite{Bures1969}
\begin{align}
d_{B}(\rho,\sigma) & \coloneqq\inf_{U\in\mathbb{U}}\left\Vert \sqrt{\rho}-U\sqrt{\sigma}\right\Vert _{2}\\
 & =\inf_{U\in\mathbb{U}}\left\Vert |\psi^{\rho}\rangle-\left(U\otimes I\right)|\psi^{\sigma}\rangle\right\Vert , \label{eq:Bures_U_Purification}
\end{align}
where $|\psi^{\rho}\rangle$ and $|\psi^{\sigma}\rangle$ are defined
in~\eqref{eq:canonical-pure-rho} and~\eqref{eq:canonical-pure-sigma},
respectively, and the second equality follows from~\eqref{eq:hilbert-schmidt-euclidean-norm}.
\begin{prop}
For $d$-dimensional states $\rho$ and $\sigma$ with canonical purifications
in~\eqref{eq:canonical-pure-rho} and~\eqref{eq:canonical-pure-sigma},
respectively, the following equalities hold:
\begin{align}
d_{B}(\rho,\sigma) & =\sqrt{2\left(1-\sqrt{F}(\rho,\sigma)\right)}\label{eq:bures-state-alt-1}\\
 & =\inf_{W\in\mathbb{B}}\sqrt{2\left(1-\Re\!\left[\langle\psi^{\rho}|\left(W\otimes I\right)|\psi^{\sigma}\rangle\right]\right)}\\
 & =\inf_{U_{W}\in\mathbb{U}}\sqrt{2\left(1-\Re\!\left[\langle0|\langle\psi^{\rho}|\left(U_{W}\otimes I\right)|0\rangle|\psi^{\sigma}\rangle\right]\right)}\label{eq:bures-state-alt-2}\\
 & =\inf_{U_{W}\in\mathbb{U}}\left\Vert |0\rangle|\psi^{\rho}\rangle-\left(U_{W}\otimes I\right)|0\rangle|\psi^{\sigma}\rangle\right\Vert .\label{eq:bures-state-alt-3}
\end{align}
In~\eqref{eq:bures-state-alt-2} and~\eqref{eq:bures-state-alt-3},
$|0\rangle|\psi^{\sigma}\rangle,|0\rangle|\psi^{\rho}\rangle\in\mathbb{C}^{2}\otimes\mathbb{C}^{d}\otimes\mathbb{C}^{d}$
and $U_{W}$ acts nontrivially on the first two systems of $\mathbb{C}^{2}\otimes\mathbb{C}^{d}\otimes\mathbb{C}^{d}$.
\end{prop}

\begin{proof}
Consider from~\eqref{eq:Bures_U_Purification} that
\begin{align}
  d_B(\rho, \sigma)  &=\inf_{U\in\mathbb{U}}\left\Vert |\psi^{\rho}\rangle-\left(U\otimes I\right)|\psi^{\sigma}\rangle\right\Vert \\
 & =\inf_{U\in\mathbb{U}}\sqrt{\left(\langle\psi^{\rho}|-\langle\psi^{\sigma}|\left(U^{\dag}\otimes I\right)\right)\left(|\psi^{\rho}\rangle-\left(U\otimes I\right)|\psi^{\sigma}\rangle\right)}\\
 & =\inf_{U\in\mathbb{U}}\sqrt{\langle\psi^{\rho}|\psi^{\rho}\rangle+\langle\psi^{\sigma}|\left(U^{\dag}\otimes I\right)\left(U\otimes I\right)|\psi^{\sigma}\rangle-2\Re\!\left[\langle\psi^{\rho}|\left(U\otimes I\right)|\psi^{\sigma}\rangle\right]}\\
 & =\inf_{U\in\mathbb{U}}\sqrt{2\left(1-\Re\!\left[\langle\psi^{\rho}|\left(U\otimes I\right)|\psi^{\sigma}\rangle\right]\right)}\\
 & =\sqrt{2\left(1-\sup_{U\in\mathbb{U}}\Re\!\left[\langle\psi^{\rho}|\left(U\otimes I\right)|\psi^{\sigma}\rangle\right]\right)}\\
 & =\sqrt{2\left(1-\sqrt{F}(\rho,\sigma)\right)},
\end{align}
where the last equality follows from Proposition~\ref{prop:fidelity-different-forms}.
By appealing to~\eqref{eq:unify-fidelity-opt-unitary-contractions}
and~\eqref{eq:unify-fidelity-opt-unitary-contractions-2nd}, we can
then write the other expressions in~\eqref{eq:bures-state-alt-1}--\eqref{eq:bures-state-alt-3},
with~\eqref{eq:bures-state-alt-3} following by applying the same
reasoning as above.
\end{proof}

\subsection{Bures distance of channels}

Now let us consider the case of quantum channels $\mathcal{N}_{1}$
and $\mathcal{N}_{2}$ that have $d_{A}$-dimensional input spaces
and $d_{B}$-dimensional output spaces, where $d_{A},d_{B}\in\mathbb{N}$.
The Bures distance of channels $\mathcal{N}_{1}$ and $\mathcal{N}_{2}$
is defined as~\cite{Gilchrist2005}
\begin{equation}
d_{B}(\mathcal{N}_{1},\mathcal{N}_{2})\coloneqq\sup_{\rho_{RA}\in\mathbb{D}}d_{B}\!\left(\left(\id\otimes\mathcal{N}_{1}\right)\left(\rho_{RA}\right),\left(\id\otimes\mathcal{N}_{2}\right)\left(\rho_{RA}\right)\right),
\end{equation}
where the supremum is over every bipartite state $\rho_{RA}$ and
also implicitly over the dimension of the reference system $R$ (taken
to be unbounded in the definition). However, see~\eqref{eq:bures-distance-optimize-pure-states}
for a simplification of the optimization task. 

Canonical Kraus representations for the channels arise from taking
spectral decompositions of the Choi operators of these channels~\cite[Theorem~4.4.1]{Wilde2017},
which are defined for $i\in\left\{ 1,2\right\} $ as follows:
\begin{align}
\Gamma^{\mathcal{N}_{i}} & \coloneqq\left(\id\otimes\mathcal{N}_{i}\right)\left(\Gamma\right),\\
\Gamma & \coloneqq|\Gamma\rangle\!\langle\Gamma|.
\end{align}
Let us denote the canonical Kraus representations by $\left(K_{i}\right)_{i}$
and $\left(L_{j}\right)_{j}$, respectively, and we can take the number
$d_{E}$ of them, without loss of generality, to be the same and no
more than $d_{A}d_{B}$, so that we simply set $d_{E}=d_{A}d_{B}$.
From these, we can construct canonical isometric extensions of the
channels $\mathcal{N}_{1}$ and $\mathcal{N}_{2}$ as follows, respectively:
\begin{align}
V_{1} & \coloneqq\sum_{i}|i\rangle\otimes K_{i},\label{eq:isometric-ext-ch-1}\\
V_{2} & \coloneqq\sum_{j}|j\rangle\otimes L_{j},\label{eq:isometric-ext-ch-2}
\end{align}
where the orthonormal bases for the environments in each case is the
same and the environments have the same dimension. Thus, $V_{i}\colon\mathbb{C}^{d_{A}}\mapsto\mathbb{C}^{d_{E}}\otimes\mathbb{C}^{d_{B}}$
for $i\in\left\{ 1,2\right\} $, where $d_{E}=d_{A}d_{B}$, and $V_{1}$
and $V_{2}$ are isometries because they satisfy $V_{1}^{\dag}V_{1}=\sum_{i}K_{i}^{\dag}K_{i}=I$
and $V_{2}^{\dag}V_{2}=\sum_{j}L_{j}^{\dag}L_{j}=I$.

The following theorem is strongly related to the developments in~\cite{Kretschmann2008},
in particular Theorem 1 therein.
\begin{thm}
\label{thm:Bures-distance-channels}The Bures distance of channels
$\mathcal{N}_{1}$ and $\mathcal{N}_{2}$ can be written as follows:
\begin{align}
d_{B}(\mathcal{N}_{1},\mathcal{N}_{2}) & =\inf_{U_{W}\in\mathbb{U}}\left\Vert \widetilde{V}_{1}-\left(U_{W}\otimes I_{B}\right)\widetilde{V}_{2}\right\Vert ,\\
 & =\sqrt{2}\inf_{W\in\mathbb{B}}\left\Vert I-\Re\!\left[V_{1}^{\dag}\left(W_{E}\otimes I_{B}\right)V_{2}\right]\right\Vert ^{\frac{1}{2}},\label{eq:bures-distance-chs-reduced}
\end{align}
where $\widetilde{V}_{1}$ and $\widetilde{V}_{2}$ are the following
isometric extensions of the channels $\mathcal{N}_{1}$ and $\mathcal{N}_{2}$,
respectively,
\begin{align}
\widetilde{V}_{1} & \coloneqq|0\rangle\otimes V_{1},\label{eq:qubit-augmented-iso-ext-1}\\
\widetilde{V}_{2} & \coloneqq|0\rangle\otimes V_{2},\label{eq:qubit-augmented-iso-ext-2}
\end{align}
$V_{1}$ and $V_{2}$ are the isometric extensions in~\eqref{eq:isometric-ext-ch-1}
and~\eqref{eq:isometric-ext-ch-2}, respectively. Observing that $\widetilde{V}_{i}\colon\mathbb{C}^{d_{A}}\mapsto\mathbb{C}^{2}\otimes\mathbb{C}^{d_{E}}\otimes\mathbb{C}^{d_{B}}$,
let us note that the unitary $U_{W}$ acts nontrivially on the first
two systems of $\mathbb{C}^{2}\otimes\mathbb{C}^{d_{E}}\otimes\mathbb{C}^{d_{B}}$.
\end{thm}

\begin{proof}
As a consequence of the data-processing inequality for the Bures distance
$d_{B}$ of states, purification, and the Schmidt decomposition theorem,
the following equality holds:
\begin{equation}
d_{B}(\mathcal{N}_{1},\mathcal{N}_{2})=\sup_{|\psi\rangle_{RA}\in\mathbb{P}}d_{B}\!\left(\left(\id\otimes\mathcal{N}_{1}\right)\left(\psi_{RA}\right),\left(\id\otimes\mathcal{N}_{2}\right)\left(\psi_{RA}\right)\right),\label{eq:bures-distance-optimize-pure-states}
\end{equation}
where $\psi_{RA}\equiv|\psi\rangle\!\langle\psi|_{RA}$ and the reference
system $R$ is isomorphic to the channel input system $A$ (see~\cite[Proposition~7.82]{Khatri2024}
for further details). Noting that every pure bipartite state vector
$|\psi\rangle_{RA}$ can be written as follows:
\begin{equation}
|\psi\rangle_{RA}=\left(U_{R}\otimes\sqrt{\rho}\right)|\Gamma\rangle
\end{equation}
where $U_{R}$ is a unitary and $\rho$ is a density operator, we
can take canonical purifications of $\left(\id\otimes\mathcal{N}_{1}\right)\left(\psi_{RA}\right)$
and $\left(\id\otimes\mathcal{N}_{2}\right)\left(\psi_{RA}\right)$
to be as follows, respectively:
\begin{align}
V_{1}|\psi\rangle_{RA} & =\left(U_{R}\otimes V_{1}\sqrt{\rho}\right)|\Gamma\rangle=\sum_{i}|i\rangle_{E}\otimes\left(U_{R}\otimes K_{i}\sqrt{\rho}\right)|\Gamma\rangle,\\
V_{2}|\psi\rangle_{RA} & =\left(U_{R}\otimes V_{2}\sqrt{\rho}\right)|\Gamma\rangle=\sum_{j}|j\rangle_{E}\otimes\left(U_{R}\otimes L_{j}\sqrt{\rho}\right)|\Gamma\rangle.
\end{align}
By appealing to~\eqref{eq:bures-state-alt-1}--\eqref{eq:bures-state-alt-3},
we can write the Bures distance of channels as follows: 
\begin{align}
d_{B}(\mathcal{N}_{1},\mathcal{N}_{2}) & =\sup_{\left|\psi\right\rangle_{RA}\in\mathbb{P}}\sqrt{2\left(1-\sqrt{F}\!\left(\left(\id\otimes\mathcal{N}_{1}\right)\left(\psi_{RA}\right),\left(\id\otimes\mathcal{N}_{2}\right)\left(\psi_{RA}\right)\right)\right)}\\
 & =\sup_{\left|\psi\right\rangle_{RA}\in\mathbb{P}}\inf_{W\in\mathbb{B}}\sqrt{2\left(1-\Re\!\left[\langle\psi|_{RA}V_{1}^{\dag}\left(W_{E}\otimes I_{RB}\right)V_{2}|\psi\rangle_{RA}\right]\right)},\label{eq:root-fid-chs-bures-distance}
\end{align}
where we have denoted the channel output system by $B$ and the environment
system by $E$. Now consider that
\begin{align}
 & \langle\psi|_{RA}V_{1}^{\dag}\left(W_{E}\otimes I_{RB}\right)V_{2}|\psi\rangle_{RA}\nonumber \\
 & =\langle\psi|_{RA}\left(I_{R}\otimes V_{1}^{\dag}\left(W_{E}\otimes I_{B}\right)V_{2}\right)|\psi\rangle_{RA}\label{eq:reduction-fidelity-chs-1}\\
 & =\langle\Gamma|\left(U_{R}^{\dag}\otimes\sqrt{\rho}\right)\left(I_{R}\otimes V_{1}^{\dag}\left(W_{E}\otimes I_{B}\right)V_{2}\right)\left(U_{R}\otimes\sqrt{\rho}\right)|\Gamma\rangle\\
 & =\langle\Gamma|\left(U_{R}^{\dag}U_{R}\otimes\sqrt{\rho}\right)V_{1}^{\dag}\left(W_{E}\otimes I_{B}\right)V_{2}\sqrt{\rho}|\Gamma\rangle\\
 & =\langle\Gamma|\left(I_{R}\otimes\sqrt{\rho} \ V_{1}^{\dag}\left(W_{E}\otimes I_{B}\right)V_{2}\sqrt{\rho}\right)|\Gamma\rangle\\
 & =\Tr\!\left[\sqrt{\rho} \ V_{1}^{\dag}\left(W_{E}\otimes I_{B}\right)V_{2}\sqrt{\rho} \right]\\
 & =\Tr\!\left[V_{1}^{\dag}\left(W_{E}\otimes I_{B}\right)V_{2}\rho\right].\label{eq:reduction-fidelity-chs-last}
\end{align}
Then
\begin{align}
 & \sup_{|\psi\rangle_{RA}\in\mathbb{P}}\inf_{W\in\mathbb{B}}\sqrt{2\left(1-\Re\!\left[\langle\psi|_{RA}V_{1}^{\dag}\left(W_{E}\otimes I_{RB}\right)V_{2}|\psi\rangle_{RA}\right]\right)}\nonumber \\
 & =\sup_{\rho\in\mathbb{D},U\in\mathbb{U}}\inf_{W\in\mathbb{B}}\sqrt{2\left(1-\Re\!\left[\Tr\!\left[V_{1}^{\dag}\left(W_{E}\otimes I_{B}\right)V_{2}\rho\right]\right]\right)}\\
 & =\sup_{\rho\in\mathbb{D}}\inf_{W\in\mathbb{B}}\sqrt{2\left(1-\Re\!\left[\Tr\!\left[V_{1}^{\dag}\left(W_{E}\otimes I_{B}\right)V_{2}\rho\right]\right]\right)}\\
 & =\sqrt{2\left(1-\inf_{\rho\in\mathbb{D}}\sup_{W\in\mathbb{B}}\Re\!\left[\Tr\!\left[V_{1}^{\dag}\left(W_{E}\otimes I_{B}\right)V_{2}\rho\right]\right]\right)}\\
 & =\sqrt{2\left(1-\sup_{W\in\mathbb{B}}\inf_{\rho\in\mathbb{D}}\Re\!\left[\Tr\!\left[V_{1}^{\dag}\left(W_{E}\otimes I_{B}\right)V_{2}\rho\right]\right]\right)}\label{eq:minimax-step-bures-channels}\\
 & =\inf_{W\in\mathbb{B}}\sup_{\rho\in\mathbb{D}}\sqrt{2\left(1-\Re\!\left[\Tr\!\left[V_{1}^{\dag}\left(W_{E}\otimes I_{B}\right)V_{2}\rho\right]\right]\right)}\\
 & =\inf_{W\in\mathbb{B}}\sup_{|\psi\rangle_{RA}\in\mathbb{P}}\sqrt{2\left(1-\Re\!\left[\langle\psi|_{RA}V_{1}^{\dag}\left(W_{E}\otimes I_{RB}\right)V_{2}|\psi\rangle_{RA}\right]\right)} \label{eq:rho_to_psi}\\
 & =\inf_{U_{W}\in\mathbb{U}}\sup_{|\psi\rangle_{RA}\in\mathbb{P}}\sqrt{2\left(1-\Re\!\left[\langle0|\langle\psi|_{RA}V_{1}^{\dag}\left(U_{W}\otimes I_{RB}\right)|0\rangle V_{2}|\psi\rangle_{RA}\right]\right)}\\
 & =\inf_{U_{W}\in\mathbb{U}}\sup_{|\psi\rangle_{RA}\in\mathbb{P}}\left\Vert |0\rangle V_{1}|\psi\rangle_{RA}-\left(U_{W}\otimes I_{RB}\right)|0\rangle V_{2}|\psi\rangle_{RA}\right\Vert \\
 & =\inf_{U_{W}\in\mathbb{U}}\sup_{|\psi\rangle_{RA}\in\mathbb{P}}\left\Vert \widetilde{V}_{1}|\psi\rangle_{RA}-\left(U_{W}\otimes I_{RB}\right)\widetilde{V}_{2}|\psi\rangle_{RA}\right\Vert \label{eq:define-tilde-V}\\
 & =\inf_{U_{W}\in\mathbb{U}}\left\Vert I_{R}\otimes\widetilde{V}_{1}-\left(U_{W}\otimes I_{RB}\right)\left(I_{R}\otimes\widetilde{V}_{2}\right)\right\Vert \\
 & =\inf_{U_{W}\in\mathbb{U}}\left\Vert I_{R}\otimes\left(\widetilde{V}_{1}-\left(U_{W}\otimes I_{B}\right)\widetilde{V}_{2}\right)\right\Vert \\
 & =\inf_{U_{W}\in\mathbb{U}}\left\Vert \widetilde{V}_{1}-\left(U_{W}\otimes I_{B}\right)\widetilde{V}_{2}\right\Vert ,\label{eq:final-eq-bures-channels-proof}
\end{align}
The equality in~\eqref{eq:minimax-step-bures-channels} follows from
the minimax theorem: indeed, the set of density operators is convex
and compact, the set of contractions is convex and compact, the function
\begin{equation}
\rho\mapsto\Re\!\left[\Tr\!\left[V_{1}^{\dag}\left(W_{E}\otimes I_{B}\right)V_{2}\rho\right]\right]
\end{equation}
is linear in $\rho$ (and thus convex in $\rho$), and the function
\begin{equation}
W\mapsto\Re\!\left[\Tr\!\left[V_{1}^{\dag}\left(W_{E}\otimes I_{B}\right)V_{2}\rho\right]\right]
\end{equation}
is real linear in $W$ (and thus concave in $W$). The equality in~\eqref{eq:rho_to_psi} follows from~\eqref{eq:reduction-fidelity-chs-last}. Additionally, the
equality in~\eqref{eq:define-tilde-V} follows from defining $\widetilde{V}_{i}\colon\mathbb{C}^{d_{A}}\mapsto\mathbb{C}^{2}\otimes\mathbb{C}^{d_{E}}\otimes\mathbb{C}^{d_{B}}$
as the following isometry, for $i\in\left\{ 1,2\right\} $:
\begin{equation}
\widetilde{V}_{i}|\varphi\rangle_{A}\coloneqq|0\rangle\otimes V_{i}|\varphi\rangle_{A}.
\end{equation}
That is, these isometries can be written as in~\eqref{eq:qubit-augmented-iso-ext-1}--\eqref{eq:qubit-augmented-iso-ext-2}.
The final equality in~\eqref{eq:final-eq-bures-channels-proof} follows
because the spectral norm is multiplicative on tensor products (i.e.,
$\left\Vert A\otimes B\right\Vert =\left\Vert A\right\Vert \left\Vert B\right\Vert $)
and $\left\Vert I\right\Vert =1$.
\end{proof}
\begin{rem}
\label{rem:isometric-overlap-to-kraus}Relevant for our applications
in channel discrimination, we can alternatively write the overlap
term in Theorem~\ref{thm:Bures-distance-channels} in terms of the
Kraus operators $\left(K_{i}\right)_{i}$ and $\left(L_{j}\right)_{j}$
of channels $\mathcal{N}_{1}$ and $\mathcal{N}_{2}$, respectively,
as follows:
\begin{align}
V_{1}^{\dag}\left(W_{E}\otimes I_{B}\right)V_{2} & =\left(\sum_{i}\langle i|\otimes K_{i}^{\dag}\right)\left(W_{E}\otimes I_{B}\right)\left(\sum_{j}|j\rangle\otimes L_{j}\right)\\
 & =\sum_{i,j}\langle i|W|j\rangle\otimes K_{i}^{\dag}L_{j}\\
 & =\sum_{i,j}W_{i,j}K_{i}^{\dag}L_{j},
\end{align}
where we defined $W_{i,j}\coloneqq\langle i|W|j\rangle$. This allows
for connecting to the formulations in~\cite{Yuan2017}.
\end{rem}

\section{SLD Fisher information of states and channels}\label{sec:SLD-Fisher-information-states-chs}

\subsection{SLD Fisher information of states}

In this section, we provide alternative formulas for the symmetric
logarithmic derivative (SLD) Fisher information of a smooth family
$\left(\rho_{\theta}\right)_{\theta\in\Theta}$ of positive definite
states. These formulations are related to those given in~\cite[Theorem~1]{Fujiwara2008},
but we expand on the formulation given there and provide an alternative
proof. To begin with, let us recall that the SLD Fisher information
of a smooth family $\left(\rho_{\theta}\right)_{\theta\in\Theta}$
of positive definite states can be written as follows~\cite[Eq.~(88)]{Sidhu2020}:
\begin{equation}
I_{F}\!\left(\theta;\left(\rho_{\theta}\right)_{\theta}\right)=\sum_{\ell,m}\frac{2}{\lambda_{\ell}+\lambda_{m}}\Tr\!\left[\Pi_{\ell}\left(\partial_{\theta}\rho_{\theta}\right)\Pi_{m}\left(\partial_{\theta}\rho_{\theta}\right)\right],
\end{equation}
where $\partial_{\theta}\equiv\frac{\partial}{\partial\theta}$,
\begin{equation}
\rho_{\theta}=\sum_{k}\lambda_{k}\Pi_{k}\label{eq:spectral-decomp-rho-theta}
\end{equation}
is the spectral decomposition of $\rho_{\theta}$, with the dependence
of each eigenvalue $\lambda_{k}$ and each eigenprojection $\Pi_{k}$
on $\theta$ suppressed in the notation. It is well known that the
following equalities hold, relating Bures distance to SLD Fisher information
\cite{Huebner1992}:
\begin{align}
d_{B}^{2}\!\left(\rho_{\theta},\rho_{\theta+\delta}\right) & =\frac{\delta^{2}}{4}I_{F}\!\left(\theta;\left(\rho_{\theta}\right)_{\theta}\right)+o(\delta^{2})\label{eq:SLD-Fisher-to-Bures}\\
I_{F}\!\left(\theta;\left(\rho_{\theta}\right)_{\theta}\right) & =2\left.\frac{\partial^{2}}{\partial\delta^{2}}d_{B}^{2}\!\left(\rho_{\theta},\rho_{\theta+\delta}\right)\right|_{\delta=0}\\
 & =4\lim_{\delta\to0}\frac{d_{B}^{2}\!\left(\rho_{\theta},\rho_{\theta+\delta}\right)}{\delta^{2}}.
\end{align}

\begin{thm}
\label{thm:SLD-Fisher-states}For a smooth family $\left(\rho_{\theta}\right)_{\theta\in\Theta}$
of positive definite states, the following equality holds: 
\begin{align}
\frac{1}{4}I_{F}\!\left(\theta;\left(\rho_{\theta}\right)_{\theta}\right) & =\inf_{\left(U_{\theta}\right)_{\theta\in\Theta}}\left\Vert \partial_{\theta}\left(\left(U_{\theta}\otimes I\right)|\psi_{\theta}\rangle\right)\right\Vert ^{2}\label{eq:unitary-herm-op-SLD-Fisher-states}\\
 & =\inf_{\left(H_{\theta}\right)_{\theta\in\Theta}}\left\Vert \left(\partial_{\theta}\sqrt{\rho_{\theta}}\right)-i\sqrt{\rho_{\theta}}H_{\theta}\right\Vert _{2}^{2}\label{eq:herm-opt-SLD-fisher-states}\\
 & =\Tr\!\left[\left(\partial_{\theta}\sqrt{\rho_{\theta}}\right)^{2}\right]-\Tr\!\left[\left(H_{\theta}^{\star}\right)^{2}\rho_{\theta}\right],
\end{align}
where $|\psi_{\theta}\rangle\coloneqq\left(I\otimes\sqrt{\rho_{\theta}}\right)|\Gamma\rangle$
is the canonical purification of $\rho_{\theta}$ and the optimizations
in~\eqref{eq:unitary-herm-op-SLD-Fisher-states} and~\eqref{eq:herm-opt-SLD-fisher-states}
are over every smooth family $\left(U_{\theta}\right)_{\theta\in\Theta}$
of unitary operators and every smooth family $\left(H_{\theta}\right)_{\theta\in\Theta}$
of Hermitian operators, respectively. The optimal choice $H_{\theta}^{\star}$
in~\eqref{eq:herm-opt-SLD-fisher-states} satisfies
\begin{equation}
\rho_{\theta}H_{\theta}^{\star}+H_{\theta}^{\star}\rho_{\theta}=i\left[\partial_{\theta}\sqrt{\rho_{\theta}},\sqrt{\rho_{\theta}}\right]
\end{equation}
and is thus given by
\begin{align}
H_{\theta}^{\star} & \coloneqq i\int_{0}^{\infty}ds\ e^{-s\rho_{\theta}}\left[\partial_{\theta}\sqrt{\rho_{\theta}},\sqrt{\rho_{\theta}}\right]e^{-s\rho_{\theta}}.\\
 & =-i\sum_{\ell,m}\frac{\lambda_{\ell}^{\frac{1}{2}}-\lambda_{m}^{\frac{1}{2}}}{\left(\lambda_{\ell}+\lambda_{m}\right)\left(\lambda_{\ell}^{\frac{1}{2}}+\lambda_{m}^{\frac{1}{2}}\right)}\,\Pi_{\ell}\left(\partial_{\theta}\rho_{\theta}\right)\Pi_{m},
\end{align}
where we take the spectral decomposition of $\rho_{\theta}$ as in
\eqref{eq:spectral-decomp-rho-theta}.
\end{thm}

\begin{proof}
Consider that
\begin{align}
\inf_{\left(U_{\theta}\right)_{\theta\in\Theta}}\left\Vert \partial_{\theta}\left(\left(U_{\theta}\otimes I\right)|\psi_{\theta}\rangle\right)\right\Vert ^{2} & =\inf_{\left(U_{\theta}\right)_{\theta\in\Theta}}\left\Vert \left(\left(\partial_{\theta}U_{\theta}\right)\otimes I\right)|\psi_{\theta}\rangle+\left(U_{\theta}\otimes I\right)|\partial_{\theta}\psi_{\theta}\rangle\right\Vert ^{2}\\
 & =\inf_{\left(U_{\theta}\right)_{\theta\in\Theta}}\left\Vert U_{\theta}^{\dag}\left[\left(\left(\partial_{\theta}U_{\theta}\right)\otimes I\right)|\psi_{\theta}\rangle+\left(U_{\theta}\otimes I\right)|\partial_{\theta}\psi_{\theta}\rangle\right]\right\Vert ^{2}\\
 & =\inf_{\left(U_{\theta}\right)_{\theta\in\Theta}}\left\Vert \left(U_{\theta}^{\dag}\left(\partial_{\theta}U_{\theta}\right)\otimes I\right)|\psi_{\theta}\rangle+|\partial_{\theta}\psi_{\theta}\rangle\right\Vert ^{2}\\
 & =\inf_{\left(U_{\theta}\right)_{\theta\in\Theta}}\left\Vert \left|\partial_{\theta}\psi_{\theta}\right\rangle +\left(U_{\theta}^{\dag}\left(\partial_{\theta}U_{\theta}\right)\otimes I\right)|\psi_{\theta}\rangle\right\Vert ^{2}.
\end{align}
We claim that $\left(H_{\theta}\right)_{\theta\in\Theta}$, where
$iU_{\theta}^{\dag}\left(\partial_{\theta}U_{\theta}\right)\equiv H_{\theta}\in\mathbb{H}$,
is a smooth family of Hermitian operators. Indeed, by setting $U_{\theta}=e^{-iG_{\theta}}$,
we conclude from Duhamel's formula (see, e.g.,~\cite[Proposition~47]{Wilde2025})
that
\begin{align}
i\left(U_{\theta}\right)^{\dag}\left(\partial_{\theta}U_{\theta}\right) & =ie^{iG_{\theta}}\left(\partial_{\theta}e^{-iG_{\theta}}\right)\label{eq:hermitian-op-from-unitary-1}\\
 & =ie^{iG_{\theta}}\int_{0}^{1}dt\ e^{-iG_{\theta}t}\left(\partial_{\theta}\left(-iG_{\theta}\right)\right)e^{-iG_{\theta}\left(1-t\right)}\\
 & =\int_{0}^{1}dt\ e^{iG_{\theta}\left(1-t\right)}\left(\partial_{\theta}G_{\theta}\right)e^{-iG_{\theta}\left(1-t\right)},
\end{align}
which is evidently Hermitian. To solve for $U_{\theta}$ from $H_{\theta}$,
we can solve the following matrix differential equation for $U_{\theta}$:
\begin{equation}
H_{\theta}=i\left(U_{\theta}\right)^{\dag}\left(\partial_{\theta}U_{\theta}\right)\qquad\Longleftrightarrow\qquad-iU_{\theta}H_{\theta}=\partial_{\theta}U_{\theta},
\end{equation}
which is equivalent to the time-dependent Schr\"odinger equation
and well known to have a solution given by the following Dyson series
\cite{Reed1975}:
\begin{equation}
U_{\theta}=I+\sum_{k=1}^{\infty}\left(-i\right)^{k}\int_{0}^{t}dt_{1}\int_{0}^{t_{1}}dt_{2}\cdots\int_{0}^{t_{k-1}}dt_{k}\,H(t_{1})H(t_{2})\cdots H(t_{k}).\label{eq:hermitian-op-from-unitary-2}
\end{equation}
This Dyson series is often abbreviated as the time-ordered exponential
$\mathcal{T}\exp\!\left(-i\int_{0}^{t}d\tau\,H(\tau)\right)$. Continuing,
we make the substitution $U_{\theta}^{\dag}\left(\partial_{\theta}U_{\theta}\right)\to-iH_{\theta}$
and find that
\begin{align}
 & \inf_{\left(U_{\theta}\right)_{\theta\in\Theta}}\left\Vert \left|\partial_{\theta}\psi_{\theta}\right\rangle +\left(U_{\theta}^{\dag}\left(\partial_{\theta}U_{\theta}\right)\otimes I\right)|\psi_{\theta}\rangle\right\Vert ^{2}\nonumber \\
 & =\inf_{\left(H_{\theta}\right)_{\theta\in\Theta}}\left\Vert \left|\partial_{\theta}\psi_{\theta}\right\rangle -i\left(H_{\theta}\otimes I\right)|\psi_{\theta}\rangle\right\Vert ^{2}\\
 & =\inf_{\left(H_{\theta}\right)_{\theta\in\Theta}}\left\Vert \left(I\otimes\left(\partial_{\theta}\sqrt{\rho_{\theta}}\right)\right)|\Gamma\rangle-i\left(H_{\theta}\otimes\sqrt{\rho_{\theta}}\right)|\Gamma\rangle\right\Vert ^{2}\\
 & =\inf_{\left(H_{\theta}\right)_{\theta\in\Theta}}\left\Vert \left(I\otimes\left(\partial_{\theta}\sqrt{\rho_{\theta}}\right)\right)|\Gamma\rangle-i\left(I\otimes\sqrt{\rho_{\theta}}H_{\theta}^{T}\right)|\Gamma\rangle\right\Vert ^{2}\\
 & =\inf_{\left(H_{\theta}\right)_{\theta\in\Theta}}\left\Vert \left(I\otimes\left(\partial_{\theta}\sqrt{\rho_{\theta}}\right)-i\left(I\otimes\sqrt{\rho_{\theta}}H_{\theta}\right)\right)|\Gamma\rangle\right\Vert ^{2}\\
 & =\inf_{\left(H_{\theta}\right)_{\theta\in\Theta}}\left\Vert \left(\partial_{\theta}\sqrt{\rho_{\theta}}\right)-i\sqrt{\rho_{\theta}}H_{\theta}\right\Vert _{2}^{2}.
\end{align}
At this point, observe that the function $H_{\theta}\mapsto\left\Vert \left(\partial_{\theta}\sqrt{\rho_{\theta}}\right)-i\sqrt{\rho_{\theta}}H_{\theta}\right\Vert _{2}^{2}$
is convex in $H_{\theta}$, because it is a composition of a convex
function (the Schatten $2$-norm) and a convex and non-decreasing
function (the square function $x\mapsto x^{2}$). Since our goal is
to minimize this expression, convexity in $H_{\theta}$ implies that
a local minimum is a global minimum. Let us expand the objective function
as follows:
\begin{align}
\left\Vert \left(\partial_{\theta}\sqrt{\rho_{\theta}}\right)-i\sqrt{\rho_{\theta}}H_{\theta}\right\Vert _{2}^{2} & =\Tr\!\left[\left(\left(\partial_{\theta}\sqrt{\rho_{\theta}}\right)-i\sqrt{\rho_{\theta}}H_{\theta}\right)^{\dag}\left(\left(\partial_{\theta}\sqrt{\rho_{\theta}}\right)-i\sqrt{\rho_{\theta}}H_{\theta}\right)\right]\\
 & =\Tr\!\left[\left(\left(\partial_{\theta}\sqrt{\rho_{\theta}}\right)+iH_{\theta}\sqrt{\rho_{\theta}}\right)\left(\left(\partial_{\theta}\sqrt{\rho_{\theta}}\right)-i\sqrt{\rho_{\theta}}H_{\theta}\right)\right]\\
 & =\Tr\!\left[\left(\partial_{\theta}\sqrt{\rho_{\theta}}\right)^{2}\right]-i\Tr\!\left[\left(\partial_{\theta}\sqrt{\rho_{\theta}}\right)\sqrt{\rho_{\theta}}H_{\theta}\right]\nonumber \\
 & \qquad+i\Tr\!\left[H_{\theta}\sqrt{\rho_{\theta}}\left(\partial_{\theta}\sqrt{\rho_{\theta}}\right)\right]+\Tr\!\left[H_{\theta}^{2}\rho_{\theta}\right]\\
 & =\Tr\!\left[\left(\partial_{\theta}\sqrt{\rho_{\theta}}\right)^{2}\right]+\Tr\!\left[H_{\theta}^{2}\rho_{\theta}\right]-i\Tr\!\left[\left[\left(\partial_{\theta}\sqrt{\rho_{\theta}}\right),\sqrt{\rho_{\theta}}\right]H_{\theta}\right].\label{eq:obj-func-SLD-Fisher-states}
\end{align}
We can then take the matrix derivative of~\eqref{eq:obj-func-SLD-Fisher-states}
with respect to $H_{\theta}$ and find that
\begin{multline}
\frac{\partial}{\partial H_{\theta}}\left(\Tr\!\left[\left(\partial_{\theta}\sqrt{\rho_{\theta}}\right)^{2}\right]+\Tr\!\left[H_{\theta}^{2}\rho_{\theta}\right]-i\Tr\!\left[\left[\left(\partial_{\theta}\sqrt{\rho_{\theta}}\right),\sqrt{\rho_{\theta}}\right]H_{\theta}\right]\right)\\
=\rho_{\theta}H_{\theta}+H_{\theta}\rho_{\theta}-i\left[\partial_{\theta}\sqrt{\rho_{\theta}},\sqrt{\rho_{\theta}}\right],
\end{multline}
where we used the facts that $\frac{\partial}{\partial X}\Tr\!\left[AX^{2}\right]=AX+XA$
and $\frac{\partial}{\partial X}\Tr\!\left[AX\right]=A$. Setting
the matrix derivative equal to zero and rewriting implies that the
optimal $H_{\theta}^{\star}$ satisfies the following equation:
\begin{equation}
\rho_{\theta}H_{\theta}^{\star}+H_{\theta}^{\star}\rho_{\theta}=i\left[\partial_{\theta}\sqrt{\rho_{\theta}},\sqrt{\rho_{\theta}}\right].\label{eq:matrix-grad-eq-SLD-Fisher-states}
\end{equation}
By plugging~\eqref{eq:matrix-grad-eq-SLD-Fisher-states} into~\eqref{eq:obj-func-SLD-Fisher-states},
observe that the objective function can be written as
\begin{align}
 & \Tr\!\left[\left(\partial_{\theta}\sqrt{\rho_{\theta}}\right)^{2}\right]+\Tr\!\left[\left(H_{\theta}^{\star}\right)^{2}\rho_{\theta}\right]-i\Tr\!\left[\left[\left(\partial_{\theta}\sqrt{\rho_{\theta}}\right),\sqrt{\rho_{\theta}}\right]H_{\theta}^{\star}\right]\nonumber \\
 & =\Tr\!\left[\left(\partial_{\theta}\sqrt{\rho_{\theta}}\right)^{2}\right]+\Tr\!\left[\left(H_{\theta}^{\star}\right)^{2}\rho_{\theta}\right]-\Tr\!\left[\left(\rho_{\theta}H_{\theta}^{\star}+H_{\theta}^{\star}\rho_{\theta}\right)H_{\theta}^{\star}\right]\\
 & =\Tr\!\left[\left(\partial_{\theta}\sqrt{\rho_{\theta}}\right)^{2}\right]+\Tr\!\left[\left(H_{\theta}^{\star}\right)^{2}\rho_{\theta}\right]-2\Tr\!\left[\left(H_{\theta}^{\star}\right)^{2}\rho_{\theta}\right]\\
 & =\Tr\!\left[\left(\partial_{\theta}\sqrt{\rho_{\theta}}\right)^{2}\right]-\Tr\!\left[\left(H_{\theta}^{\star}\right)^{2}\rho_{\theta}\right].
\end{align}
By employing~\eqref{eq:matrix-grad-eq-SLD-Fisher-states}, this formula
can alternatively be written as
\begin{equation}
\Tr\!\left[\left(\partial_{\theta}\sqrt{\rho_{\theta}}\right)^{2}\right]-\frac{i}{2}\Tr\!\left[\left[\left(\partial_{\theta}\sqrt{\rho_{\theta}}\right),\sqrt{\rho_{\theta}}\right]H_{\theta}^{\star}\right].
\end{equation}

The equation in~\eqref{eq:matrix-grad-eq-SLD-Fisher-states} has the
form of a Sylvester--Lyapunov equation, which is well known to have
the following solution~\cite[Theorem~5]{Lancaster1970}:
\begin{align}
H_{\theta}^{\star} & =i\int_{0}^{\infty}ds\ e^{-s\rho_{\theta}}\left[\partial_{\theta}\sqrt{\rho_{\theta}},\sqrt{\rho_{\theta}}\right]e^{-s\rho_{\theta}}\\
 & =i\sum_{\ell,m}\frac{1}{\lambda_{\ell}+\lambda_{m}}\Pi_{\ell}\left[\partial_{\theta}\sqrt{\rho_{\theta}},\sqrt{\rho_{\theta}}\right]\Pi_{m},
\end{align}
where we employed the integral $\int_{0}^{\infty}ds\ e^{-s\left(x+y\right)}=\frac{1}{x+y}$,
holding for $x,y>0$, as well as the spectral decomposition of $\rho_{\theta}$.
Now recall that (see, e.g.,~\cite[Theorem~42]{Wilde2025})
\begin{align}
\partial_{\theta}\sqrt{\rho_{\theta}} & =\sum_{\ell,m}f_{x^{\frac{1}{2}}}^{[1]}(\lambda_{\ell},\lambda_{m})\Pi_{\ell}\left(\partial_{\theta}\rho_{\theta}\right)\Pi_{m}\\
 & =\sum_{\ell,m}\frac{1}{\lambda_{\ell}^{\frac{1}{2}}+\lambda_{m}^{\frac{1}{2}}}\Pi_{\ell}\left(\partial_{\theta}\rho_{\theta}\right)\Pi_{m},
\end{align}
where we used the first divided difference of the square-root function,
defined for $x,y>0$ as
\begin{equation}
f_{x^{\frac{1}{2}}}^{[1]}(x,y)\coloneqq\begin{cases}
\frac{1}{2x^{\frac{1}{2}}} & :x=y\\
\frac{x^{\frac{1}{2}}-y^{\frac{1}{2}}}{x-y} & :x\neq y
\end{cases}.
\end{equation}
This further implies that
\begin{align}
\left[\partial_{\theta}\sqrt{\rho_{\theta}},\sqrt{\rho_{\theta}}\right] & =\left[\sum_{r,s}\frac{1}{\lambda_{r}^{\frac{1}{2}}+\lambda_{s}^{\frac{1}{2}}}\Pi_{r}\left(\partial_{\theta}\rho_{\theta}\right)\Pi_{s},\sqrt{\rho_{\theta}}\right]\\
 & =\sum_{r,s}\frac{1}{\lambda_{r}^{\frac{1}{2}}+\lambda_{s}^{\frac{1}{2}}}\left[\Pi_{r}\left(\partial_{\theta}\rho_{\theta}\right)\Pi_{s},\sqrt{\rho_{\theta}}\right]\\
 & =\sum_{r,s}\frac{1}{\lambda_{r}^{\frac{1}{2}}+\lambda_{s}^{\frac{1}{2}}}\left(\Pi_{r}\left(\partial_{\theta}\rho_{\theta}\right)\Pi_{s}\sqrt{\rho_{\theta}}-\sqrt{\rho_{\theta}}\Pi_{r}\left(\partial_{\theta}\rho_{\theta}\right)\Pi_{s}\right)\\
 & =\sum_{r,s}\left(\frac{\lambda_{s}^{\frac{1}{2}}-\lambda_{r}^{\frac{1}{2}}}{\lambda_{r}^{\frac{1}{2}}+\lambda_{s}^{\frac{1}{2}}}\right)\Pi_{r}\left(\partial_{\theta}\rho_{\theta}\right)\Pi_{s},
\end{align}
and
\begin{align}
H_{\theta}^{\star} & =i\sum_{\ell,m}\frac{1}{\lambda_{\ell}+\lambda_{m}}\Pi_{\ell}\left[\partial_{\theta}\sqrt{\rho_{\theta}},\sqrt{\rho_{\theta}}\right]\Pi_{m}\\
 & =i\sum_{\ell,m}\frac{1}{\lambda_{\ell}+\lambda_{m}}\Pi_{\ell}\sum_{r,s}\left(\frac{\lambda_{s}^{\frac{1}{2}}-\lambda_{r}^{\frac{1}{2}}}{\lambda_{r}^{\frac{1}{2}}+\lambda_{s}^{\frac{1}{2}}}\right)\Pi_{r}\left(\partial_{\theta}\rho_{\theta}\right)\Pi_{s}\Pi_{m}\\
 & =i\sum_{\ell,m,r,s}\frac{\lambda_{s}^{\frac{1}{2}}-\lambda_{r}^{\frac{1}{2}}}{\left(\lambda_{\ell}+\lambda_{m}\right)\left(\lambda_{r}^{\frac{1}{2}}+\lambda_{s}^{\frac{1}{2}}\right)}\,\Pi_{\ell}\Pi_{r}\left(\partial_{\theta}\rho_{\theta}\right)\Pi_{s}\Pi_{m}\\
 & =i\sum_{\ell,m}\frac{\lambda_{m}^{\frac{1}{2}}-\lambda_{\ell}^{\frac{1}{2}}}{\left(\lambda_{\ell}+\lambda_{m}\right)\left(\lambda_{\ell}^{\frac{1}{2}}+\lambda_{m}^{\frac{1}{2}}\right)}\,\Pi_{\ell}\left(\partial_{\theta}\rho_{\theta}\right)\Pi_{m}\\
 & =-i\sum_{\ell,m}\frac{\left(\lambda_{\ell}^{\frac{1}{2}}-\lambda_{m}^{\frac{1}{2}}\right)}{\left(\lambda_{\ell}+\lambda_{m}\right)\left(\lambda_{\ell}^{\frac{1}{2}}+\lambda_{m}^{\frac{1}{2}}\right)}\,\Pi_{\ell}\left(\partial_{\theta}\rho_{\theta}\right)\Pi_{m}.
\end{align}

Then consider that
\begin{align}
 & \Tr\!\left[\left(\partial_{\theta}\sqrt{\rho_{\theta}}\right)^{2}\right]-\frac{i}{2}\Tr\!\left[\left[\left(\partial_{\theta}\sqrt{\rho_{\theta}}\right),\sqrt{\rho_{\theta}}\right]H_{\theta}^{\star}\right]\nonumber \\
 & =\Tr\!\left[\left(\partial_{\theta}\sqrt{\rho_{\theta}}\right)^{2}\right]-\frac{i}{2}\Tr\!\left[\begin{array}{c}
\left(\sum_{r,s}\left(\frac{\lambda_{s}^{\frac{1}{2}}-\lambda_{r}^{\frac{1}{2}}}{\lambda_{r}^{\frac{1}{2}}+\lambda_{s}^{\frac{1}{2}}}\right)\Pi_{r}\left(\partial_{\theta}\rho_{\theta}\right)\Pi_{s}\right)\times\\
\left(-i\sum_{\ell,m}\frac{\left(\lambda_{\ell}^{\frac{1}{2}}-\lambda_{m}^{\frac{1}{2}}\right)}{\left(\lambda_{\ell}+\lambda_{m}\right)\left(\lambda_{\ell}^{\frac{1}{2}}+\lambda_{m}^{\frac{1}{2}}\right)}\,\Pi_{\ell}\left(\partial_{\theta}\rho_{\theta}\right)\Pi_{m}\right)
\end{array}\right]\\
 & =\Tr\!\left[\left(\partial_{\theta}\sqrt{\rho_{\theta}}\right)^{2}\right]-\frac{1}{2}\Tr\!\left[\begin{array}{c}
\left(\sum_{r,s}\left(\frac{\lambda_{s}^{\frac{1}{2}}-\lambda_{r}^{\frac{1}{2}}}{\lambda_{r}^{\frac{1}{2}}+\lambda_{s}^{\frac{1}{2}}}\right)\Pi_{r}\left(\partial_{\theta}\rho_{\theta}\right)\Pi_{s}\right)\times\\
\left(\sum_{\ell,m}\frac{\left(\lambda_{\ell}^{\frac{1}{2}}-\lambda_{m}^{\frac{1}{2}}\right)}{\left(\lambda_{\ell}+\lambda_{m}\right)\left(\lambda_{\ell}^{\frac{1}{2}}+\lambda_{m}^{\frac{1}{2}}\right)}\,\Pi_{\ell}\left(\partial_{\theta}\rho_{\theta}\right)\Pi_{m}\right)
\end{array}\right].\label{eq:SLD-fisher-purify-proof-1}
\end{align}
Furthermore,
\begin{align}
\Tr\!\left[\left(\partial_{\theta}\sqrt{\rho_{\theta}}\right)^{2}\right] & =\Tr\!\left[\left(\sum_{\ell,m}\frac{1}{\lambda_{\ell}^{\frac{1}{2}}+\lambda_{m}^{\frac{1}{2}}}\Pi_{\ell}\left(\partial_{\theta}\rho_{\theta}\right)\Pi_{m}\right)^{2}\right]\\
 & =\sum_{\ell,m,r,s}\left(\frac{1}{\lambda_{\ell}^{\frac{1}{2}}+\lambda_{m}^{\frac{1}{2}}}\right)\left(\frac{1}{\lambda_{r}^{\frac{1}{2}}+\lambda_{s}^{\frac{1}{2}}}\right)\Tr\!\left[\Pi_{\ell}\left(\partial_{\theta}\rho_{\theta}\right)\Pi_{m}\Pi_{r}\left(\partial_{\theta}\rho_{\theta}\right)\Pi_{s}\right]\\
 & =\sum_{\ell,m}\frac{1}{\left(\lambda_{\ell}^{\frac{1}{2}}+\lambda_{m}^{\frac{1}{2}}\right)^{2}}\Tr\!\left[\Pi_{\ell}\left(\partial_{\theta}\rho_{\theta}\right)\Pi_{m}\left(\partial_{\theta}\rho_{\theta}\right)\right].\label{eq:SLD-fisher-purify-proof-2}
\end{align}
and
\begin{align}
 & -\frac{i}{2}\Tr\!\left[\left[\left(\partial_{\theta}\sqrt{\rho_{\theta}}\right),\sqrt{\rho_{\theta}}\right]H_{\theta}^{\star}\right]\nonumber \\
 & =-\frac{1}{2}\Tr\!\left[\left(\sum_{r,s}\left(\frac{\lambda_{s}^{\frac{1}{2}}-\lambda_{r}^{\frac{1}{2}}}{\lambda_{r}^{\frac{1}{2}}+\lambda_{s}^{\frac{1}{2}}}\right)\Pi_{r}\left(\partial_{\theta}\rho_{\theta}\right)\Pi_{s}\right)\left(\sum_{\ell,m}\frac{\left(\lambda_{\ell}^{\frac{1}{2}}-\lambda_{m}^{\frac{1}{2}}\right)}{\left(\lambda_{\ell}+\lambda_{m}\right)\left(\lambda_{\ell}^{\frac{1}{2}}+\lambda_{m}^{\frac{1}{2}}\right)}\,\Pi_{\ell}\left(\partial_{\theta}\rho_{\theta}\right)\Pi_{m}\right)\right]\\
 & =-\frac{1}{2}\sum_{r,s,\ell,m}\frac{\left(\lambda_{s}^{\frac{1}{2}}-\lambda_{r}^{\frac{1}{2}}\right)\left(\lambda_{\ell}^{\frac{1}{2}}-\lambda_{m}^{\frac{1}{2}}\right)}{\left(\lambda_{r}^{\frac{1}{2}}+\lambda_{s}^{\frac{1}{2}}\right)\left(\lambda_{\ell}+\lambda_{m}\right)\left(\lambda_{\ell}^{\frac{1}{2}}+\lambda_{m}^{\frac{1}{2}}\right)}\Tr\!\left[\Pi_{r}\left(\partial_{\theta}\rho_{\theta}\right)\Pi_{s}\Pi_{\ell}\left(\partial_{\theta}\rho_{\theta}\right)\Pi_{m}\right]\\
 & =-\frac{1}{2}\sum_{\ell,m}\frac{\left(\lambda_{\ell}^{\frac{1}{2}}-\lambda_{m}^{\frac{1}{2}}\right)^{2}}{\left(\lambda_{\ell}^{\frac{1}{2}}+\lambda_{m}^{\frac{1}{2}}\right)^{2}\left(\lambda_{\ell}+\lambda_{m}\right)}\Tr\!\left[\Pi_{\ell}\left(\partial_{\theta}\rho_{\theta}\right)\Pi_{m}\left(\partial_{\theta}\rho_{\theta}\right)\right].\label{eq:SLD-fisher-purify-proof-3}
\end{align}
Finally, we combine~\eqref{eq:SLD-fisher-purify-proof-1},~\eqref{eq:SLD-fisher-purify-proof-2},
and~\eqref{eq:SLD-fisher-purify-proof-3} to conclude that
\begin{align}
 & \Tr\!\left[\left(\partial_{\theta}\sqrt{\rho_{\theta}}\right)^{2}\right]-\frac{i}{2}\Tr\!\left[\left[\left(\partial_{\theta}\sqrt{\rho_{\theta}}\right),\sqrt{\rho_{\theta}}\right]H_{\theta}^{\star}\right]\\
 & =\sum_{\ell,m}\left(\frac{1}{\left(\lambda_{\ell}^{\frac{1}{2}}+\lambda_{m}^{\frac{1}{2}}\right)^{2}}-\frac{\left(\lambda_{\ell}^{\frac{1}{2}}-\lambda_{m}^{\frac{1}{2}}\right)^{2}}{2\left(\lambda_{\ell}^{\frac{1}{2}}+\lambda_{m}^{\frac{1}{2}}\right)^{2}\left(\lambda_{\ell}+\lambda_{m}\right)}\right)\Tr\!\left[\Pi_{\ell}\left(\partial_{\theta}\rho_{\theta}\right)\Pi_{m}\left(\partial_{\theta}\rho_{\theta}\right)\right]\\
 & =\sum_{\ell,m}\frac{1}{\left(\lambda_{\ell}^{\frac{1}{2}}+\lambda_{m}^{\frac{1}{2}}\right)^{2}}\left(1-\frac{\left(\lambda_{\ell}^{\frac{1}{2}}-\lambda_{m}^{\frac{1}{2}}\right)^{2}}{2\left(\lambda_{\ell}+\lambda_{m}\right)}\right)\Tr\!\left[\Pi_{\ell}\left(\partial_{\theta}\rho_{\theta}\right)\Pi_{m}\left(\partial_{\theta}\rho_{\theta}\right)\right].
\end{align}
Now consider that, for $x,y>0$,
\begin{align}
\frac{1}{\left(x^{\frac{1}{2}}+y^{\frac{1}{2}}\right)^{2}}\left(1-\frac{\left(x^{\frac{1}{2}}-y^{\frac{1}{2}}\right)^{2}}{2\left(x+y\right)}\right) & =\frac{1}{\left(x^{\frac{1}{2}}+y^{\frac{1}{2}}\right)^{2}}\left(\frac{2\left(x+y\right)-\left(x^{\frac{1}{2}}-y^{\frac{1}{2}}\right)^{2}}{2\left(x+y\right)}\right)\\
 & =\frac{1}{x+y+2\sqrt{xy}}\left(\frac{2\left(x+y\right)-\left(x+y-2\sqrt{xy}\right)}{2\left(x+y\right)}\right)\\
 & =\frac{1}{x+y+2\sqrt{xy}}\left(\frac{x+y+2\sqrt{xy}}{2\left(x+y\right)}\right)\\
 & =\frac{1}{2\left(x+y\right)},
\end{align}
implying that
\begin{align}
 & \sum_{\ell,m}\left(\frac{1}{\left(\lambda_{\ell}^{\frac{1}{2}}+\lambda_{m}^{\frac{1}{2}}\right)^{2}}\left(1-\frac{\left(\lambda_{\ell}^{\frac{1}{2}}-\lambda_{m}^{\frac{1}{2}}\right)^{2}}{2\left(\lambda_{\ell}+\lambda_{m}\right)}\right)\right)\Tr\!\left[\Pi_{\ell}\left(\partial_{\theta}\rho_{\theta}\right)\Pi_{m}\left(\partial_{\theta}\rho_{\theta}\right)\right]\nonumber \\
 & =\sum_{\ell,m}\frac{1}{2\left(\lambda_{\ell}+\lambda_{m}\right)}\Tr\!\left[\Pi_{\ell}\left(\partial_{\theta}\rho_{\theta}\right)\Pi_{m}\left(\partial_{\theta}\rho_{\theta}\right)\right]\\
 & =\frac{1}{4}\sum_{\ell,m}\frac{2}{\lambda_{\ell}+\lambda_{m}}\Tr\!\left[\Pi_{\ell}\left(\partial_{\theta}\rho_{\theta}\right)\Pi_{m}\left(\partial_{\theta}\rho_{\theta}\right)\right]\\
 & =\frac{1}{4}I_{F}\!\left(\theta;\left(\rho_{\theta}\right)_{\theta}\right),
\end{align}
thus concluding the proof.
\end{proof}

\subsection{SLD Fisher information of channels}

The SLD Fisher information of a smooth family $\left(\mathcal{N}_{\theta}\right)_{\theta\in\Theta}$
of quantum channels is defined as follows~\cite{Fujiwara2001}:
\begin{align}
I_{F}\!\left(\theta;\left(\mathcal{N}_{\theta}\right)_{\theta\in\Theta}\right) & \coloneqq\sup_{\rho_{RA}\in\mathbb{D}}I_{F}\!\left(\theta;\left(\left(\id_{R}\otimes\mathcal{N}_{\theta}\right)(\rho_{RA})\right)_{\theta\in\Theta}\right).
\end{align}

Let us suppose that the channel $\mathcal{N}_{\theta}$ has a Kraus
decomposition as $\mathcal{N}_{\theta}(\rho)=\sum_{i}K_{i}^{\theta}\rho\left(K_{i}^{\theta}\right)^{\dag}$,
where $\left(K_{i}^{\theta}\right)_{\theta\in\Theta}$ is a smooth
family of Kraus operators for each $i$, and let us further suppose
that the number of Kraus operators for $\mathcal{N}_{\theta}$ is
fixed for all $\theta\in\Theta$. Then we take a canonical isometric
extension of the channel $\mathcal{N}_{\theta}$ as follows:
\begin{equation}
V_{\theta}\coloneqq\sum_{i}|i\rangle\otimes K_{i}^{\theta},\label{eq:canonical-isometric-ext-smooth-family}
\end{equation}
where $V_{\theta}\colon\mathbb{C}^{d_{A}}\mapsto\mathbb{C}^{d_{E}}\otimes\mathbb{C}^{d_{B}}$.

The following theorem is strongly related to~\cite[Theorem~4]{Fujiwara2008},
with a key difference being that the formulas for SLD Fisher information
of channels are expressed below in terms of the smooth family $\left(V_{\theta}\right)_{\theta\in\Theta}$
of canonical isometric extensions of the channels in the family $\left(\mathcal{N}_{\theta}\right)_{\theta\in\Theta}$.
\begin{thm}
\label{thm:SLD-Fisher-channels}The SLD Fisher information of a smooth
family $\left(\mathcal{N}_{\theta}\right)_{\theta\in\Theta}$ of quantum
channels can be written as follows:
\begin{align}
\frac{1}{4}I_{F}\!\left(\theta;\left(\mathcal{N}_{\theta}\right)_{\theta\in\Theta}\right) & =\sup_{|\psi\rangle_{RA}\in\mathbb{P}}\inf_{\left(U_{\theta}\right)_{\theta\in\Theta}}\left\Vert \partial_{\theta}\left[\left(U_{\theta}\otimes I_{RB}\right)V_{\theta}|\psi\rangle_{RA}\right]\right\Vert ^{2}\label{eq:SLD-channels-formula-1}\\
 & =\inf_{\left(H_{\theta}\right)_{\theta\in\Theta}}\left\Vert \left(\partial_{\theta}V_{\theta}\right)-i\left(H_{\theta}\otimes I_{B}\right)V_{\theta}\right\Vert ^{2}\label{eq:SLD-Fisher-channels-opt-Herm}\\
 & =\inf_{\left(U_{\theta}\right)_{\theta\in\Theta}}\left\Vert \partial_{\theta}\left[\left(U_{\theta}\otimes I_{B}\right)V_{\theta}\right]\right\Vert ^{2},\label{eq:SLD-channels-final-exp-thm}
\end{align}
where $V_{\theta}\colon\mathbb{C}^{d_{A}}\mapsto\mathbb{C}^{d_{E}}\otimes\mathbb{C}^{d_{B}}$
is a canonical isometric extension of $\mathcal{N}_{\theta}$, as
defined in~\eqref{eq:canonical-isometric-ext-smooth-family}, and
the optimizations in~\eqref{eq:SLD-Fisher-channels-opt-Herm} and
\eqref{eq:SLD-channels-final-exp-thm} are over every smooth family
$\left(H_{\theta}\right)_{\theta\in\Theta}$ of Hermitian operators
and every smooth family $\left(U_{\theta}\right)_{\theta\in\Theta}$
of unitary operators, respectively.
\end{thm}

\begin{proof}
Let us begin by noting that
\begin{equation}
I_{F}\!\left(\theta;\left(\mathcal{N}_{\theta}\right)_{\theta\in\Theta}\right)=\sup_{|\psi\rangle_{RA}\in\mathbb{P}}I_{F}\!\left(\theta;\left(\left(\id_{R}\otimes\mathcal{N}_{\theta}\right)(\psi_{RA})\right)_{\theta\in\Theta}\right),\label{eq:SLD-fisher-channels-pure-state-opt}
\end{equation}
which follows from the data-processing inequality for the SLD Fisher
information of states, purification, and the Schmidt decomposition
theorem. Furthermore, in~\eqref{eq:SLD-fisher-channels-pure-state-opt},
the reference system $R$ and the channel input system $A$ are isomorphic.
(See also~\cite[Remark~10]{Katariya2021} in this context.) Consider
that we can write an arbitrary state vector $|\psi\rangle_{RA}$ in
terms of a unitary $U$ and a density operator $\rho$ as
\begin{equation}
|\psi\rangle_{RA}=\left(U_{R}\otimes\sqrt{\rho}\right)|\Gamma\rangle.
\end{equation}
Then
\begin{equation}
\left(U_{\theta}\otimes I_{RB}\right)V_{\theta}|\psi\rangle_{RA}=\sum_{i}U_{\theta}|i\rangle\otimes\left(U_{R}\otimes K_{i}^{\theta}\sqrt{\rho}\right)|\Gamma\rangle
\end{equation}
and, by appealing to Theorem~\ref{thm:SLD-Fisher-states}, we can
write
\begin{align}
\frac{1}{4}I_{F}\!\left(\theta;\left(\mathcal{N}_{\theta}\right)_{\theta\in\Theta}\right)= & \sup_{|\psi\rangle_{RA}\in\mathbb{P}}\inf_{\left(U_{\theta}\right)_{\theta\in\Theta}}\left\Vert \partial_{\theta}\left[\left(U_{\theta}\otimes I_{RB}\right)V_{\theta}|\psi\rangle_{RA}\right]\right\Vert ^{2}.\label{eq:rewrite-SLD-Fisher-channels-purify}
\end{align}
This proves~\eqref{eq:SLD-channels-formula-1}.

Now let us work with the objective function in~\eqref{eq:rewrite-SLD-Fisher-channels-purify}.
Defining
\begin{equation}
|\psi^{\rho}\rangle_{RA}\coloneqq\left(I\otimes\sqrt{\rho}\right)|\Gamma\rangle,
\end{equation}
consider that
\begin{align}
 & \left\Vert \partial_{\theta}\left[\left(U_{\theta}\otimes I_{RB}\right)V_{\theta}|\psi\rangle_{RA}\right]\right\Vert ^{2}\nonumber \\
 & =\left\Vert \partial_{\theta}\left[\left(U_{\theta}\otimes I_{RB}\right)V_{\theta}|\psi^{\rho}\rangle_{RA}\right]\right\Vert ^{2}\\
 & =\left\Vert \left[\left(\left(\partial_{\theta}U_{\theta}\right)\otimes I_{RB}\right)V_{\theta}+\left(U_{\theta}\otimes I_{RB}\right)\left(\partial_{\theta}V_{\theta}\right)\right]|\psi^{\rho}\rangle_{RA}\right\Vert ^{2}\\
 & =\left\Vert U_{\theta}^{\dag}\left[\left(\left(\partial_{\theta}U_{\theta}\right)\otimes I_{RB}\right)V_{\theta}+\left(U_{\theta}\otimes I_{RB}\right)\left(\partial_{\theta}V_{\theta}\right)\right]|\psi^{\rho}\rangle_{RA}\right\Vert ^{2}\\
 & =\left\Vert \left[\left(U_{\theta}^{\dag}\left(\partial_{\theta}U_{\theta}\right)\otimes I_{RB}\right)V_{\theta}+\left(U_{\theta}^{\dag}U_{\theta}\otimes I_{RB}\right)\left(\partial_{\theta}V_{\theta}\right)\right]|\psi^{\rho}\rangle_{RA}\right\Vert ^{2}\\
 & =\left\Vert \left[\partial_{\theta}V_{\theta}+\left(U_{\theta}^{\dag}\left(\partial_{\theta}U_{\theta}\right)\otimes I_{RB}\right)V_{\theta}\right]|\psi^{\rho}\rangle_{RA}\right\Vert ^{2}\\
 & =\left\Vert \left[\partial_{\theta}V_{\theta}-i\left(H_{\theta}\otimes I_{RB}\right)V_{\theta}\right]|\psi^{\rho}\rangle_{RA}\right\Vert ^{2}\label{eq:SLD-channels-proof-herm-introduce}\\
 & =\langle\psi^{\rho}|_{RA}\left[\partial_{\theta}V_{\theta}^{\dag}+iV_{\theta}^{\dag}\left(H_{\theta}\otimes I_{RB}\right)\right]\left[\partial_{\theta}V_{\theta}-i\left(H_{\theta}\otimes I_{RB}\right)V_{\theta}\right]|\psi^{\rho}\rangle_{RA}\\
 & =\Tr\!\left[\left[\partial_{\theta}V_{\theta}^{\dag}+iV_{\theta}^{\dag}\left(H_{\theta}\otimes I_{RB}\right)\right]\left[\partial_{\theta}V_{\theta}-i\left(H_{\theta}\otimes I_{RB}\right)V_{\theta}\right]\rho\right].\label{eq:SLD-channels-linear-in-rho}
\end{align}
To establish the equality in~\eqref{eq:SLD-channels-proof-herm-introduce},
we invoked the same reasoning around~\eqref{eq:hermitian-op-from-unitary-1}--\eqref{eq:hermitian-op-from-unitary-2}.
Now observe that we have written the objective function in two different
ways in~\eqref{eq:SLD-channels-proof-herm-introduce} and~\eqref{eq:SLD-channels-linear-in-rho}
to see that it is is convex in $H_{\theta}$ and linear in $\rho$.
Indeed, the function
\begin{equation}
H_{\theta}\mapsto\left\Vert \left[\partial_{\theta}V_{\theta}-i\left(H_{\theta}\otimes I_{RB}\right)V_{\theta}\right]|\psi\rangle_{RA}\right\Vert ^{2}
\end{equation}
is convex in $H_{\theta}$, being a composition of a convex function
(the norm) and a convex and non-decreasing function (the square function
$x\mapsto x^{2}$), and the function
\begin{equation}
\rho\mapsto\Tr\!\left[\left[\partial_{\theta}V_{\theta}^{\dag}+iV_{\theta}^{\dag}\left(H_{\theta}\otimes I_{RB}\right)\right]\left[\partial_{\theta}V_{\theta}-i\left(H_{\theta}\otimes I_{RB}\right)V_{\theta}\right]\rho\right]
\end{equation}
is linear in $\rho$. Then we can apply the standard minimax theorem
because the set of density operators is convex and compact and the
set of Hermitian operators is convex. Thus,
\begin{align}
 & \sup_{|\psi\rangle_{RA}\in\mathbb{P}}\inf_{\left(U_{\theta}\right)_{\theta\in\Theta}}\left\Vert \partial_{\theta}\left[\left(U_{\theta}\otimes I_{RB}\right)V_{\theta}|\psi\rangle_{RA}\right]\right\Vert ^{2}\nonumber \\
 & =\sup_{\rho\in\mathbb{D}}\inf_{\left(H_{\theta}\right)_{\theta\in\Theta}}\left\Vert \left[\partial_{\theta}V_{\theta}-i\left(H_{\theta}\otimes I_{RB}\right)V_{\theta}\right]|\psi^{\rho}\rangle_{RA}\right\Vert ^{2}\\
 & =\inf_{\left(H_{\theta}\right)_{\theta\in\Theta}}\sup_{\rho\in\mathbb{D}}\left\Vert \left[\partial_{\theta}V_{\theta}-i\left(H_{\theta}\otimes I_{RB}\right)V_{\theta}\right]|\psi^{\rho}\rangle_{RA}\right\Vert ^{2}\\
 & =\inf_{\left(H_{\theta}\right)_{\theta\in\Theta}}\sup_{|\psi\rangle_{RA}\in\mathbb{P}}\left\Vert \left[\partial_{\theta}V_{\theta}-i\left(H_{\theta}\otimes I_{RB}\right)V_{\theta}\right]|\psi\rangle_{RA}\right\Vert ^{2}\\
 & =\inf_{\left(H_{\theta}\right)_{\theta\in\Theta}}\left\Vert \partial_{\theta}V_{\theta}-i\left(H_{\theta}\otimes I_{RB}\right)V_{\theta}\right\Vert ^{2}\\
 & =\inf_{\left(H_{\theta}\right)_{\theta\in\Theta}}\left\Vert \partial_{\theta}V_{\theta}-i\left(H_{\theta}\otimes I_{B}\right)V_{\theta}\right\Vert ^{2}.
\end{align}
This proves~\eqref{eq:SLD-Fisher-channels-opt-Herm}.

The final expression in~\eqref{eq:SLD-channels-final-exp-thm} follows
by noting that
\begin{equation}
\inf_{\left(H_{\theta}\right)_{\theta\in\Theta}}\left\Vert \partial_{\theta}V_{\theta}-i\left(H_{\theta}\otimes I_{B}\right)V_{\theta}\right\Vert ^{2}=\inf_{\left(U_{\theta}\right)_{\theta\in\Theta}}\left\Vert \partial_{\theta}\left[\left(U_{\theta}\otimes I_{B}\right)V_{\theta}\right]\right\Vert ^{2},
\end{equation}
again as a consequence of the reasoning around~\eqref{eq:hermitian-op-from-unitary-1}--\eqref{eq:hermitian-op-from-unitary-2}.
\end{proof}
\begin{rem}
\label{rem:simplify-overlap-term-herm}Observe that the overlap term
$\left\Vert \left(\left(U_{\theta}\otimes I_{B}\right)V_{\theta}\right)^{\dag}\partial_{\theta}\left[\left(U_{\theta}\otimes I_{B}\right)V_{\theta}\right]\right\Vert $
simplifies as follows:
\begin{align}
 & \left\Vert \left(\left(U_{\theta}\otimes I_{B}\right)V_{\theta}\right)^{\dag}\partial_{\theta}\left[\left(U_{\theta}\otimes I_{B}\right)V_{\theta}\right]\right\Vert \nonumber \\
 & =\left\Vert V_{\theta}^{\dag}\left(U_{\theta}^{\dag}\otimes I_{B}\right)\partial_{\theta}\left[\left(U_{\theta}\otimes I_{B}\right)V_{\theta}\right]\right\Vert \\
 & =\left\Vert V_{\theta}^{\dag}\left(U_{\theta}^{\dag}\otimes I_{B}\right)\left[\left(\left(\partial_{\theta}U_{\theta}\right)\otimes I_{B}\right)V_{\theta}+\left(U_{\theta}\otimes I_{B}\right)\left(\partial_{\theta}V_{\theta}\right)\right]\right\Vert \\
 & =\left\Vert V_{\theta}^{\dag}\left[\left(U_{\theta}^{\dag}\left(\partial_{\theta}U_{\theta}\right)\otimes I_{B}\right)V_{\theta}+\left(\left(U_{\theta}\right)^{\dag}\otimes I_{B}\right)\left(U_{\theta}\otimes I_{B}\right)\left(\partial_{\theta}V_{\theta}\right)\right]\right\Vert \\
 & =\left\Vert V_{\theta}^{\dag}\left[\left(U_{\theta}^{\dag}\left(\partial_{\theta}U_{\theta}\right)\otimes I_{B}\right)V_{\theta}+\partial_{\theta}V_{\theta}\right]\right\Vert \\
 & =\left\Vert V_{\theta}^{\dag}\left(\partial_{\theta}V_{\theta}-i\left(H_{\theta}\otimes I_{B}\right)V_{\theta}\right)\right\Vert ,\label{eq:overlap-term-SLD-Fisher-channels}
\end{align}
where again we define $H_{\theta}\coloneqq iU_{\theta}^{\dag}\left(\partial_{\theta}U_{\theta}\right)$
(see~\eqref{eq:hermitian-op-from-unitary-1}--\eqref{eq:hermitian-op-from-unitary-2}).
The expression in~\eqref{eq:overlap-term-SLD-Fisher-channels} finds
application in bounds on query complexity of channel estimation.
\end{rem}

\medskip{}
\begin{rem}
Relevant for our applications in channel estimation, we can alternatively
express the objective function in~\eqref{eq:SLD-Fisher-channels-opt-Herm}
and the overlap term in~\eqref{eq:overlap-term-SLD-Fisher-channels}
in terms of Kraus operators $\left(K_{i}^{\theta}\right)_{i}$ of
$\mathcal{N}_{\theta}$ as follows, by making use of~\eqref{eq:canonical-isometric-ext-smooth-family}.
This connects to the notation used in~\cite[Section~4.1]{Fujiwara2008}
and in several subsequent papers on channel estimation~\cite{DemkowiczDobrzanski2012,Kolodynski2013,Demkowicz2014,Zhou2021,Kurdzialek2023}.
To begin with, consider that
\begin{align}
\left(U_{\theta}\otimes I_{B}\right)V_{\theta} & =\left(U_{\theta}\otimes I_{B}\right)\left(\sum_{j}|j\rangle\otimes K_{j}^{\theta}\right)\\
 & =\sum_{j}U_{\theta}|j\rangle\otimes K_{j}^{\theta}\\
 & =\sum_{j}\sum_{i}|i\rangle\!\langle i|U_{\theta}|j\rangle\otimes K_{j}^{\theta}\\
 & =\sum_{j}\sum_{i}|i\rangle\otimes u_{i,j}^{\theta}K_{j}^{\theta}\\
 & =\sum_{i}|i\rangle\otimes\sum_{j}u_{i,j}^{\theta}K_{j}^{\theta}\\
 & =\sum_{i}|i\rangle\otimes\widetilde{K}_{i}^{\theta},
\end{align}
where we defined $u_{i,j}^{\theta}\coloneqq\langle i|U_{\theta}|j\rangle$
and the Kraus operators $\widetilde{K}_{i}^{\theta}\coloneqq\sum_{j}u_{i,j}^{\theta}K_{j}^{\theta}$.
This implies that
\begin{equation}
\partial_{\theta}\left[\left(U_{\theta}\otimes I_{B}\right)V_{\theta}\right]=\sum_{i}|i\rangle\otimes\partial_{\theta}\widetilde{K}_{i}^{\theta}.
\end{equation}
Then the objective function in~\eqref{eq:SLD-Fisher-channels-opt-Herm}
simplifies as follows:
\begin{align}
\left\Vert \partial_{\theta}\left[\left(U_{\theta}\otimes I_{B}\right)V_{\theta}\right]\right\Vert ^{2} & =\left\Vert \sum_{i}|i\rangle\otimes\partial_{\theta}\widetilde{K}_{i}^{\theta}\right\Vert ^{2}\\
 & =\left\Vert \left(\sum_{j}|j\rangle\otimes\partial_{\theta}\widetilde{K}_{j}^{\theta}\right)^{\dag}\left(\sum_{i}|i\rangle\otimes\partial_{\theta}\widetilde{K}_{i}^{\theta}\right)\right\Vert \\
 & =\left\Vert \left(\sum_{j}\langle j|\otimes\left(\partial_{\theta}\widetilde{K}_{j}^{\theta}\right)^{\dag}\right)\left(\sum_{i}|i\rangle\otimes\partial_{\theta}\widetilde{K}_{i}^{\theta}\right)\right\Vert \\
 & =\left\Vert \sum_{j}\left(\partial_{\theta}\widetilde{K}_{j}^{\theta}\right)^{\dag}\partial_{\theta}\widetilde{K}_{j}^{\theta}\right\Vert ,
\end{align}
which is written in~\cite[Section~4.1]{Fujiwara2008} as $\left\Vert \alpha\right\Vert $,
where
\begin{equation}
\alpha\coloneqq\sum_{j}\left(\partial_{\theta}\widetilde{K}_{j}^{\theta}\right)^{\dag}\partial_{\theta}\widetilde{K}_{j}^{\theta}.
\end{equation}
Additionally, the overlap term in~\eqref{eq:overlap-term-SLD-Fisher-channels}
simplifies as follows:
\begin{align}
\left\Vert \left(\left(U_{\theta}\otimes I_{B}\right)V_{\theta}\right)^{\dag}\partial_{\theta}\left[\left(U_{\theta}\otimes I_{B}\right)V_{\theta}\right]\right\Vert  & =\left\Vert \left(\sum_{j}|j\rangle\otimes\widetilde{K}_{j}^{\theta}\right)^{\dag}\left(\sum_{i}|i\rangle\otimes\partial_{\theta}\widetilde{K}_{i}^{\theta}\right)\right\Vert \\
 & =\left\Vert \left(\sum_{j}\langle j|\otimes\left(\widetilde{K}_{j}^{\theta}\right)^{\dag}\right)\left(\sum_{i}|i\rangle\otimes\partial_{\theta}\widetilde{K}_{i}^{\theta}\right)\right\Vert \\
 & =\left\Vert \sum_{j}\left(\widetilde{K}_{j}^{\theta}\right)^{\dag}\partial_{\theta}\widetilde{K}_{j}^{\theta}\right\Vert ,
\end{align}
which is written in~\cite[Section~4.1]{Fujiwara2008} as $\left\Vert \beta\right\Vert $,
where
\begin{equation}
\beta\coloneqq\sum_{j}\left(\widetilde{K}_{j}^{\theta}\right)^{\dag}\partial_{\theta}\widetilde{K}_{j}^{\theta}.
\end{equation}

The following theorem was stated in~\cite[Eq.~(18)]{Yuan2017}; below
we provide a proof for it.
\end{rem}

\begin{thm}
\label{thm:smooth-family-SLD-Fisher-ch-expand}Let $\left(\mathcal{N}_{\theta}\right)_{\theta\in\Theta}$
be a smooth family of quantum channels. The following equalities hold:
\begin{align}
d_{B}^{2}\!\left(\mathcal{N}_{\theta},\mathcal{N}_{\theta+\delta}\right) & =\frac{\delta^{2}}{4}I_{F}\!\left(\theta;\left(\mathcal{N}_{\theta}\right)_{\theta}\right)+o(\delta^{2}),\\
I_{F}\!\left(\theta;\left(\mathcal{N}_{\theta}\right)_{\theta}\right) & =2\left.\frac{\partial^{2}}{\partial\delta^{2}}d_{B}^{2}(\mathcal{N}_{\theta},\mathcal{N}_{\theta+\delta})\right|_{\delta=0}\label{eq:SLD-Fisher-channels-2nd-deriv-Bures}\\
 & =4\lim_{\delta\to0}\frac{d_{B}^{2}(\mathcal{N}_{\theta},\mathcal{N}_{\theta+\delta})}{\delta^{2}}.\label{eq:SLD-Fisher-channels-limit-Bures}
\end{align}
\end{thm}

\begin{proof}
Recall~\eqref{eq:SLD-fisher-channels-pure-state-opt}. For a fixed
pure input state $\psi_{RA}$, applying~\eqref{eq:SLD-Fisher-to-Bures}
leads to
\begin{align}
d_{B}^{2}\!\left(\mathcal{N}_{\theta}\!\left(\psi_{RA}\right),\mathcal{N}_{\theta+\delta}\!\left(\psi_{RA}\right)\right) & =\frac{\delta^{2}}{4}I_{F}\!\left(\theta;\left(\mathcal{N}_{\theta}\!\left(\psi_{RA}\right)\right)_{\theta}\right)+o(\delta^{2}).\label{eq:uniform-taylor-expansion-Bures}
\end{align}
The convergence in the $o(\delta^{2})$ term is uniform in $\psi_{RA}$,
as a consequence of the following general fact: Let $K$ be compact
and let $f\colon K\times\left[-\varepsilon,\varepsilon\right]\to\mathbb{R}$
be such that $f(\psi,\cdot)$ is $C^{2}$ for each $\psi\in K$ and
$\left(\psi,\delta\right)\mapsto f''(\psi,\delta)$ is continuous.
Then
\begin{equation}
f(\psi,\delta)=f(\psi,0)+f'(\psi,0)\delta+\frac{1}{2}f''(\psi,0)\delta^{2}+o(\delta^{2})
\end{equation}
uniformly in $\psi\in K$. Applying this statement to our case, the
equality in~\eqref{eq:uniform-taylor-expansion-Bures} holds because
the set of pure states on $RA$ is compact. Then we conclude that
\begin{equation}
\sup_{\psi_{RA}\in\mathbb{P}}d_{B}^{2}\!\left(\mathcal{N}_{\theta}\!\left(\psi_{RA}\right),\mathcal{N}_{\theta+\delta}\!\left(\psi_{RA}\right)\right)=\frac{\delta^{2}}{4}\left[\sup_{\psi_{RA}\in\mathbb{P}}I_{F}\!\left(\theta;\left(\mathcal{N}_{\theta}\!\left(\psi_{RA}\right)\right)_{\theta}\right)\right]+o(\delta^{2}),
\end{equation}
given that the convergence is uniform and the supremum is taken over
a compact set of states, so that the maximizing $\psi_{RA}$ is smooth
in $\delta$ near $0$. This equality is equivalent to
\begin{equation}
d_{B}^{2}\!\left(\mathcal{N}_{\theta},\mathcal{N}_{\theta+\delta}\right)=\frac{\delta^{2}}{4}I_{F}\!\left(\theta;\left(\mathcal{N}_{\theta}\right)_{\theta}\right)+o(\delta^{2}).\label{eq:SLD-Fisher-channels-Bures-Taylor-proof}
\end{equation}
Then~\eqref{eq:SLD-Fisher-channels-2nd-deriv-Bures} and~\eqref{eq:SLD-Fisher-channels-limit-Bures}
are immediate consequences of~\eqref{eq:SLD-Fisher-channels-Bures-Taylor-proof}.
\end{proof}

\section{Upper bounds on Bures distance of channels}\label{sec:Upper-bounds-Bures}

In this section, we present upper bounds on the Bures distance of
channels $\mathcal{N}_{1}$ and $\mathcal{N}_{2}$ in both the parallel
and adaptive settings. The bounds are expressed in terms of a canonical
isometric extension $V_{i}\colon\mathbb{C}^{d_{A}}\mapsto\mathbb{C}^{d_{E}}\otimes\mathbb{C}^{d_{B}}$
of the channel $\mathcal{N}_{i}$, for $i\in\left\{ 1,2\right\} $,
as given in~\eqref{eq:isometric-ext-ch-1} and~\eqref{eq:isometric-ext-ch-2}. 

\subsection{Upper bound on Bures distance of channels in the parallel setting}

Ref.~\cite[Eq.~(23)]{Yuan2017} established an upper bound on the
squared Bures distance of the tensor-power channels $\mathcal{N}_{1}^{\otimes n}$
and $\mathcal{N}_{2}^{\otimes n}$. Here we restate this upper bound
in a different form, in terms of canonical isometric extensions of
the channels $\mathcal{N}_{1}$ and $\mathcal{N}_{2}$, and we provide
an alternate proof based on basic properties of isometries and the
spectral norm.
\begin{thm}
\label{thm:parallel-bures-dist-channels}For channels $\mathcal{N}_{1}$
and $\mathcal{N}_{2}$ with canonical isometric extensions $V_1$ and $V_2$, the following inequality holds for all $n\in\mathbb{N}$:
\begin{equation}
d_{B}^{2}\!\left(\mathcal{N}_{1}^{\otimes n},\mathcal{N}_{2}^{\otimes n}\right)\leq n\inf_{W\in\mathbb{B}}\left\{ 2\left\Vert I-\Re\!\left[M_{W}\right]\right\Vert +\left(n-1\right)\left\Vert I-M_{W}\right\Vert ^{2}\right\} ,\label{eq:bures-parallel-upper-bnd}
\end{equation}
where
\begin{equation}
M_{W}\equiv V_{1}^{\dag}\left(W_{E}\otimes I_{B}\right)V_{2}.
\end{equation}
\end{thm}

\begin{proof}
Let $V_{i}\colon\mathbb{C}^{d_{A}}\mapsto\mathbb{C}^{d_{E}}\otimes\mathbb{C}^{d_{B}}$
be a canonical isometric extension of $\mathcal{N}_{i}$, and set
\begin{align}
\hat{V}_{1} & \equiv|0\rangle\otimes V_{1},\label{eq:isometric-exts-ch-disc-extra-qubit-1}\\
\hat{V}_{2} & \equiv\left(U_{W}\otimes I\right)|0\rangle\otimes V_{2},\label{eq:isometric-exts-ch-disc-extra-qubit-2}
\end{align}
where $|0\rangle\in\mathbb{C}^{2}$ and $U_{W}$ is an arbitrary unitary
acting nontrivially on the first two systems of $\mathbb{C}^{2}\otimes\mathbb{C}^{d_{E}}\otimes\mathbb{C}^{d_{B}}$.
As such, we can apply Theorem~\ref{thm:Bures-distance-channels} to
$d_{B}^{2}\!\left(\mathcal{N}_{1}^{\otimes n},\mathcal{N}_{2}^{\otimes n}\right)$
to conclude that
\begin{equation}
d_{B}^{2}\!\left(\mathcal{N}_{1}^{\otimes n},\mathcal{N}_{2}^{\otimes n}\right)\leq\left\Vert \hat{V}_{1}^{\otimes n}-\hat{V}_{2}^{\otimes n}\right\Vert ^{2}.
\end{equation}
Continuing, we find that
\begin{align}
 & \left\Vert \hat{V}_{1}^{\otimes n}-\hat{V}_{2}^{\otimes n}\right\Vert ^{2}\nonumber \\
 & =\left\Vert \sum_{i=1}^{n}\hat{V}_{1}^{\otimes i-1}\otimes\left[\hat{V}_{1}-\hat{V}_{2}\right]_{i}\otimes\hat{V}_{2}^{\otimes n-i}\right\Vert ^{2}\\
 & =\left\Vert \left(\sum_{j=1}^{n}\hat{V}_{1}^{\otimes j-1}\otimes\left[\hat{V}_{1}-\hat{V}_{2}\right]_{j}\otimes\hat{V}_{2}^{\otimes n-j}\right)^{\dag}\left(\sum_{i=1}^{n}\hat{V}_{1}^{\otimes i-1}\otimes\left[\hat{V}_{1}-\hat{V}_{2}\right]_{i}\otimes\hat{V}_{2}^{\otimes n-i}\right)\right\Vert \\
 & =\left\Vert \left(\sum_{j=1}^{n}\left(\hat{V}_{1}^{\dag}\right)^{\otimes j-1}\otimes\left[\hat{V}_{1}^{\dag}-\hat{V}_{2}^{\dag}\right]_{j}\otimes\left(\hat{V}_{2}^{\dag}\right)^{\otimes n-j}\right)\left(\sum_{i=1}^{n}\hat{V}_{1}^{\otimes i-1}\otimes\left[\hat{V}_{1}-\hat{V}_{2}\right]_{i}\otimes\hat{V}_{2}^{\otimes n-i}\right)\right\Vert \\
 & =\left\Vert %\left(
 \sum_{i,j=1}^{n} \left( % jjm: moved parenthesis
 \left(\hat{V}_{1}^{\dag}\right)^{\otimes j-1}\otimes\left[\hat{V}_{1}^{\dag}-\hat{V}_{2}^{\dag}\right]_{j}\otimes\left(\hat{V}_{2}^{\dag}\right)^{\otimes n-j}\right)\left(\hat{V}_{1}^{\otimes i-1}\otimes\left[\hat{V}_{1}-\hat{V}_{2}\right]_{i}\otimes\hat{V}_{2}^{\otimes n-i}\right)\right\Vert \\
 & \leq\sum_{i=1}^{n}\left\Vert \left(\left(\hat{V}_{1}^{\dag}\right)^{\otimes i-1}\otimes\left[\hat{V}_{1}^{\dag}-\hat{V}_{2}^{\dag}\right]_{i}\otimes\left(\hat{V}_{2}^{\dag}\right)^{\otimes n-i}\right)\left(\hat{V}_{1}^{\otimes i-1}\otimes\left[\hat{V}_{1}-\hat{V}_{2}\right]_{i}\otimes\hat{V}_{2}^{\otimes n-i}\right)\right\Vert \nonumber \\
 & \qquad+\sum_{i,j=1:i<j}^{n}\left\Vert \left(\left(\hat{V}_{1}^{\dag}\right)^{\otimes j-1}\otimes\left[\hat{V}_{1}^{\dag}-\hat{V}_{2}^{\dag}\right]_{j}\otimes\left(\hat{V}_{2}^{\dag}\right)^{\otimes n-j}\right)\left(\hat{V}_{1}^{\otimes i-1}\otimes\left[\hat{V}_{1}-\hat{V}_{2}\right]_{i}\otimes\hat{V}_{2}^{\otimes n-i}\right)\right\Vert \nonumber \\
 & \qquad+\sum_{i,j=1:i>j}^{n}\left\Vert \left(\left(\hat{V}_{1}^{\dag}\right)^{\otimes j-1}\otimes\left[\hat{V}_{1}^{\dag}-\hat{V}_{2}^{\dag}\right]_{j}\otimes\left(\hat{V}_{2}^{\dag}\right)^{\otimes n-j}\right)\left(\hat{V}_{1}^{\otimes i-1}\otimes\left[\hat{V}_{1}-\hat{V}_{2}\right]_{i}\otimes\hat{V}_{2}^{\otimes n-i}\right)\right\Vert \\
 & =\sum_{i=1}^{n}\left\Vert \left[\hat{V}_{1}^{\dag}-\hat{V}_{2}^{\dag}\right]\left[\hat{V}_{1}-\hat{V}_{2}\right]\right\Vert +\sum_{i,j=1:i<j}^{n}\left\Vert \left[\hat{V}_{1}^{\dag}\right]_{i}\left[\hat{V}_{1}-\hat{V}_{2}\right]_{i}\right\Vert \left\Vert \left[\hat{V}_{1}^{\dag}-\hat{V}_{2}^{\dag}\right]_{j}\left[\hat{V}_{2}\right]_{j}\right\Vert \nonumber \\
 & \qquad+\sum_{i,j=1:i>j}^{n}\left\Vert \left[\hat{V}_{1}^{\dag}-\hat{V}_{2}^{\dag}\right]_{j}\left[\hat{V}_{1}\right]_{j}\right\Vert \left\Vert \left[\hat{V}_{2}^{\dag}\right]_{i}\left[\hat{V}_{1}-\hat{V}_{2}\right]_{i}\right\Vert \\
 & =n\left\Vert \hat{V}_{1}-\hat{V}_{2}\right\Vert ^{2}+\sum_{i,j=1:i<j}^{n}\left\Vert \hat{V}_{1}^{\dag}\left(\hat{V}_{1}-\hat{V}_{2}\right)\right\Vert \left\Vert \left(\hat{V}_{1}^{\dag}-\hat{V}_{2}^{\dag}\right)\hat{V}_{2}\right\Vert \nonumber \\
 & \qquad+\sum_{i,j=1:i>j}^{n}\left\Vert \left(\hat{V}_{1}^{\dag}-\hat{V}_{2}^{\dag}\right)\hat{V}_{1}\right\Vert \left\Vert \hat{V}_{2}^{\dag}\left(\hat{V}_{1}-\hat{V}_{2}\right)\right\Vert \\
 & =n\left\Vert \hat{V}_{1}-\hat{V}_{2}\right\Vert ^{2}+\frac{n\left(n-1\right)}{2}\left\Vert I-\hat{V}_{1}^{\dag}\hat{V}_{2}\right\Vert \left\Vert \hat{V}_{1}^{\dag}\hat{V}_{2}-I\right\Vert +\frac{n\left(n-1\right)}{2}\left\Vert I-\hat{V}_{2}^{\dag}\hat{V}_{1}\right\Vert \left\Vert \hat{V}_{2}^{\dag}\hat{V}_{1}-I\right\Vert \\
 & =n\left\Vert \hat{V}_{1}-\hat{V}_{2}\right\Vert ^{2}+n\left(n-1\right)\left\Vert I-\hat{V}_{1}^{\dag}\hat{V}_{2}\right\Vert ^{2}\\
 & =n\left(2\left\Vert I-\Re\!\left[V_{1}^{\dag}\left(W_{E}\otimes I_{B}\right)V_{2}\right]\right\Vert +\left(n-1\right)\left\Vert I-V_{1}^{\dag}\left(W_{E}\otimes I_{B}\right)V_{2}\right\Vert ^{2}\right).
\end{align}
The last equality follows by expanding the first norm as
\begin{align}
\left\Vert \hat{V}_{1}-\hat{V}_{2}\right\Vert ^{2} & =\left\Vert \left(\hat{V}_{1}-\hat{V}_{2}\right)^{\dag}\left(\hat{V}_{1}-\hat{V}_{2}\right)\right\Vert \label{eq:norm-iso-exts-parallel-disc-simplify}\\
 & =\left\Vert 2I-\hat{V}_{2}^{\dag}\hat{V}_{1}-\hat{V}_{1}^{\dag}\hat{V}_{2}\right\Vert \\
 & =2\left\Vert I-\Re\!\left[\hat{V}_{1}^{\dag}\hat{V}_{2}\right]\right\Vert ,
\end{align}
and observing that
\begin{align}
\hat{V}_{1}^{\dag}\hat{V}_{2} & =\left(\langle0|\otimes V_{1}^{\dag}\right)\left(U_{W}\otimes I\right)\left(|0\rangle\otimes V_{2}\right)\\
 & =V_{1}^{\dag}\left[\left(\langle0|\otimes I\right)U_{W}\left(|0\rangle\otimes I\right)\otimes I_{B}\right]V_{2}\\
 & =V_{1}^{\dag}\left[W_{E}\otimes I_{B}\right]V_{2}.\label{eq:overlap-iso-exts-parallel-disc-simplify}
\end{align}
Given that $W_{E}$ is an arbitrary contraction, we can obtain the
tightest upper bound using this approach by minimizing over all such
contractions, thus completing the proof.
\end{proof}
\begin{rem}
To see that the upper bound in~\eqref{eq:bures-parallel-upper-bnd}
is equivalent to~\cite[Eq.~(23)]{Yuan2017}, we appeal to Remark~\ref{rem:isometric-overlap-to-kraus}
and define $K_{W}\coloneqq\sum_{i,j}W_{i,j}K_{i}^{\dag}L_{j}$. Then
\begin{equation}
2\left\Vert I-\Re\!\left[M_{W}\right]\right\Vert +\left(n-1\right)\left\Vert I-M_{W}\right\Vert ^{2}=\left\Vert 2I-K_{W}-K_{W}^{\dag}\right\Vert +\left(n-1\right)\left\Vert I-K_{W}\right\Vert ^{2},
\end{equation}
thus establishing the equivalence.
\end{rem}

\subsection{Upper bound on Bures distance of channels in the adaptive setting}

For channels $\mathcal{N}_{1}$ and $\mathcal{N}_{2}$, a general
adaptive $n$-query protocol consists of a tuple $\left(\mathcal{A}_{i}\right)_{i=1}^{n-1}$
of channels. Applying the protocol to $n$ queries of the channel
$\mathcal{N}_{i}$ leads to the following channel $\mathcal{P}_{i}^{(n)}$,
for $i\in\left\{ 1,2\right\} $:
\begin{equation}
\mathcal{P}_{i}^{(n)}\coloneqq\left(\id\otimes\mathcal{N}_{i}\right)\circ\mathcal{A}_{n-1}\circ\cdots\circ\mathcal{A}_{2}\circ\left(\id\otimes\mathcal{N}_{i}\right)\circ\mathcal{A}_{1}\circ\left(\id\otimes\mathcal{N}_{i}\right).\label{eq:adaptive-protocol-bures-distance}
\end{equation}
This in turn leads to the following distinguishability measure for
channels $\mathcal{N}_{1}$ and $\mathcal{N}_{2}$:
\begin{equation}
\sup_{\left(\mathcal{A}_{i}\right)_{i=1}^{n-1}}d_{B}^{2}\!\left(\mathcal{P}_{1}^{(n)},\mathcal{P}_{2}^{(n)}\right),
\end{equation}
which is related to the strategy fidelity of~\cite{Gutoski2018} and
is a special case of the generalized strategy divergence introduced
in~\cite[Definition~1]{Wang2019}.

In Theorem~\ref{thm:seq-bures-distance-channels} below, we establish
an upper bound on the $n$-query adaptive squared Bures distance of
channels $\mathcal{N}_{1}$ and $\mathcal{N}_{2}$:
\begin{thm}
\label{thm:seq-bures-distance-channels}For channels $\mathcal{N}_{1}$
and $\mathcal{N}_{2}$ with canonical isometric extensions $V_1$ and $V_2$, the following inequality holds for all $n\in\mathbb{N}$:
\begin{multline}
\sup_{\left(\mathcal{A}_{i}\right)_{i=1}^{n-1}}d_{B}^{2}\!\left(\mathcal{P}_{1}^{(n)},\mathcal{P}_{2}^{(n)}\right)\leq\\
n\inf_{W\in\mathbb{B}}\left\{ 2\left\Vert I-\Re\!\left[M_{W}\right]\right\Vert +\left(n-1\right)\left(2\left\Vert I-\Re\!\left[M_{W}\right]\right\Vert \right)^{\frac{1}{2}}\left\Vert I-M_{W}\right\Vert \right\} ,\label{eq:bures-distance-adaptive-setting}
\end{multline}
where
\begin{equation}
M_{W}\equiv V_{1}^{\dag}\left(W_{E}\otimes I_{B}\right)V_{2}.
\end{equation}
\end{thm}

\begin{proof}
Let us begin by observing that the channel $\mathcal{P}_{i}^{(n)}$
in~\eqref{eq:adaptive-protocol-bures-distance} has an isometric extension
of the following form:
\begin{equation}
\left(I\otimes\hat{V}_{i}\right)A_{n-1}\cdots A_{2}\left(I\otimes\hat{V}_{i}\right)A_{1}\left(I\otimes\hat{V}_{i}\right),
\end{equation}
where $\hat{V}_{i}$ is chosen as in~\eqref{eq:isometric-exts-ch-disc-extra-qubit-1}--\eqref{eq:isometric-exts-ch-disc-extra-qubit-2}
and $A_{i}$ is an isometric extension of $\mathcal{A}_{i}$. Suppressing
implicit tensor products with identity in what follows, we can apply
Theorem~\ref{thm:Bures-distance-channels} to $d_{B}^{2}\!\left(\mathcal{P}_{i}^{(n)},\mathcal{P}_{2}^{(n)}\right)$
to conclude that
\begin{equation}
d_{B}^{2}\!\left(\mathcal{P}_{i}^{(n)},\mathcal{P}_{2}^{(n)}\right)\leq\left\Vert \hat{V}_{1}A_{n-1}\cdots A_{2}\hat{V}_{1}A_{1}\hat{V}_{1}-\hat{V}_{2}A_{n-1}\cdots A_{2}\hat{V}_{2}A_{1}\hat{V}_{2}\right\Vert ^{2}.\label{eq:up-bnd-bures-adaptive-step-1}
\end{equation}
Continuing, consider that 
\begin{align}
 & \left\Vert \hat{V}_{1}A_{n-1}\cdots A_{2}\hat{V}_{1}A_{1}\hat{V}_{1}-\hat{V}_{2}A_{n-1}\cdots A_{2}\hat{V}_{2}A_{1}\hat{V}_{2}\right\Vert ^{2}\nonumber \\
 & =\left\Vert \sum_{i=1}^{n}\hat{V}_{1}A_{n-1}\cdots A_{i+1}\left(\hat{V}_{1}-\hat{V}_{2}\right)A_{i-1}\cdots A_{1}\hat{V}_{2}\right\Vert ^{2}\\
 & =\left\Vert \begin{array}{c}
\left(\sum_{j=1}^{n}\hat{V}_{1}A_{n-1}\cdots A_{j+1}\left(\hat{V}_{1}-\hat{V}_{2}\right)A_{j-1}\cdots A_{1}\hat{V}_{2}\right)^{\dag}\times\\
\left(\sum_{i=1}^{n}\hat{V}_{1}A_{n-1}\cdots A_{i+1}\left(\hat{V}_{1}-\hat{V}_{2}\right)A_{i-1}\cdots A_{1}\hat{V}_{2}\right)
\end{array}\right\Vert \\
 & =\left\Vert \begin{array}{c}
\left(\sum_{j=1}^{n}\hat{V}_{2}^{\dag}A_{1}^{\dag}\cdots A_{j-1}^{\dag}\left(\hat{V}_{1}-\hat{V}_{2}\right)^{\dag}A_{j+1}^{\dag}\cdots A_{n-1}^{\dag}\hat{V}_{1}^{\dag}\right)\times\\
\left(\sum_{i=1}^{n}\hat{V}_{1}A_{n-1}\cdots A_{i+1}\left(\hat{V}_{1}-\hat{V}_{2}\right)A_{i-1}\cdots A_{1}\hat{V}_{2}\right)
\end{array}\right\Vert \\
 & =\left\Vert \begin{array}{c}
\sum_{i,j=1}^{n}\left(\hat{V}_{2}^{\dag}A_{1}^{\dag}\cdots A_{j-1}^{\dag}\left(\hat{V}_{1}-\hat{V}_{2}\right)^{\dag}A_{j+1}^{\dag}\cdots A_{n-1}^{\dag}\hat{V}_{1}^{\dag}\right)\times\\
\left(\hat{V}_{1}A_{n-1}\cdots A_{i+1}\left(\hat{V}_{1}-\hat{V}_{2}\right)A_{i-1}\cdots A_{1}\hat{V}_{2}\right)
\end{array}\right\Vert \\
 & \leq\sum_{i=1}^{n}\left\Vert \begin{array}{c}
\left(\hat{V}_{2}^{\dag}A_{1}^{\dag}\cdots A_{i-1}^{\dag}\left(\hat{V}_{1}-\hat{V}_{2}\right)^{\dag}A_{i+1}^{\dag}\cdots A_{n-1}^{\dag}\hat{V}_{1}^{\dag}\right)\times\\
\left(\hat{V}_{1}A_{n-1}\cdots A_{i+1}\left(\hat{V}_{1}-\hat{V}_{2}\right)A_{i-1}\cdots A_{1}\hat{V}_{2}\right)
\end{array}\right\Vert \nonumber \\
 & \quad+\sum_{i,j=1:i>j}^{n}\left\Vert \begin{array}{c}
\left(\hat{V}_{2}^{\dag}A_{1}^{\dag}\cdots A_{j-1}^{\dag}\left(\hat{V}_{1}-\hat{V}_{2}\right)^{\dag}A_{j+1}^{\dag}\cdots A_{n-1}^{\dag}\hat{V}_{1}^{\dag}\right)\times\\
\left(\hat{V}_{1}A_{n-1}\cdots A_{i+1}\left(\hat{V}_{1}-\hat{V}_{2}\right)A_{i-1}\cdots A_{1}\hat{V}_{2}\right)
\end{array}\right\Vert \nonumber \\
 & \quad+\sum_{i,j=1:i<j}^{n}\left\Vert \begin{array}{c}
\left(\hat{V}_{2}^{\dag}A_{1}^{\dag}\cdots A_{j-1}^{\dag}\left(\hat{V}_{1}-\hat{V}_{2}\right)^{\dag}A_{j+1}^{\dag}\cdots A_{n-1}^{\dag}\hat{V}_{1}^{\dag}\right)\times\\
\left(\hat{V}_{1}A_{n-1}\cdots A_{i+1}\left(\hat{V}_{1}-\hat{V}_{2}\right)A_{i-1}\cdots A_{1}\hat{V}_{2}\right)
\end{array}\right\Vert .\label{eq:adaptive-ch-disc-proof-up-bnd}
\end{align}
Let us bound the $i$th term of the first sum in~\eqref{eq:adaptive-ch-disc-proof-up-bnd}
as follows:
\begin{align}
 & \left\Vert \left(\hat{V}_{2}^{\dag}A_{1}^{\dag}\cdots A_{i-1}^{\dag}\left(\hat{V}_{1}-\hat{V}_{2}\right)^{\dag}A_{i+1}^{\dag}\cdots A_{n-1}^{\dag}\hat{V}_{1}^{\dag}\right)\left(\hat{V}_{1}A_{n-1}\cdots A_{i+1}\left(\hat{V}_{1}-\hat{V}_{2}\right)A_{i-1}\cdots A_{1}\hat{V}_{2}\right)\right\Vert \nonumber \\
 & =\left\Vert \hat{V}_{2}^{\dag}A_{1}^{\dag}\cdots A_{i-1}^{\dag}\left(\hat{V}_{1}-\hat{V}_{2}\right)^{\dag}\left(\hat{V}_{1}-\hat{V}_{2}\right)A_{i-1}\cdots A_{1}\hat{V}_{2}\right\Vert \\
 & \leq\left\Vert \left(\hat{V}_{1}-\hat{V}_{2}\right)^{\dag}\left(\hat{V}_{1}-\hat{V}_{2}\right)\right\Vert \left\Vert \hat{V}_{2}^{\dag}A_{1}^{\dag}\cdots A_{i-1}^{\dag}A_{i-1}\cdots A_{1}\hat{V}_{2}\right\Vert \\
 & =\left\Vert \hat{V}_{1}-\hat{V}_{2}\right\Vert ^{2}.\label{eq:adaptive-bures-bound-proof-1st-term}
\end{align}
Also, for $i>j$, let us bound the $(i,j)$th term of the second sum
in~\eqref{eq:adaptive-ch-disc-proof-up-bnd} as follows:
\begin{align}
 & \left\Vert \left(\hat{V}_{2}^{\dag}A_{1}^{\dag}\cdots A_{j-1}^{\dag}\left(\hat{V}_{1}-\hat{V}_{2}\right)^{\dag}A_{j+1}^{\dag}\cdots A_{n-1}^{\dag}\hat{V}_{1}^{\dag}\right)\left(\hat{V}_{1}A_{n-1}\cdots A_{i+1}\left(\hat{V}_{1}-\hat{V}_{2}\right)A_{i-1}\cdots A_{1}\hat{V}_{2}\right)\right\Vert \nonumber \\
 & =\left\Vert \hat{V}_{2}^{\dag}A_{1}^{\dag}\cdots A_{j-1}^{\dag}\left(\hat{V}_{1}-\hat{V}_{2}\right)^{\dag}\hat{V}_{1}A_{j-1}\cdots A_{i+1}\left(\hat{V}_{1}-\hat{V}_{2}\right)A_{i-1}\cdots A_{1}\hat{V}_{2}\right\Vert \\
 & \leq\left\Vert \hat{V}_{2}^{\dag}A_{1}^{\dag}\cdots A_{j-1}^{\dag}\right\Vert \left\Vert \left(\hat{V}_{1}-\hat{V}_{2}\right)^{\dag}\hat{V}_{1}\right\Vert \left\Vert A_{j-1}\cdots A_{i+1}\right\Vert \left\Vert \hat{V}_{1}-\hat{V}_{2}\right\Vert \left\Vert A_{i-1}\cdots A_{1}\hat{V}_{2}\right\Vert \\
 & =\left\Vert \left(\hat{V}_{1}-\hat{V}_{2}\right)^{\dag}\hat{V}_{1}\right\Vert \left\Vert \hat{V}_{1}-\hat{V}_{2}\right\Vert \\
 & =\left\Vert I-\hat{V}_{2}^{\dag}\hat{V}_{1}\right\Vert \left\Vert \hat{V}_{1}-\hat{V}_{2}\right\Vert \\
 & =\left\Vert \hat{V}_{1}-\hat{V}_{2}\right\Vert \left\Vert I-\hat{V}_{1}^{\dag}\hat{V}_{2}\right\Vert .\label{eq:adaptive-bures-bound-proof-2nd-term}
\end{align}
For $i<j$, let us bound the $(i,j)$th term of the third sum in~\eqref{eq:adaptive-ch-disc-proof-up-bnd}
as follows:
\begin{align}
 & \left\Vert \left(\hat{V}_{2}^{\dag}A_{1}^{\dag}\cdots A_{j-1}^{\dag}\left(\hat{V}_{1}-\hat{V}_{2}\right)^{\dag}A_{j+1}^{\dag}\cdots A_{n-1}^{\dag}\hat{V}_{1}^{\dag}\right)\left(\hat{V}_{1}A_{n-1}\cdots A_{i+1}\left(\hat{V}_{1}-\hat{V}_{2}\right)A_{i-1}\cdots A_{1}\hat{V}_{2}\right)\right\Vert \nonumber \\
 & =\left\Vert \hat{V}_{2}^{\dag}A_{1}^{\dag}\cdots A_{j-1}^{\dag}\left(\hat{V}_{1}-\hat{V}_{2}\right)^{\dag}A_{j+1}^{\dag}\cdots A_{i-1}^{\dag}\hat{V}_{1}^{\dag}\left(\hat{V}_{1}-\hat{V}_{2}\right)A_{i-1}\cdots A_{1}\hat{V}_{2}\right\Vert \\
 & \leq\left\Vert \hat{V}_{2}^{\dag}A_{1}^{\dag}\cdots A_{j-1}^{\dag}\right\Vert \left\Vert \left(\hat{V}_{1}-\hat{V}_{2}\right)^{\dag}\right\Vert \left\Vert A_{j+1}^{\dag}\cdots A_{i-1}^{\dag}\right\Vert \left\Vert \hat{V}_{1}^{\dag}\left(\hat{V}_{1}-\hat{V}_{2}\right)\right\Vert \left\Vert A_{i-1}\cdots A_{1}\hat{V}_{2}\right\Vert \\
 & =\left\Vert \hat{V}_{1}-\hat{V}_{2}\right\Vert \left\Vert \hat{V}_{1}^{\dag}\left(\hat{V}_{1}-\hat{V}_{2}\right)\right\Vert \\
 & =\left\Vert \hat{V}_{1}-\hat{V}_{2}\right\Vert \left\Vert I-\hat{V}_{1}^{\dag}\hat{V}_{2}\right\Vert .\label{eq:adaptive-bures-bound-proof-3rd-term}
\end{align}
Then, applying~\eqref{eq:up-bnd-bures-adaptive-step-1},~\eqref{eq:adaptive-ch-disc-proof-up-bnd},
\eqref{eq:adaptive-bures-bound-proof-1st-term},~\eqref{eq:adaptive-bures-bound-proof-2nd-term},
and~\eqref{eq:adaptive-bures-bound-proof-3rd-term}, we conclude that
\begin{align}
 & d_{B}^{2}\!\left(\mathcal{P}_{i}^{(n)},\mathcal{P}_{2}^{(n)}\right)\nonumber \\
 & \leq n\left\Vert \hat{V}_{1}-\hat{V}_{2}\right\Vert ^{2}+\frac{n\left(n-1\right)}{2}\left\Vert \hat{V}_{1}-\hat{V}_{2}\right\Vert \left\Vert I-\hat{V}_{1}^{\dag}\hat{V}_{2}\right\Vert +\frac{n\left(n-1\right)}{2}\left\Vert \hat{V}_{1}-\hat{V}_{2}\right\Vert \left\Vert I-\hat{V}_{1}^{\dag}\hat{V}_{2}\right\Vert \\
 & =n\left\Vert \hat{V}_{1}-\hat{V}_{2}\right\Vert ^{2}+n\left(n-1\right)\left\Vert \hat{V}_{1}-\hat{V}_{2}\right\Vert \left\Vert I-\hat{V}_{1}^{\dag}\hat{V}_{2}\right\Vert ,
\end{align}
thus concluding the proof after applying~\eqref{eq:norm-iso-exts-parallel-disc-simplify}--\eqref{eq:overlap-iso-exts-parallel-disc-simplify}
and the reasoning thereafter.
\end{proof}
\begin{rem}
We can alternatively express~\eqref{eq:bures-distance-adaptive-setting}
using the notation of~\cite{Yuan2017}. Appealing to Remark~\ref{rem:isometric-overlap-to-kraus}
and defining $K_{W}\coloneqq\sum_{i,j}W_{i,j}K_{i}^{\dag}L_{j}$,
we can write the objective function in~\eqref{eq:bures-distance-adaptive-setting}
as follows:
\begin{multline}
2\left\Vert I-\Re\!\left[M_{W}\right]\right\Vert +\left(n-1\right)\left(2\left\Vert I-\Re\!\left[M_{W}\right]\right\Vert \right)^{\frac{1}{2}}\left\Vert I-M_{W}\right\Vert \\
=\left\Vert 2I-K_{W}-K_{W}^{\dag}\right\Vert +\left(n-1\right)\left\Vert 2I-K_{W}-K_{W}^{\dag}\right\Vert ^{\frac{1}{2}}\left\Vert I-K_{W}\right\Vert .
\end{multline}
\begin{rem}
\label{rem:parallel-adaptive-disc-compare-up-bnds}The upper bound
in Theorem~\ref{thm:parallel-bures-dist-channels} does not exceed
the upper bound in Theorem~\ref{thm:seq-bures-distance-channels},
consistent with the intuition that the Bures distance for parallel
protocols should not exceed the Bures distance of adaptive protocols,
given that parallel protocols are contained with the set of all adaptive
protocols. This can be seen by comparing the objective functions in
Theorems~\ref{thm:parallel-bures-dist-channels} and~\ref{thm:seq-bures-distance-channels}
and noting that the following inequality holds for every contraction~$W$:
\begin{align}
\left\Vert I-M_{W}\right\Vert  & =\left\Vert \left(I-M_{W}\right)^{\dag}\left(I-M_{W}\right)\right\Vert ^{\frac{1}{2}}\\
 & =\left\Vert I-M_{W}-M_{W}^{\dag}+M_{W}^{\dag}M_{W}\right\Vert ^{\frac{1}{2}}\\
 & \leq\left\Vert I-M_{W}-M_{W}^{\dag}+I\right\Vert ^{\frac{1}{2}}\\
 & =\left(2\left\Vert I-\Re\!\left[M_{W}\right]\right\Vert \right)^{\frac{1}{2}},
\end{align}
where the inequality follows because
\begin{align}
I-M_{W}-M_{W}^{\dag}+M_{W}^{\dag}M_{W} & \leq I-M_{W}-M_{W}^{\dag}+I\\
\Longleftrightarrow\qquad M_{W}^{\dag}M_{W} & \leq I,\label{eq:M_W-contraction}
\end{align}
with the final inequality in~\eqref{eq:M_W-contraction} holding because
$M_{W}$ is a contraction, itself being a composition of contractions.
\end{rem}

\end{rem}

\section{Upper bounds on SLD Fisher information of channels}\label{sec:Upper-bounds-SLD-Fisher}

In this section, we present upper bounds on the SLD Fisher information
of a smooth family $\left(\mathcal{N}_{\theta}\right)_{\theta\in\Theta}$
of channels, in both the parallel and adaptive settings. The bounds
are expressed in terms of a canonical isometric extension $V_{\theta}\colon\mathbb{C}^{d_{A}}\mapsto\mathbb{C}^{d_{E}}\otimes\mathbb{C}^{d_{B}}$
of the channel $\mathcal{N}_{\theta}$, as given in~\eqref{eq:canonical-isometric-ext-smooth-family}. 

\subsection{Upper bound on SLD Fisher information in the parallel setting}

\cite[Section~4.1]{Fujiwara2008} established an upper bound on the
SLD Fisher information of the smooth family $\smash{\left(\mathcal{N}_{\theta}^{\otimes n}\right)_{\theta\in\Theta}}$
of channels, for all $n\in\mathbb{N}$. This upper bound has subsequently
been employed in several foundational papers on quantum estimation
theory~\cite{DemkowiczDobrzanski2012,Kolodynski2013,Demkowicz2014,Zhou2021}.
Here we restate this upper bound in a different form, in terms of
canonical isometric extensions of the channel family $\left(\mathcal{N}_{\theta}\right)_{\theta\in\Theta}$,
and we provide an alternate proof based on basic properties of isometries
and the spectral norm.

\begin{thm}
\label{thm:upper-bound-SLD-Fisher-parallel}For a smooth family $\left(\mathcal{N}_{\theta}\right)_{\theta\in\Theta}$
of quantum channels, the following inequality holds for all $n\in\mathbb{N}$:
\begin{align}
\frac{1}{4}I_{F}\!\left(\theta;\left(\mathcal{N}_{\theta}^{\otimes n}\right)_{\theta\in\Theta}\right) & \leq n\inf_{\left(H_{\theta}\right)_{\theta\in\Theta}}\left\{ \left\Vert M_{H_{\theta}}\right\Vert ^{2}+\left(n-1\right)\left\Vert V_{\theta}^{\dag}M_{H_{\theta}}\right\Vert ^{2}\right\} ,
\end{align}
where $V_{\theta}\colon\mathbb{C}^{d_{A}}\mapsto\mathbb{C}^{d_{E}}\otimes\mathbb{C}^{d_{B}}$
is a canonical isometric extension of $\mathcal{N}_{\theta}$ of the
form in~\eqref{eq:canonical-isometric-ext-smooth-family}, the optimization
is over every smooth family $\left(H_{\theta}\right)_{\theta\in\Theta}$
of Hermitian operators, and
\begin{equation}
M_{H_{\theta}}\equiv\left(\partial_{\theta}V_{\theta}\right)-i\left(H_{\theta}\otimes I_{B}\right)V_{\theta}.
\end{equation}
\end{thm}

\begin{proof}
Let us employ the shorthand
\begin{equation}
\hat{V}_{\theta}\equiv\left(U_{\theta}\otimes I_{B}\right)V_{\theta},\label{eq:isometric-extension-parallel-estimation}
\end{equation}
where $\left(U_{\theta}\right)_{\theta\in\Theta}$ is a smooth family
of unitary operators, so that
\begin{equation}
\partial_{\theta}\hat{V}_{\theta}=\partial_{\theta}\left[\left(U_{\theta}\otimes I_{RB}\right)V_{\theta}\right]=\left(\partial_{\theta}V_{\theta}\right)-i\left(H_{\theta}\otimes I_{B}\right)V_{\theta},\label{eq:deriv-V-hat}
\end{equation}
where $\left(H_{\theta}\right)_{\theta\in\Theta}$ is a smooth family
of Hermitian operators, which follows from reasoning similar to that
given for~\eqref{eq:SLD-channels-proof-herm-introduce}. We can use
the fact that $\hat{V}_{\theta}^{\otimes n}$ is a particular isometric
extension of $\mathcal{N}_{\theta}^{\otimes n}$ and apply Theorem
\ref{thm:SLD-Fisher-channels} to conclude that
\begin{align}
 & \frac{1}{4}I_{F}\!\left(\theta;\left(\mathcal{N}_{\theta}^{\otimes n}\right)_{\theta\in\Theta}\right)\nonumber \\
 & \leq\left\Vert \partial_{\theta}\left[\hat{V}_{\theta}^{\otimes n}\right]\right\Vert ^{2}\\
 & =\left\Vert \sum_{i=1}^{n}\hat{V}_{\theta}^{\otimes i-1}\otimes\left[\partial_{\theta}\hat{V}_{\theta}\right]\otimes\hat{V}_{\theta}^{\otimes n-i}\right\Vert ^{2}\\
 & =\left\Vert \left(\sum_{j=1}^{n}\hat{V}_{\theta}^{\otimes j-1}\otimes\left[\partial_{\theta}\hat{V}_{\theta}\right]_{j}\otimes\hat{V}_{\theta}^{\otimes n-j}\right)^{\dag}\left(\sum_{i=1}^{n}\hat{V}_{\theta}^{\otimes i-1}\otimes\left[\partial_{\theta}\hat{V}_{\theta}\right]_{i}\otimes\hat{V}_{\theta}^{\otimes n-i}\right)\right\Vert \\
 & =\left\Vert \left(\sum_{j=1}^{n}\left(\hat{V}_{\theta}^{\dag}\right)^{\otimes j-1}\otimes\left[\partial_{\theta}\hat{V}_{\theta}^{\dag}\right]_{j}\otimes\left(\hat{V}_{\theta}^{\dag}\right)^{\otimes n-j}\right)\left(\sum_{i=1}^{n}\hat{V}_{\theta}^{\otimes i-1}\otimes\left[\partial_{\theta}\hat{V}_{\theta}\right]_{i}\otimes\hat{V}_{\theta}^{\otimes n-i}\right)\right\Vert \\
 & =\left\Vert \sum_{i,j=1}^{n}\left(\left(\hat{V}_{\theta}^{\dag}\right)^{\otimes j-1}\otimes\left[\partial_{\theta}\hat{V}_{\theta}^{\dag}\right]_{j}\otimes\left(\hat{V}_{\theta}^{\dag}\right)^{\otimes n-j}\right)\left(\hat{V}_{\theta}^{\otimes i-1}\otimes\left[\partial_{\theta}\hat{V}_{\theta}\right]_{i}\otimes\hat{V}_{\theta}^{\otimes n-i}\right)\right\Vert \\
 & \leq\sum_{i=1}^{n}\left\Vert \left(\left(\hat{V}_{\theta}^{\dag}\right)^{\otimes i-1}\otimes\left[\partial_{\theta}\hat{V}_{\theta}^{\dag}\right]_{i}\otimes\left(\hat{V}_{\theta}^{\dag}\right)^{\otimes n-i}\right)\left(\hat{V}_{\theta}^{\otimes i-1}\otimes\left[\partial_{\theta}\hat{V}_{\theta}\right]_{i}\otimes\hat{V}_{\theta}^{\otimes n-i}\right)\right\Vert \nonumber \\
 & \qquad+\sum_{i,j=1:i\neq j}^{n}\left\Vert \left(\left(\hat{V}_{\theta}^{\dag}\right)^{\otimes j-1}\otimes\left[\partial_{\theta}\hat{V}_{\theta}^{\dag}\right]_{j}\otimes\left(\hat{V}_{\theta}^{\dag}\right)^{\otimes n-j}\right)\left(\hat{V}_{\theta}^{\otimes i-1}\otimes\left[\partial_{\theta}\hat{V}_{\theta}\right]_{i}\otimes\hat{V}_{\theta}^{\otimes n-i}\right)\right\Vert \\
 & =\sum_{i=1}^{n}\left\Vert \left(\hat{V}_{\theta}^{\dag}\hat{V}_{\theta}\right)^{\otimes i-1}\otimes\left[\left(\partial_{\theta}\hat{V}_{\theta}^{\dag}\right)\left(\partial_{\theta}\hat{V}_{\theta}\right)\right]_{i}\otimes\left(\hat{V}_{\theta}^{\dag}\hat{V}_{\theta}\right)^{\otimes n-i}\right\Vert \nonumber \\
 & \qquad+\sum_{i,j=1:i\neq j}^{n}\left\Vert \hat{V}_{\theta}^{\dag}\left(\partial_{\theta}\hat{V}_{\theta}\right)\right\Vert \left\Vert \left(\partial_{\theta}\hat{V}_{\theta}^{\dag}\right)\hat{V}_{\theta}\right\Vert \\
 & =\sum_{i=1}^{n}\left\Vert \left(\partial_{\theta}\hat{V}_{\theta}^{\dag}\right)\left(\partial_{\theta}\hat{V}_{\theta}\right)\right\Vert +\sum_{i,j=1:i\neq j}^{n}\left\Vert \hat{V}_{\theta}^{\dag}\left(\partial_{\theta}\hat{V}_{\theta}\right)\right\Vert ^{2}\\
 & =\sum_{i=1}^{n}\left\Vert \partial_{\theta}\hat{V}_{\theta}\right\Vert ^{2}+\sum_{i,j=1:i\neq j}^{n}\left\Vert \hat{V}_{\theta}^{\dag}\left(\partial_{\theta}\hat{V}_{\theta}\right)\right\Vert ^{2}\\
 & =n\left\Vert \partial_{\theta}\hat{V}_{\theta}\right\Vert ^{2}+n\left(n-1\right)\left\Vert \hat{V}_{\theta}^{\dag}\left(\partial_{\theta}\hat{V}_{\theta}\right)\right\Vert ^{2}.
\end{align}
The first inequality follows from the triangle inequality for the
spectral norm. The other equalities follow from the spectral norm
being multiplicative under tensor products and it being equal to one
for the identity operator. Applying~\eqref{eq:deriv-V-hat} and Remark
\ref{rem:simplify-overlap-term-herm} and minimizing over every smooth
family $\left(H_{\theta}\right)_{\theta\in\Theta}$ of Hermitian operators,
we conclude the proof.
\end{proof}

\subsection{Upper bound on SLD Fisher information in the adaptive setting}

For a smooth family $\left(\mathcal{N}_{\theta}\right)_{\theta\in\Theta}$
of quantum channels, a general adaptive $n$-query protocol for channel
estimation consists of a tuple $\left(\mathcal{A}_{i}\right)_{i=1}^{n-1}$
of channels. Applying the protocol to $n$ queries of the channel
$\mathcal{N}_{\theta}$ leads to the following channel $\mathcal{P}_{\theta}^{(n)}$
for $\theta\in\Theta$:
\begin{equation}
\mathcal{P}_{\theta}^{(n)}\coloneqq\left(\id\otimes\mathcal{N}_{\theta}\right)\circ\mathcal{A}_{n-1}\circ\cdots\circ\mathcal{A}_{2}\circ\left(\id\otimes\mathcal{N}_{\theta}\right)\circ\mathcal{A}_{1}\circ\left(\id\otimes\mathcal{N}_{\theta}\right).\label{eq:general-protocol-adaptive-estimation}
\end{equation}
This in turn leads to the following information measure for the family
$\left(\mathcal{N}_{\theta}\right)_{\theta\in\Theta}$:
\begin{equation}
\sup_{\left(\mathcal{A}_{i}\right)_{i=1}^{n-1}}I_{F}\!\left(\theta;\left(\mathcal{P}_{\theta}^{(n)}\right)_{\theta\in\mathbb{R}}\right),\label{eq:adaptive-SLD-Fisher-def}
\end{equation}
which is related to the general measure considered in~\cite[Theorem~18]{Katariya2021}.

Ref.~\cite[Eq.~(11)]{Kurdzialek2023} established an upper bound
on the $n$-query adaptive SLD Fisher information in~\eqref{eq:adaptive-SLD-Fisher-def},
for all $n\in\mathbb{N}$, which improves upon a previous bound from
\cite[Eq.~(9)]{Demkowicz2014}. Here we restate this upper bound in
a different form, in terms of canonical isometric extensions of the
channel family $\left(\mathcal{N}_{\theta}\right)_{\theta\in\Theta}$,
and we provide an alternate proof based on basic properties of isometries
and the spectral norm.
\begin{thm}
\label{thm:upper-bound-SLD-Fisher-adaptive}For a smooth family $\left(\mathcal{N}_{\theta}\right)_{\theta\in\Theta}$
of quantum channels, the following inequality holds for all $n\in\mathbb{N}$:
\begin{align}
\frac{1}{4}\sup_{\left(\mathcal{A}_{i}\right)_{i=1}^{n-1}}I_{F}\!\left(\theta;\left(\mathcal{P}_{\theta}^{(n)}\right)_{\theta\in\mathbb{R}}\right) & \leq n\inf_{\left(H_{\theta}\right)_{\theta\in\Theta}}\left\{ \left\Vert M_{H_{\theta}}\right\Vert ^{2}+\left(n-1\right)\left\Vert V_{\theta}^{\dag}M_{H_{\theta}}\right\Vert \left\Vert M_{H_{\theta}}\right\Vert \right\} ,
\end{align}
where $V_{\theta}\colon\mathbb{C}^{d_{A}}\mapsto\mathbb{C}^{d_{E}}\otimes\mathbb{C}^{d_{B}}$
is a canonical isometric extension of $\mathcal{N}_{\theta}$ of the
form in~\eqref{eq:canonical-isometric-ext-smooth-family}, the optimization
is over every smooth family $\left(H_{\theta}\right)_{\theta\in\Theta}$
of Hermitian operators, and
\begin{equation}
M_{H_{\theta}}\equiv\left(\partial_{\theta}V_{\theta}\right)-i\left(H_{\theta}\otimes I_{B}\right)V_{\theta}.
\end{equation}
\end{thm}

\begin{proof}
Let $\hat{V}_{\theta}$ be defined as in~\eqref{eq:isometric-extension-parallel-estimation}.
We can take isometric extensions of every step of the protocol $\mathcal{P}_{\theta}^{(n)}$
in~\eqref{eq:general-protocol-adaptive-estimation} to write an isometric
extension of $\mathcal{P}_{\theta}^{(n)}$ as follows:
\begin{equation}
\left(I\otimes\hat{V}_{\theta}\right)A_{n-1}\cdots A_{2}\left(I\otimes\hat{V}_{\theta}\right)A_{1}\left(I\otimes\hat{V}_{\theta}\right),
\end{equation}
where $A_{i}$ is an isometric extension of $\mathcal{A}_{i}$. Suppressing
implicit tensor products with identity in what follows, we apply Theorem
\ref{thm:SLD-Fisher-channels} to conclude that
\begin{align}
 & \frac{1}{4}I_{F}\!\left(\theta;\left(\mathcal{P}_{\theta}^{(n)}\right)_{\theta\in\mathbb{R}}\right)\nonumber \\
 & \leq\left\Vert \partial_{\theta}\left[\hat{V}_{\theta}A_{n-1}\cdots A_{2}\hat{V}_{\theta}A_{1}\hat{V}_{\theta}\right]\right\Vert ^{2}\\
 & =\left\Vert \sum_{i=1}^{n}\hat{V}_{\theta}A_{n-1}\cdots A_{i+1}\left(\partial_{\theta}\hat{V}_{\theta}\right)A_{i-1}\cdots A_{1}\hat{V}_{\theta}\right\Vert ^{2}\\
 & =\left\Vert \begin{array}{c}
\left(\sum_{j=1}^{n}\hat{V}_{\theta}A_{n-1}\cdots A_{j+1}\left(\partial_{\theta}\hat{V}_{\theta}\right)A_{j-1}\cdots A_{1}\hat{V}_{\theta}\right)^{\dag}\times\\
\left(\sum_{i=1}^{n}\hat{V}_{\theta}A_{n-1}\cdots A_{i+1}\left(\partial_{\theta}\hat{V}_{\theta}\right)A_{i-1}\cdots A_{1}\hat{V}_{\theta}\right)
\end{array}\right\Vert \\
 & =\left\Vert \begin{array}{c}
\left(\sum_{j=1}^{n}\hat{V}_{\theta}^{\dag}A_{1}^{\dag}\cdots A_{j-1}^{\dag}\left(\partial_{\theta}\hat{V}_{\theta}^{\dag}\right)A_{j+1}^{\dag}\cdots A_{n-1}^{\dag}\hat{V}_{\theta}^{\dag}\right)\times\\
\left(\sum_{i=1}^{n}\hat{V}_{\theta}A_{n-1}\cdots A_{i+1}\left(\partial_{\theta}\hat{V}_{\theta}\right)A_{i-1}\cdots A_{1}\hat{V}_{\theta}\right)
\end{array}\right\Vert \\
 & =\left\Vert \sum_{i,j=1}^{n}\left(\hat{V}_{\theta}^{\dag}A_{1}^{\dag}\cdots A_{j-1}^{\dag}\left(\partial_{\theta}\hat{V}_{\theta}^{\dag}\right)A_{j+1}^{\dag}\cdots A_{n-1}^{\dag}\hat{V}_{\theta}^{\dag}\right)\left(\hat{V}_{\theta}A_{n-1}\cdots A_{i+1}\left(\partial_{\theta}\hat{V}_{\theta}\right)A_{i-1}\cdots A_{1}\hat{V}_{\theta}\right)\right\Vert \\
 & \leq\sum_{i=1}^{n}\left\Vert \hat{V}_{\theta}^{\dag}A_{1}^{\dag}\cdots A_{i-1}^{\dag}\left(\partial_{\theta}\hat{V}_{\theta}^{\dag}\right)A_{i+1}^{\dag}\cdots A_{n-1}^{\dag}\hat{V}_{\theta}^{\dag}\hat{V}_{\theta}A_{n-1}\cdots A_{i+1}\left(\partial_{\theta}\hat{V}_{\theta}\right)A_{i-1}\cdots A_{1}\hat{V}_{\theta}\right\Vert \nonumber \\
 & \ +\sum_{\substack{i,j=1:\\
i\neq j
}
}^{n}\left\Vert \left(\hat{V}_{\theta}^{\dag}A_{1}^{\dag}\cdots A_{j-1}^{\dag}\left(\partial_{\theta}\hat{V}_{\theta}^{\dag}\right)A_{j+1}^{\dag}\cdots A_{n-1}^{\dag}\hat{V}_{\theta}^{\dag}\right)\left(\hat{V}_{\theta}A_{n-1}\cdots A_{i+1}\left(\partial_{\theta}\hat{V}_{\theta}\right)A_{i-1}\cdots A_{1}\hat{V}_{\theta}\right)\right\Vert .
\end{align}
Consider that
\begin{align}
 & \sum_{i=1}^{n}\left\Vert \hat{V}_{\theta}^{\dag}A_{1}^{\dag}\cdots A_{i-1}^{\dag}\left(\partial_{\theta}\hat{V}_{\theta}^{\dag}\right)A_{i+1}^{\dag}\cdots A_{n-1}^{\dag}\hat{V}_{\theta}^{\dag}\hat{V}_{\theta}A_{n-1}\cdots A_{i+1}\left(\partial_{\theta}\hat{V}_{\theta}\right)A_{i-1}\cdots A_{1}\hat{V}_{\theta}\right\Vert \nonumber \\
 & =\sum_{i=1}^{n}\left\Vert \hat{V}_{\theta}^{\dag}A_{1}^{\dag}\cdots A_{i-1}^{\dag}\left(\partial_{\theta}\hat{V}_{\theta}^{\dag}\right)\left(\partial_{\theta}\hat{V}_{\theta}\right)A_{i-1}\cdots A_{1}\hat{V}_{\theta}\right\Vert \\
 & \leq\sum_{i=1}^{n}\left\Vert \hat{V}_{\theta}^{\dag}A_{1}^{\dag}\cdots A_{i-1}^{\dag}A_{i-1}\cdots A_{1}\hat{V}_{\theta}\right\Vert \left\Vert \left(\partial_{\theta}\hat{V}_{\theta}^{\dag}\right)\left(\partial_{\theta}\hat{V}_{\theta}\right)\right\Vert \\
 & =\sum_{i=1}^{n}\left\Vert \partial_{\theta}\hat{V}_{\theta}\right\Vert ^{2}\\
 & =n\left\Vert \partial_{\theta}\hat{V}_{\theta}\right\Vert ^{2}.
\end{align}
For the cross terms where $i\neq j$, a generic term such that $i>j$
(without loss of generality) can be bounded from above as follows:
\begin{align}
 & \left\Vert \left(\hat{V}_{\theta}^{\dag}A_{1}^{\dag}\cdots A_{j-1}^{\dag}\left(\partial_{\theta}\hat{V}_{\theta}^{\dag}\right)A_{j+1}^{\dag}\cdots A_{n-1}^{\dag}\hat{V}_{\theta}^{\dag}\right)\left(\hat{V}_{\theta}A_{n-1}\cdots A_{i+1}\left(\partial_{\theta}\hat{V}_{\theta}\right)A_{i-1}\cdots A_{1}\hat{V}_{\theta}\right)\right\Vert \nonumber \\
 & =\left\Vert \hat{V}_{\theta}^{\dag}A_{1}^{\dag}\cdots A_{j-1}^{\dag}\left(\partial_{\theta}\hat{V}_{\theta}^{\dag}\right)\hat{V}_{\theta}A_{j-1}\cdots A_{i+1}\left(\partial_{\theta}\hat{V}_{\theta}\right)A_{i-1}\cdots A_{1}\hat{V}_{\theta}\right\Vert \\
 & \leq\left\Vert \hat{V}_{\theta}^{\dag}A_{1}^{\dag}\cdots A_{j-1}^{\dag}\right\Vert \left\Vert \left(\partial_{\theta}\hat{V}_{\theta}^{\dag}\right)\hat{V}_{\theta}\right\Vert \left\Vert A_{j-1}\cdots A_{i+1}\right\Vert \left\Vert \partial_{\theta}\hat{V}_{\theta}\right\Vert \left\Vert A_{i-1}\cdots A_{1}\hat{V}_{\theta}\right\Vert \\
 & =\left\Vert \left(\partial_{\theta}\hat{V}_{\theta}^{\dag}\right)\hat{V}_{\theta}\right\Vert \left\Vert \partial_{\theta}\hat{V}_{\theta}\right\Vert \\
 & =\left\Vert \hat{V}_{\theta}^{\dag}\partial_{\theta}\hat{V}_{\theta}\right\Vert \left\Vert \partial_{\theta}\hat{V}_{\theta}\right\Vert ,
\end{align}
where the inequality follows from submultiplicativity of the spectral
norm. Thus, we conclude that
\begin{multline}
\sum_{\substack{i,j=1:\\
i\neq j
}
}^{n}\left\Vert \left(\hat{V}_{\theta}^{\dag}A_{1}^{\dag}\cdots A_{j-1}^{\dag}\left(\partial_{\theta}\hat{V}_{\theta}^{\dag}\right)A_{j+1}^{\dag}\cdots A_{n-1}^{\dag}\hat{V}_{\theta}^{\dag}\right)\left(\hat{V}_{\theta}A_{n-1}\cdots A_{i+1}\left(\partial_{\theta}\hat{V}_{\theta}\right)A_{i-1}\cdots A_{1}\hat{V}_{\theta}\right)\right\Vert \\
\leq n\left(n-1\right)\left\Vert \hat{V}_{\theta}^{\dag}\partial_{\theta}\hat{V}_{\theta}\right\Vert \left\Vert \partial_{\theta}\hat{V}_{\theta}\right\Vert .
\end{multline}
Applying~\eqref{eq:deriv-V-hat} and Remark~\ref{rem:simplify-overlap-term-herm}
and minimizing over every smooth family $\left(H_{\theta}\right)_{\theta\in\Theta}$
of Hermitian operators, we conclude the proof.
\end{proof}

\begin{rem}
The upper bound in Theorem~\ref{thm:upper-bound-SLD-Fisher-parallel}
does not exceed the upper bound in Theorem~\ref{thm:upper-bound-SLD-Fisher-adaptive},
consistent with the intuition that the SLD Fisher information for
parallel protocols should not exceed the SLD Fisher information of
adaptive protocols, given that parallel protocols are contained with
the set of all adaptive protocols. This can be seen by comparing the
objective functions in Theorems~\ref{thm:upper-bound-SLD-Fisher-parallel}
and~\ref{thm:upper-bound-SLD-Fisher-adaptive} and noting that the
following inequality holds for every Hermitian operator $H_{\theta}$
and isometry $V_{\theta}$:
\begin{align}
\left\Vert V_{\theta}^{\dag}M_{H_{\theta}}\right\Vert  & =\left\Vert V_{\theta}V_{\theta}^{\dag}M_{H_{\theta}}\right\Vert \\
 & \leq\left\Vert M_{H_{\theta}}\right\Vert 
\end{align}
where the equality follows because the spectral norm is invariant
under left multiplication by the isometry $V_{\theta}$ and the inequality
follows because $V_{\theta}V_{\theta}^{\dag}$ is a projection and
the spectral norm does not increase under the action of a projection.
\end{rem}

\section{Applications to query complexity}\label{sec:Applications-to-query-comp}

In this section, we present lower bounds on the error probabilities
and query complexities of channel discrimination and estimation. With
the main framework in place now, our proofs of these lower bounds
follow somewhat directly from definitions and the upper bounds presented
in Sections~\ref{sec:Upper-bounds-Bures} and~\ref{sec:Upper-bounds-SLD-Fisher}.

\subsection{Lower bounds on query complexities of channel discrimination}

\label{subsec:Lower-bounds-query-comp-ch-disc}We begin by establishing
bounds on the performance of channel discrimination in the parallel
and adaptive settings, in terms of the metric defined in~\eqref{eq:alt-channel-disc-metric}:
\begin{cor}
\label{cor:error-prob-bounds-ch-disc}Let $p\in\left(0,1\right)$
and set $q=1-p$. For channels $\mathcal{N}_{1}$ and $\mathcal{N}_{2}$ with canonical isometric extensions $V_1$ and $V_2$,
the following bounds hold for the error probabilities of channel discrimination:
\begin{align}
1-\frac{p_{e,\|}\left(1-p_{e,\|}\right)}{pq} & \leq n\inf_{\substack{W\in\mathbb{B}
}
}\left\{ a_{W}+\left(n-1\right)b_{W}\right\} ,\label{eq:parallel-error-prob-up-bnd-disc}\\
1-\frac{p_{e,\adaptivedisc}\left(1-p_{e,\adaptivedisc}\right)}{pq} & \leq n\inf_{\substack{W\in\mathbb{B}
}
}\left\{ a_{W}+\left(n-1\right)\sqrt{a_{W}b_{W}}\right\} ,\label{eq:adaptive-error-prob-up-bnd}
\end{align}
where
\begin{align}
a_{W} & \equiv2\left\Vert I-\Re\!\left[M_{W}\right]\right\Vert ,\\
b_{W} & \equiv\left\Vert I-M_{W}\right\Vert ^{2},\\
p_{e,\|} & \equiv p_{e,\|}(p,\mathcal{N}_{1},q,\mathcal{N}_{2},n),\label{eq:parallel-err-prob-shorthand}\\
p_{e,\adaptivedisc} & \equiv p_{e,\adaptivedisc}(p,\mathcal{N}_{1},q,\mathcal{N}_{2},n),\\
M_{W} & \equiv V_{1}^{\dag}\left(W_{E}\otimes I_{B}\right)V_{2},
\end{align}
and $p_{e,\|}$ and $p_{e,\adaptivedisc}$ are defined in~\eqref{eq:error-prob-parallel-ch-disc}
and~\eqref{eq:error-prob-adaptive-ch-disc}, respectively.
\end{cor}

\begin{proof}
Let us rewrite~\eqref{eq:error-prob-parallel-ch-disc} and~\eqref{eq:error-prob-adaptive-ch-disc}
as follows:
\begin{align}
\left\Vert p\mathcal{N}_{1}^{\otimes n}-q\mathcal{N}_{2}^{\otimes n}\right\Vert _{\diamond} & =1-2p_{e,\|},\\
\sup_{\left(\mathcal{A}_{i}\right)_{i=1}^{n-1}}\left\Vert p\mathcal{P}_{1}^{(n)}-q\mathcal{P}_{2}^{(n)}\right\Vert _{\diamond} & =1-2p_{e,\adaptivedisc}.
\end{align}
Now recall the following inequality holding for states $\rho$ and
$\sigma$ (see, e.g.,~\cite[Eqs.~(F8)--(F13)]{Cheng2025}):
\begin{equation}
\left\Vert p\rho-q\sigma\right\Vert _{1}^{2}\leq1-4pq+4pqd_{B}^{2}(\rho,\sigma).\label{eq:TD-to-Bures-basic}
\end{equation}
By applying this to the output state of a parallel discrimination
protocol and optimizing over input states, this then generalizes to
\begin{equation}
\left\Vert p\mathcal{N}_{1}^{\otimes n}-q\mathcal{N}_{2}^{\otimes n}\right\Vert _{\diamond}^{2}\leq1-4pq+4pqd_{B}^{2}(\mathcal{N}_{1}^{\otimes n},\mathcal{N}_{2}^{\otimes n}),
\end{equation}
which can be rewritten as
\begin{equation}
1-\frac{p_{e,\|}\left(1-p_{e,\|}\right)}{pq}\leq d_{B}^{2}(\mathcal{N}_{1}^{\otimes n},\mathcal{N}_{2}^{\otimes n}),\label{eq:bures-parallel-err-prob-bnd}
\end{equation}
to which we apply Theorem~\ref{thm:parallel-bures-dist-channels}
and conclude the proof of~\eqref{eq:parallel-error-prob-up-bnd-disc}.

Similarly, for adaptive protocols, by applying~\eqref{eq:TD-to-Bures-basic}
to the final output state of an adaptive channel discrimination protocol
and optimizing over all such protocols,~\eqref{eq:TD-to-Bures-basic}
generalizes to
\begin{equation}
\left(\sup_{\left(\mathcal{A}_{i}\right)_{i=1}^{n-1}}\left\Vert p\mathcal{P}_{1}^{(n)}-q\mathcal{P}_{2}^{(n)}\right\Vert _{\diamond}\right)^{2}\leq1-4pq+4pq\sup_{\left(\mathcal{A}_{i}\right)_{i=1}^{n-1}}d_{B}^{2}\!\left(\mathcal{P}_{1}^{(n)},\mathcal{P}_{2}^{(n)}\right),
\end{equation}
which can be rewritten as
\begin{equation}
1-\frac{p_{e,\adaptivedisc}\left(1-p_{e,\adaptivedisc}\right)}{pq}\leq\sup_{\left(\mathcal{A}_{i}\right)_{i=1}^{n-1}}d_{B}^{2}\!\left(\mathcal{P}_{1}^{(n)},\mathcal{P}_{2}^{(n)}\right),\label{eq:bures-adaptive-err-prob-bnd}
\end{equation}
to which we apply Theorem~\ref{thm:seq-bures-distance-channels} and
conclude the proof of~\eqref{eq:adaptive-error-prob-up-bnd}.
\end{proof}
\begin{cor}
\label{cor:query-comp-lower-bnd-ch-disc}Let $p\in\left(0,1\right)$,
set $q=1-p$, let $\varepsilon\in\left(0,\min\!\left\{ p,q\right\} \right]$,
and let $\mathcal{N}_{1}$ and $\mathcal{N}_{2}$ be quantum channels with canonical isometric extensions $V_1$ and $V_2$.
Suppose that the conditions in Remark~\ref{rem:trivial-cases-ch-disc-query-comp}
do not hold. Then the following lower bounds hold for the query complexities
of channel discrimination:
\begin{align}
n_{\|}^{\star}(p,\mathcal{N}_{1},q,\mathcal{N}_{2},\varepsilon) & \geq\left\lceil \sup_{W\in\mathbb{B}}\frac{b_{W}-a_{W}+\sqrt{\left(b_{W}-a_{W}\right)^{2}+4b_{W}\left(1-\frac{\varepsilon\left(1-\varepsilon\right)}{pq}\right)}}{2b_{W}}\right\rceil \label{eq:lower-bnd-heis-1}\\
 & \geq\left\lceil \frac{1-\frac{\varepsilon\left(1-\varepsilon\right)}{pq}}{\inf_{W\in\mathbb{B}}\left\{ a_{W}: M_{W}=I\right\} }\right\rceil ,\label{eq:lower-bnd-sql-1}\\
n_{\adaptivedisc}^{\star}(p,\mathcal{N}_{1},q,\mathcal{N}_{2},\varepsilon) & \geq\left\lceil \sup_{\substack{W\in\mathbb{B}
}
}\frac{c_{W}+\sqrt{c_{W}^{2}+4\sqrt{a_{W}b_{W}}\left(1-\frac{\varepsilon\left(1-\varepsilon\right)}{pq}\right)}}{2\sqrt{a_{W}b_{W}}}\right\rceil \label{eq:lower-bnd-heis-2}\\
 & \geq\left\lceil \frac{1-\frac{\varepsilon\left(1-\varepsilon\right)}{pq}}{\inf_{W\in\mathbb{B}}\left\{ a_{W}: M_{W}=I\right\} }\right\rceil ,\label{eq:lower-bnd-sql-2}
\end{align}
where
\begin{align}
a_{W} & \equiv2\left\Vert I-\Re\!\left[M_{W}\right]\right\Vert ,\label{eq:def-a_W-shorthand}\\
b_{W} & \equiv\left\Vert I-M_{W}\right\Vert ^{2},\label{eq:def-b_W-shorthand}\\
c_{W} & \equiv\sqrt{a_{W}b_{W}}-a_{W},\\
M_{W} & \equiv V_{1}^{\dag}\left(W_{E}\otimes I_{B}\right)V_{2}.
\end{align}
\end{cor}

\begin{proof}
This is a direct consequence of Definition~\ref{def:query-complexity-ch-disc-def},
Corollary~\ref{cor:error-prob-bounds-ch-disc} and Lemma~\ref{lem:quadratic-ineq-rewrite}
below, while noting that the function $x\mapsto x\left(1-x\right)$
is monotone increasing on the interval $\left[0,\frac{1}{2}\right]$.
In more detail, we prove~\eqref{eq:lower-bnd-heis-1} and~\eqref{eq:lower-bnd-sql-1}
and omit the proofs of~\eqref{eq:lower-bnd-heis-2} and~\eqref{eq:lower-bnd-sql-2},
as they follow similarly.

Applying Definition~\ref{def:query-complexity-ch-disc-def}, let $n\in\mathbb{N}$
be such that $p_{e,\|}\leq\varepsilon$, where we have adopted the
shorthand $p_{e,\|}$, as defined in~\eqref{eq:parallel-err-prob-shorthand}.
Let $W\in\mathbb{B}$ be a contraction (i.e., $\left\Vert W\right\Vert \leq1$).
Then from the monotonicity of $x\mapsto x\left(1-x\right)$ on $\left[0,\frac{1}{2}\right]$
and Corollary~\ref{cor:error-prob-bounds-ch-disc}, we conclude that
\begin{align}
1-\frac{\varepsilon\left(1-\varepsilon\right)}{pq} & \leq1-\frac{p_{e,\|}\left(1-p_{e,\|}\right)}{pq},\\
 & \leq n\left(a_{W}+\left(n-1\right)b_{W}\right),\label{eq:up-bnd-err-prob-for-q-comp}
\end{align}
where $a_{W}$ and $b_{W}$ are defined in~\eqref{eq:def-a_W-shorthand}
and~\eqref{eq:def-b_W-shorthand}, respectively. Applying Lemma~\ref{lem:quadratic-ineq-rewrite}
below, we conclude that
\begin{equation}
n\geq\left\lceil \frac{b_{W}-a_{W}+\sqrt{\left(b_{W}-a_{W}\right)^{2}+4b_{W}\left(1-\frac{\varepsilon\left(1-\varepsilon\right)}{pq}\right)}}{2b_{W}}\right\rceil .
\end{equation}
Since this lower bound holds for every contraction $W\in\mathbb{B}$,
we can take a supremum over all such contractions to conclude that
\begin{equation}
n\geq\sup_{\substack{W\in\mathbb{B}
}
}\left\lceil \frac{b_{W}-a_{W}+\sqrt{\left(b_{W}-a_{W}\right)^{2}+4b_{W}\left(1-\frac{\varepsilon\left(1-\varepsilon\right)}{pq}\right)}}{2b_{W}}\right\rceil .
\end{equation}
We then conclude~\eqref{eq:lower-bnd-heis-1} because this lower bound
applies to all $n\in\mathbb{N}$ such that $p_{e,\|}\leq\varepsilon$.

In the case that there exists a contraction $W$ such that $M_{W}=I$,
it follows that $b_{W}=0$, and~\eqref{eq:up-bnd-err-prob-for-q-comp}
reduces to $1-\frac{\varepsilon\left(1-\varepsilon\right)}{pq}\leq na_{W}$.
Now applying Lemma~\ref{lem:quadratic-ineq-rewrite} below, we conclude
that
\begin{equation}
n\geq\left\lceil \frac{1-\frac{\varepsilon\left(1-\varepsilon\right)}{pq}}{a_{W}}\right\rceil .
\end{equation}
Since this lower bound holds for every contraction $W\in\mathbb{B}$
satisfying $M_{W}=I$, we can take a supremum over all such contractions
(which is equivalent to an infimum over the denominator) to conclude
that
\begin{equation}
n\geq\left\lceil \frac{1-\frac{\varepsilon\left(1-\varepsilon\right)}{pq}}{\inf_{W\in\mathbb{B}}\left\{ a_{W}: M_{W}=I\right\} }\right\rceil .
\end{equation}
We then conclude~\eqref{eq:lower-bnd-sql-1} because this lower bound
applies to all $n\in\mathbb{N}$ such that $p_{e,\|}\leq\varepsilon$.
\end{proof}
\begin{rem}
\label{rem:on-triviality-lower-bounds-ch-disc}The lower bounds in
\eqref{eq:lower-bnd-sql-1} and~\eqref{eq:lower-bnd-sql-2} are only
applicable in the case that there exists a contraction $W$ satisfying
$M_{W}=I$. If this is not the case, then
\begin{equation}
\inf_{W\in\mathbb{B}}\left\{ a_{W}: M_{W}=I\right\} =+\infty,
\end{equation}
implying that these lower bounds trivially evaluate to $0$. Otherwise,
the lower bounds in~\eqref{eq:lower-bnd-sql-1} and~\eqref{eq:lower-bnd-sql-2}
do not trivially evaluate to zero, and the denominator can be efficiently
computed by semi-definite optimization as follows:
\begin{multline}
\inf_{W\in\mathbb{B}}\left\{ a_{W}: M_{W}=I\right\} =\\
\inf_{\lambda\geq0,W\in\mathbb{L}}\left\{ \lambda:\lambda I\geq2\left(I-\Re\!\left[M_{W}\right]\right)\geq-\lambda I,\ \begin{bmatrix}I & W\\
W^{\dag} & I
\end{bmatrix}\geq0,\ M_{W}=I\right\} ,
\end{multline}
where $M_{W}\equiv V_{1}^{\dag}\left(W_{E}\otimes I_{B}\right)V_{2}$
and we appealed to~\eqref{eq:spectral-norm-herm-op} for $2\left\Vert I-\Re\!\left[M\right]\right\Vert $
and~\eqref{eq:contraction-as-SDP-constraint} for $\left\Vert W\right\Vert \leq1$.
\end{rem}

\begin{lem}
\label{lem:quadratic-ineq-rewrite}Suppose that the following inequality
holds: 
\begin{equation}
n\left(a+\left(n-1\right)b\right)\geq c,\label{eq:ineq-to-reduce-quadratic-formula}
\end{equation}
where $a\geq b\geq0$, $c\geq0$, and $n\in\mathbb{N}$. Then the
following lower bound holds on $n$:
\begin{equation}
n\geq\begin{cases}
\left\lceil \frac{b-a+\sqrt{\left(b-a\right)^{2}+4bc}}{2b}\right\rceil  & :b>0\\
\left\lceil \frac{c}{a}\right\rceil  & :b=0
\end{cases}.\label{eq:lower-bound-n-quadratic-formula}
\end{equation}

Alternatively, suppose that the following inequality holds:
\begin{equation}
n\left(a+\left(n-1\right)\sqrt{ab}\right)\geq c.
\end{equation}
Then the following lower bound holds on $n$:
\begin{equation}
n\geq\begin{cases}
\left\lceil \frac{\sqrt{ab}-a+\sqrt{\left(\sqrt{ab}-a\right)^{2}+4c\sqrt{ab}}}{2\sqrt{ab}}\right\rceil  & :b>0\\
\left\lceil \frac{c}{a}\right\rceil  & :b=0
\end{cases}.\label{eq:lower-bound-n-quadratic-formula-1}
\end{equation}
\end{lem}

\begin{proof}
Consider that the inequality can be rewritten as
\begin{equation}
bn^{2}+\left(a-b\right)n-c\geq0
\end{equation}
Under the assumption that $b>0$, the parabola opens upwards, so that
the inequality holds for $n\geq n_{+}$, where $n_{+}$ is the larger
root of the quadratic equation $bn^{2}+\left(a-b\right)n-c=0$. By
the quadratic formula, the two roots are given by
\begin{equation}
n_{\pm}=\frac{-\left(a-b\right)\pm\sqrt{\left(a-b\right)^{2}+4bc}}{2b},
\end{equation}
and since $b>0$, the larger root is given by $n_{+}$. Thus, for
$n\in\mathbb{N}$, the first lower bound in~\eqref{eq:lower-bound-n-quadratic-formula}
holds. In the case that $b=0$, the inequality in~\eqref{eq:ineq-to-reduce-quadratic-formula}
simplifies to $an\geq c$, from which we conclude the second expression
in~\eqref{eq:lower-bound-n-quadratic-formula}. The proof of~\eqref{eq:lower-bound-n-quadratic-formula-1}
follows from the substitution $b\to\sqrt{ab}$.
\end{proof}

\subsubsection{Finding lower bounds on query complexities of channel discrimination
by binary search}

\label{subsec:lower-bounds-binary-search}In practice, the lower bounds
in~\eqref{eq:lower-bnd-heis-1} and~\eqref{eq:lower-bnd-heis-2} of
Corollary~\ref{cor:query-comp-lower-bnd-ch-disc} may be difficult
to compute, given that they are not obviously representable as convex
optimization problems. However, it is possible to use the bounds from
Corollary~\ref{cor:error-prob-bounds-ch-disc} in order to devise
efficient algorithms for computing lower bounds on the query complexities
of channel discrimination.

Let us begin with the parallel setting. For fixed $p\in\left(0,1\right)$,
$q=1-p$, and $\varepsilon\in\left[0,\min\!\left\{ p,q\right\} \right]$,
if there exists an $n\in\mathbb{N}$ such that $p_{e,\|}\leq\varepsilon$,
then, according to~\eqref{eq:up-bnd-err-prob-for-q-comp}, the following
inequality holds:
\begin{equation}
1-\frac{\varepsilon\left(1-\varepsilon\right)}{pq}\leq ng(n),\label{eq:key-ineq-alg-query-comp-parallel}
\end{equation}
where
\begin{equation}
g(n)\coloneqq\inf_{\substack{W\in\mathbb{B}
}
}\left\{ a_{W}+\left(n-1\right)b_{W}\right\} .
\end{equation}
The function $g(n)$ can be calculated efficiently by means of a semi-definite
program (see Proposition~\ref{prop:SDP-Bures-channel-parallel}).
Furthermore, $g(n)$ is a monotone non-decreasing function on $\mathbb{N}$
because, for every fixed contraction $W$, the following inequality
holds:
\begin{equation}
a_{W}+\left(n-1\right)b_{W}\leq a_{W}+nb_{W},\label{eq:ineq-for-monotone-gn}
\end{equation}
given that $a_{W},b_{W}\geq0$ for all $W$. Applying an infimum over
every contraction $W$ to~\eqref{eq:ineq-for-monotone-gn} then implies
that $g(n)$ is monotone non-decreasing on $\mathbb{N}$, as claimed.
We furthermore conclude that the function $n\mapsto ng(n)$ is monotone
non-decreasing on $\mathbb{N}$.

Given this observation, we can find the smallest integer $n$ satisfying
\eqref{eq:key-ineq-alg-query-comp-parallel} (i.e., the query complexity
$n_{\|}^{\star}(p,\mathcal{N}_{1},q,\mathcal{N}_{2},\varepsilon)$) by
binary search. Before doing so, it is helpful to have an upper bound
on the query complexity. Toward this end, let us recall~\cite[Theorem~8]{Nuradha2025},
which implies the following upper bound on the query complexity $n_{\|}^{\star}(p,\mathcal{N}_{1},q,\mathcal{N}_{2},\varepsilon)$
by picking $s=\frac{1}{2}$ therein:
\begin{align}
n_{\|}^{\star}(p,\mathcal{N}_{1},q,\mathcal{N}_{2},\varepsilon) & \leq\left\lceil \frac{\ln\!\left(\frac{\sqrt{pq}}{\varepsilon}\right)}{Q_{\frac{1}{2}}(\mathcal{N}_{1}\|\mathcal{N}_{2})}\right\rceil ,\\
 & \leq n_{\max}\coloneqq\left\lceil \frac{\ln\!\left(\frac{\sqrt{pq}}{\varepsilon}\right)}{-\ln\sqrt{F}(\mathcal{N}_{1},\mathcal{N}_{2})}\right\rceil 
\end{align}
where
\begin{align}
Q_{\frac{1}{2}}(\mathcal{N}_{1}\|\mathcal{N}_{2}) & \coloneqq-\inf_{\psi_{RA}\in\mathbb{P}}\ln\Tr\!\left[\left[\left(\id\otimes\mathcal{N}_{1}\right)\left(\psi_{RA}\right)\right]^{\frac{1}{2}}\left[\left(\id\otimes\mathcal{N}_{2}\right)\left(\psi_{RA}\right)\right]^{\frac{1}{2}}\right],\\
F(\mathcal{N}_{1},\mathcal{N}_{2}) & \coloneqq\inf_{\psi_{RA}\in\mathbb{P}}F\!\left(\left(\id\otimes\mathcal{N}_{1}\right)\left(\psi_{RA}\right),\left(\id\otimes\mathcal{N}_{2}\right)\left(\psi_{RA}\right)\right).
\end{align}
The second inequality follows from reasoning similar to that used
to arrive at~\cite[Corollary~8]{Cheng2025} (i.e., known inequalities
between Holevo and Uhlmann fidelities, as recalled in~\cite[Eq.~(A7)]{Cheng2025}).
Furthermore, it is known how to compute the root fidelity of channels,
$\sqrt{F}(\mathcal{N}_{1},\mathcal{N}_{2})$, efficiently by means
of semi-definite optimization (see~\cite[Proposition~55]{Katariya2021}
or Proposition~\ref{prop:root-fidelity-alt-SDP} below). As such,
the upper bound $n_{\max}$ can be computed efficiently.

With this background in place, we now detail the following algorithm:
\begin{lyxalgorithm}
\label{alg:lower-bnd-query-comp-parallel}The binary search algorithm
for finding a lower bound on $n_{\|}^{\star}(p,\mathcal{N}_{1},q,\mathcal{N}_{2},\varepsilon)$
proceeds according to the following steps:
\begin{enumerate}
\item Initialize $n_{\ell}\leftarrow1$, $n_{h}\leftarrow n_{\max}$.
\item While $n_{\ell}<n_{h}$,
\begin{enumerate}
\item Set $n_{m}\leftarrow\left\lfloor \frac{n_{\ell}+n_{h}}{2}\right\rfloor $.
\item Solve $g(n_{m})=\inf_{W\in\mathbb{B}}\left\{ a_{W}+\left(n_{m}-1\right)b_{W}\right\} $
by semi-definite optimization (see Proposition~\ref{prop:SDP-Bures-channel-parallel}).
\item If $n_{m}g(n_{m})\geq1-\frac{\varepsilon\left(1-\varepsilon\right)}{pq}$,
set $n_{h}\leftarrow n_{m}$. Otherwise, set $n_{\ell}\leftarrow n_{m}+1$.
\end{enumerate}
\item Return $n_{\ell}$.
\end{enumerate}
\end{lyxalgorithm}

The total number of iterations for Algorithm~\ref{alg:lower-bnd-query-comp-parallel}
is equal to $O(\log n_{\max})$.

A similar procedure can be conducted for finding a lower bound on
query complexity in the adaptive setting. Although the upper bound
in~\eqref{eq:adaptive-error-prob-up-bnd} does not correspond to an
SDP due to the presence of the square root, in practice, it can still
be calculated efficiently by using the procedure outlined in Remark
\ref{rem:adaptive-ch-disc-opt}, which involves a combination of semi-definite
optimization and one-dimensional grid search. For completeness, we
provide the following algorithm for efficiently computing a lower
bound on the query complexity of adaptive channel discrimination:
\begin{lyxalgorithm}
\label{alg:lower-bnd-query-comp-adaptive}The binary search algorithm
for finding a lower bound on $n_{\adaptivedisc}^{\star}(p,\mathcal{N}_{1},q,\mathcal{N}_{2},\varepsilon)$
proceeds according to the following steps:
\begin{enumerate}
\item Initialize $n_{\ell}\leftarrow1$, $n_{h}\leftarrow n_{\max}$.
\item While $n_{\ell}<n_{h}$,
\begin{enumerate}
\item Set $n_{m}\leftarrow\left\lfloor \frac{n_{\ell}+n_{h}}{2}\right\rfloor $.
\item Solve $g_{\adaptivedisc}(n_{m})=\inf_{W\in\mathbb{B}}\left\{ a_{W}+\left(n-1\right)\sqrt{a_{W}b_{W}}\right\} $
by the procedure outlined in Remark~\ref{rem:adaptive-ch-disc-opt}.
\item If $n_{m}g_{\adaptivedisc}(n_{m})\geq1-\frac{\varepsilon\left(1-\varepsilon\right)}{pq}$,
set $n_{h}\leftarrow n_{m}$. Otherwise, set $n_{\ell}\leftarrow n_{m}+1$.
\end{enumerate}
\item Return $n_{\ell}$.
\end{enumerate}
\end{lyxalgorithm}

Similar to Algorithm~\ref{alg:lower-bnd-query-comp-parallel}, the
total number of iterations of Algorithm~\ref{alg:lower-bnd-query-comp-adaptive}
is equal to $O(\log n_{\max})$.

Let us finally note that both of these algorithms are applicable and
useful to determine lower bounds on the minimum number of queries
needed to achieve perfect channel discrimination (i.e., $\varepsilon=0$).
Indeed, Ref.~\cite{Duan2009} established conditions under which
two channels can be discriminated perfectly. By setting $\varepsilon=0$
in Algorithms~\ref{alg:lower-bnd-query-comp-parallel} and~\ref{alg:lower-bnd-query-comp-adaptive},
we can thus determine lower bounds on the query complexities of perfect
channel discrimination in both the parallel and adaptive access models,
respectively.

\subsection{Lower bounds on query complexities of channel estimation}

\label{subsec:Lower-bounds-query-compl-ch-estimation}In this section,
we leverage the lower bounds from~\eqref{eq:parallel-est-minimax-lower-bnd}--\eqref{eq:adaptive-est-minimax-lower-bnd},
along with our lower bounds on the error probabilities and query complexities
of channel discrimination from Section~\ref{subsec:Lower-bounds-query-comp-ch-disc},
in order to establish lower bounds on the minimax error probabilities
and query complexities of channel estimation. We then establish asymptotic
lower bounds on these quantities in terms of SLD Fisher information,
arriving at results that generalize those from~\cite[Eq.~(53)]{Meyer2025} (see \eqref{lower-bnd-heis-ch-est-1}--\eqref{eq:lower-bnd-sql-ch-est-2}). These asymptotic lower bounds bear similarity to quantum Cramer--Rao bounds for channel estimation, which have appeared in the context of quantum estimation theory (see, e.g., \cite{DemkowiczDobrzanski2012,Kolodynski2013,Demkowicz2014,Zhou2021,Kurdzialek2023}).

\begin{cor}
\label{cor:minimax-error-ch-est}For a parameterized family $\mathcal{F}\equiv\left(\mathcal{N}_{\theta}\right)_{\theta\in \Theta}$
of channels, the following bounds hold for the minimax error probability
of channel estimation:
\begin{align}
1-4p_{e,\|}\left(1-p_{e,\|}\right) & \leq n\inf_{\substack{
\theta,\theta'\in\Theta,W\in\mathbb{B}\\
\left|\theta-\theta'\right|>2\delta
}
}\left\{ a_{W}^{\theta,\theta'}+\left(n-1\right)b_{W}^{\theta,\theta'}\right\} ,\\
1-4p_{e,\adaptivedisc}\left(1-p_{e,\adaptivedisc}\right) & \leq n\inf_{\substack{\theta,\theta'\in\Theta,W\in\mathbb{B}\\
\left|\theta-\theta'\right|>2\delta
}
}\left\{ a_{W}^{\theta,\theta'}+\left(n-1\right)\sqrt{a_{W}^{\theta,\theta'}b_{W}^{\theta,\theta'}}\right\} ,
\end{align}
where
\begin{align}
p_{e,\|} & \equiv p_{e,\|}(\delta,\mathcal{F},n),\\
p_{e,\adaptivedisc} & \equiv p_{e,\adaptivedisc}(\delta,\mathcal{F},n),\\
a_{W}^{\theta,\theta'} & \equiv2\left\Vert I-\Re\!\left[M_{W}^{\theta,\theta'}\right]\right\Vert ,\\
b_{W}^{\theta,\theta'} & \equiv\left\Vert I-M_{W}^{\theta,\theta'}\right\Vert ^{2},\\
M_{W}^{\theta,\theta'} & \equiv V_{\theta}^{\dag}\left(W_{E}\otimes I_{B}\right)V_{\theta'},
\end{align}
$V_{\theta}\colon\mathbb{C}^{d_{A}}\mapsto\mathbb{C}^{d_{E}}\otimes\mathbb{C}^{d_{B}}$
is a canonical isometric extension of $\mathcal{N}_{\theta}$ of the
form in~\eqref{eq:canonical-isometric-ext-smooth-family}, and $p_{e,\|}$ and $p_{e,\adaptivedisc}$ are defined in Definitions
\ref{def:minimax-parallel-est} and~\ref{def:minimax-adaptive-est},
respectively.
\end{cor}

\begin{proof}
These inequalities follow by combining~\eqref{eq:parallel-est-minimax-lower-bnd}
and~\eqref{eq:adaptive-est-minimax-lower-bnd} with Corollary~\ref{cor:error-prob-bounds-ch-disc},
while noting that the function $x\mapsto x\left(1-x\right)$ is monotone
increasing on the interval $\left[0,\frac{1}{2}\right]$.
\end{proof}
\begin{cor}
\label{cor:query-comp-ch-est}For a parameterized family $\mathcal{F}\equiv\left(\mathcal{N}_{\theta}\right)_{\theta \in \Theta}$
of channels and for all $\varepsilon\in\left(0,\frac{1}{2}\right]$,
the following lower bounds hold for the query complexities of channel
estimation:
\begin{align}
n_{\|}^{\star}(\varepsilon,\delta,\mathcal{F}) & \geq\left\lceil \sup_{\substack{\theta,\theta'\in\Theta, W\in\mathbb{B}, \\ \left|\theta-\theta'\right|>2\delta
}
}\frac{b_{W}^{\theta,\theta'}-a_{W}^{\theta,\theta'}+\sqrt{\left(b_{W}^{\theta,\theta'}-a_{W}^{\theta,\theta'}\right)^{2}+4b_{W}^{\theta,\theta'}\left(1-4\varepsilon\left(1-\varepsilon\right)\right)}}{2b_{W}^{\theta,\theta'}}\right\rceil \\
 & \geq\left\lceil \frac{1-4\varepsilon\left(1-\varepsilon\right)}{\inf_{W\in\mathbb{B}}\left\{ a_{W}^{\theta,\theta'}: M_{W}^{\theta,\theta'}=I,\ \left|\theta-\theta'\right|>2\delta,\ \theta,\theta'\in\Theta\right\} }\right\rceil ,\label{eq:lower-bnd-sql-est-1}\\
n_{\adaptivedisc}^{\star}(\varepsilon,\delta,\mathcal{F}) & \geq\left\lceil \sup_{\substack{\theta,\theta'\in\Theta, W\in\mathbb{B}, \\ \left|\theta-\theta'\right|>2\delta
}
}\frac{c_{W}^{\theta,\theta'}+\sqrt{\left(c_{W}^{\theta,\theta'}\right)^{2}+4\sqrt{a_{W}^{\theta,\theta'}b_{W}^{\theta,\theta'}}\left(1-4\varepsilon\left(1-\varepsilon\right)\right)}}{2\sqrt{a_{W}^{\theta,\theta'}b_{W}^{\theta,\theta'}}}\right\rceil \\
 & \geq\left\lceil \frac{1-4\varepsilon\left(1-\varepsilon\right)}{\inf_{W\in\mathbb{B}}\left\{ a_{W}^{\theta,\theta'}: M_{W}^{\theta,\theta'}=I,\ \left|\theta-\theta'\right|>2\delta,\ \theta,\theta'\in\Theta\right\} }\right\rceil ,\label{eq:lower-bnd-sql-est-2}
\end{align}
where
\begin{align}
a_{W}^{\theta,\theta'} & \equiv2\left\Vert I-\Re\!\left[M_{W}^{\theta,\theta'}\right]\right\Vert ,\\
b_{W}^{\theta,\theta'} & \equiv\left\Vert I-M_{W}^{\theta,\theta'}\right\Vert ^{2},\\
c_{W}^{\theta,\theta'} & \equiv\sqrt{a_{W}^{\theta,\theta'}b_{W}^{\theta,\theta'}}-a_{W}^{\theta,\theta'},\\
M_{W}^{\theta,\theta'} & \equiv V_{\theta}^{\dag}\left(W_{E}\otimes I_{B}\right)V_{\theta'},
\end{align}
and $V_{\theta}\colon\mathbb{C}^{d_{A}}\mapsto\mathbb{C}^{d_{E}}\otimes\mathbb{C}^{d_{B}}$
is a canonical isometric extension of $\mathcal{N}_{\theta}$ of the
form in~\eqref{eq:canonical-isometric-ext-smooth-family}.
\end{cor}

\begin{proof}
These inequalities follow by combining~\eqref{eq:parallel-est-query-comp-lower-bnd}
and~\eqref{eq:adaptive-est-query-comp-lower-bnd} with Corollary~\ref{cor:query-comp-lower-bnd-ch-disc},
while noting that the function $x\mapsto x\left(1-x\right)$ is monotone
increasing on the interval $\left[0,\frac{1}{2}\right]$.
\end{proof}
\begin{rem}
Similar to what was stated in Remark~\ref{rem:on-triviality-lower-bounds-ch-disc},
the lower bounds in~\eqref{eq:lower-bnd-sql-est-1} and~\eqref{eq:lower-bnd-sql-est-2}
are only applicable in the case that, for some $\theta,\theta'\in\Theta$
satisfying $\left|\theta-\theta'\right|>2\delta$, there exists $W\in\mathbb{B}$
such that $M_{W}^{\theta,\theta'}=I$. Otherwise these lower bounds
are trivially equal to $0$.
\end{rem}

Finally, we obtain asymptotic expressions for bounds on the minimax
error probabilities and query complexities by appealing to~\eqref{eq:bures-parallel-err-prob-bnd},
\eqref{eq:bures-adaptive-err-prob-bnd},~\eqref{eq:parallel-est-minimax-lower-bnd},
\eqref{eq:adaptive-est-minimax-lower-bnd}, and Theorems~\ref{thm:smooth-family-SLD-Fisher-ch-expand},~\ref{thm:upper-bound-SLD-Fisher-parallel}, and~\ref{thm:upper-bound-SLD-Fisher-adaptive}.\begin{cor}
\label{cor:minimax-error-query-comp-SLD-Fisher}For a smooth family
$\mathcal{F}\equiv\left(\mathcal{N}_{\theta}\right)$ of channels
and $\varepsilon\in\left(0,\frac{1}{2}\right]$, the following bounds
hold for the minimax error probabilities and query complexities of
channel estimation:
\begin{align}
1-4p_{e,\|}\left(1-p_{e,\|}\right) & \leq\frac{\delta^{2}}{4}n\inf_{\substack{\theta\in\Theta,\\
\left(H_{\theta}\right)_{\theta\in\Theta}
}
}\left\{ a_{\theta}+\left(n-1\right)b_{\theta}\right\} +o(\delta^{2}),\label{eq:small-delta-parallel-bound}\\
1-4p_{e,\adaptivedisc}\left(1-p_{e,\adaptivedisc}\right) & \leq\frac{\delta^{2}}{4}n\inf_{\substack{\theta\in\Theta,\\
\left(H_{\theta}\right)_{\theta\in\Theta}
}
}\left\{ a_{\theta}+\left(n-1\right)\sqrt{a_{\theta}b_{\theta}}\right\} +o(\delta^{2}),\label{eq:small-delta-adaptive-bound}
\end{align}
\begin{align}
n_{\|}^{\star}(\varepsilon,\delta,\mathcal{F}) & \geq O\!\left(\frac{1}{\delta}\sqrt{\frac{1-4\varepsilon\left(1-\varepsilon\right)}{\inf_{\theta\in\Theta,\left(H_{\theta}\right)_{\theta\in\Theta}}\left\{ b_{\theta}\right\} }}\right)\label{lower-bnd-heis-ch-est-1}\\
 & \geq O\!\left(\frac{1}{\delta^{2}}\left(\frac{1-4\varepsilon\left(1-\varepsilon\right)}{\inf_{\theta\in\Theta}\left\{ a_{\theta}:V_{\theta}^{\dag}M_{H_{\theta}}=0,\ \left(H_{\theta}\right)_{\theta\in\Theta}\right\} }\right)\right),\label{eq:lower-bnd-sql-ch-est-1}\\
n_{\adaptivedisc}^{\star}(\varepsilon,\delta,\mathcal{F}) & \geq O\!\left(\frac{1}{\delta}\sqrt{\frac{1-4\varepsilon\left(1-\varepsilon\right)}{\inf_{\theta\in\Theta,\left(H_{\theta}\right)_{\theta\in\Theta}}\left\{ \sqrt{a_{\theta}b_{\theta}}\right\} }}\right)\label{eq:lower-bnd-heis-ch-est-2-2}\\
 & \geq O\!\left(\frac{1}{\delta^{2}}\left(\frac{1-4\varepsilon\left(1-\varepsilon\right)}{\inf_{\theta\in\Theta}\left\{ a_{\theta}:V_{\theta}^{\dag}M_{H_{\theta}}=0,\ \left(H_{\theta}\right)_{\theta\in\Theta}\right\} }\right)\right),\label{eq:lower-bnd-sql-ch-est-2}
\end{align}
where
\begin{align}
p_{e,\|} & \equiv p_{e,\|}(\delta,\mathcal{F},n),\\
p_{e,\adaptivedisc} & \equiv p_{e,\adaptivedisc}(\delta,\mathcal{F},n),\\
a_{\theta} & \equiv\left\Vert M_{H_{\theta}}\right\Vert ^{2},\\
b_{\theta} & \equiv\left\Vert V_{\theta}^{\dag}M_{H_{\theta}}\right\Vert ^{2},\\
M_{H_{\theta}} & \equiv\left(\partial_{\theta}V_{\theta}\right)-i\left(H_{\theta}\otimes I_{B}\right)V_{\theta},
\end{align}
$V_{\theta}\colon\mathbb{C}^{d_{A}}\mapsto\mathbb{C}^{d_{E}}\otimes\mathbb{C}^{d_{B}}$
is a canonical isometric extension of $\mathcal{N}_{\theta}$ of the
form in~\eqref{eq:canonical-isometric-ext-smooth-family} and the
optimizations are over every smooth family $\left(H_{\theta}\right)_{\theta\in\Theta}$
of Hermitian operators.
\end{cor}

\begin{proof}
The inequalities in~\eqref{eq:small-delta-parallel-bound} and~\eqref{eq:small-delta-adaptive-bound}
are indeed direct consequences of~\eqref{eq:bures-parallel-err-prob-bnd},
\eqref{eq:bures-adaptive-err-prob-bnd},~\eqref{eq:parallel-est-minimax-lower-bnd},
\eqref{eq:adaptive-est-minimax-lower-bnd}, and Theorem~\ref{thm:smooth-family-SLD-Fisher-ch-expand},
as stated.

To obtain the query complexity lower bounds in~\eqref{eq:lower-bnd-heis-1}--\eqref{eq:lower-bnd-sql-2},
we apply Lemma~\ref{lem:quadratic-ineq-rewrite}. Let us begin with
the parallel setting. Let $\theta\in\Theta$ and let $\left(H_{\theta}\right)_{\theta\in\Theta}$
be a smooth family of Hermitian operators. Then the following bound
holds, as a consequence of~\eqref{eq:small-delta-parallel-bound}:
\begin{equation}
1-4p_{e,\|}\left(1-p_{e,\|}\right)\leq\frac{n}{4}\left(\delta^{2}a_{\theta}+\left(n-1\right)\delta^{2}b_{\theta}\right)+o(\delta^{2}).\label{eq:upper-bound-success-est-small-delta}
\end{equation}
Ignoring the $o(\delta^{2})$ term, as we are considering the limit
$\delta\to0$, and applying Lemma~\ref{lem:quadratic-ineq-rewrite},
we conclude that
\begin{align}
n & \geq4\left(\frac{\delta^{2}b_{\theta}-\delta^{2}a_{\theta}+\sqrt{\left(\delta^{2}b_{\theta}-\delta^{2}a_{\theta}\right)^{2}+4\delta^{2}b_{\theta}\left(1-4\varepsilon\left(1-\varepsilon\right)\right)}}{2\delta^{2}b_{\theta}}\right)\\
 & =\frac{4}{\delta}\left(\frac{\delta b_{\theta}-\delta a_{\theta}+\sqrt{\left(\delta b_{\theta}-\delta a_{\theta}\right)^{2}+4b_{\theta}\left(1-4\varepsilon\left(1-\varepsilon\right)\right)}}{2b_{\theta}}\right)\\
 & =O\!\left(\frac{1}{\delta}\sqrt{\frac{1-4\varepsilon\left(1-\varepsilon\right)}{b_{\theta}}}\right).
\end{align}
Since the bound holds for every $\theta\in\Theta$ and $\left(H_{\theta}\right)_{\theta\in\Theta}$
and every $n\in\mathbb{N}$ such that $p_{e,\|}\leq\varepsilon$,
we can take an infimum over both to conclude that
\begin{equation}
n_{\|}^{\star}(\varepsilon,\delta,\mathcal{F})\geq O\!\left(\frac{1}{\delta}\sqrt{\frac{1-4\varepsilon\left(1-\varepsilon\right)}{\inf_{\theta\in\Theta,\left(H_{\theta}\right)_{\theta\in\Theta}}\left\{ b_{\theta}\right\} }}\right),
\end{equation}
which proves~\eqref{lower-bnd-heis-ch-est-1}.

In the case that there exists $\theta\in\Theta$ and a smooth family
$\left(H_{\theta}\right)_{\theta\in\Theta}$ of Hermitian operators
such that $b_{\theta}=0$, then the upper bound in~\eqref{eq:upper-bound-success-est-small-delta}
becomes
\begin{equation}
1-4p_{e,\|}\left(1-p_{e,\|}\right)\leq\frac{n}{4}\delta^{2}a_{\theta}+o(\delta^{2}).
\end{equation}
Now flipping this inequality around and applying a similar line of
reasoning, we conclude~\eqref{eq:lower-bnd-sql-ch-est-1}.

We now argue for the adaptive setting, with a similar line of reasoning.
Let $\theta\in\Theta$ and let $\left(H_{\theta}\right)_{\theta\in\Theta}$
be a smooth family of Hermitian operators. Then the following bound
holds, as a consequence of~\eqref{eq:small-delta-adaptive-bound}:
\begin{align}
1-4p_{e,\adaptivedisc}\left(1-p_{e,\adaptivedisc}\right) & \leq\frac{n}{4}\delta^{2}\left(a_{\theta}+\left(n-1\right)\sqrt{a_{\theta}b_{\theta}}\right)+o(\delta^{2})\\
 & =\frac{n}{4}\left(\delta^{2}a_{\theta}+\left(n-1\right)\sqrt{\left(\delta^{2}a_{\theta}\right)\left(\delta^{2}b_{\theta}\right)}\right)+o(\delta^{2})\label{eq:upper-bound-success-est-small-delta-adaptive}
\end{align}
Ignoring the $o(\delta^{2})$ term, as we are considering the limit
$\delta\to0$, and applying Lemma~\ref{lem:quadratic-ineq-rewrite},
we conclude that
\begin{align}
n & \geq4\left(\frac{\sqrt{\delta^{2}a_{\theta}\delta^{2}b_{\theta}}-\delta^{2}a_{\theta}+\sqrt{\left(\sqrt{\delta^{2}a_{\theta}\delta^{2}b_{\theta}}-\delta^{2}a_{\theta}\right)^{2}+4\sqrt{\delta^{2}a_{\theta}\delta^{2}b_{\theta}}\left(1-4\varepsilon\left(1-\varepsilon\right)\right)}}{2\sqrt{\delta^{2}a_{\theta}\delta^{2}b_{\theta}}}\right)\\
 & =\frac{4}{\delta}\left(\frac{\sqrt{\delta a_{\theta}\delta b_{\theta}}-\delta a_{\theta}+\sqrt{\left(\sqrt{\delta a_{\theta}\delta b_{\theta}}-\delta a_{\theta}\right)^{2}+4\sqrt{a_{\theta}b_{\theta}}\left(1-4\varepsilon\left(1-\varepsilon\right)\right)}}{2\sqrt{a_{\theta}b_{\theta}}}\right)\\
 & =O\!\left(\frac{1}{\delta}\sqrt{\frac{\left(1-4\varepsilon\left(1-\varepsilon\right)\right)}{\sqrt{a_{\theta}b_{\theta}}}}\right).
\end{align}
Since the bound holds for every $\theta\in\Theta$ and $\left(H_{\theta}\right)_{\theta\in\Theta}$
and every $n\in\mathbb{N}$ such that $p_{e,\adaptivedisc}\leq\varepsilon$,
we can take an infimum over both to conclude that
\begin{equation}
n_{\adaptivedisc}^{\star}(\varepsilon,\delta,\mathcal{F})\geq O\!\left(\frac{1}{\delta}\sqrt{\frac{1-4\varepsilon\left(1-\varepsilon\right)}{\inf_{\theta\in\Theta,\left(H_{\theta}\right)_{\theta\in\Theta}}\left\{ \sqrt{a_{\theta}b_{\theta}}\right\} }}\right),
\end{equation}
which is~\eqref{eq:lower-bnd-heis-ch-est-2-2}.

In the case that there exists $\theta\in\Theta$ and a smooth family
$\left(H_{\theta}\right)_{\theta\in\Theta}$ of Hermitian operators
such that $b_{\theta}=0$, then the upper bound in~\eqref{eq:upper-bound-success-est-small-delta-adaptive}
becomes
\begin{equation}
1-4p_{e,\|}\left(1-p_{e,\|}\right)\leq\frac{n}{4}\delta^{2}a_{\theta}+o(\delta^{2}).
\end{equation}
Now flipping this inequality around and applying a similar line of
reasoning, we conclude~\eqref{eq:lower-bnd-sql-ch-est-2}.
\end{proof}

\begin{rem}
\label{rem:on-triviality-lower-bounds-ch-est-smooth}The lower bounds
in~\eqref{eq:lower-bnd-sql-ch-est-1} and~\eqref{eq:lower-bnd-sql-ch-est-2}
are only applicable in the case that there exists $\theta\in\Theta$
and a Hermitian operator $H_{\theta}$ such that $V_{\theta}^{\dag}M_{H_{\theta}}=0$. If this is the case, then Heisenberg scaling does not occur.
If this is not the case, then
\begin{equation}
\inf_{\theta\in\Theta}\left\{ a_{\theta}:V_{\theta}^{\dag}M_{H_{\theta}}=0,\ \left(H_{\theta}\right)_{\theta\in\Theta}\right\} =+\infty,
\end{equation}
implying that these lower bounds trivially evaluate to $0$. Otherwise,
the lower bounds in~\eqref{eq:lower-bnd-sql-ch-est-1} and~\eqref{eq:lower-bnd-sql-ch-est-2}
do not trivially evaluate to zero, and the denominator can be efficiently
computed by a combination of one-dimensional grid search over $\theta\in\Theta$
and semi-definite optimization as follows:
\begin{multline}
\inf_{H_{\theta}\in\mathbb{H}}\left\{ a_{\theta}:V_{\theta}^{\dag}M_{H_{\theta}}=0\right\} =\\
\inf_{\lambda\geq0}\left\{ \lambda:\begin{bmatrix}\lambda I & M_{H_{\theta}}\\
M_{H_{\theta}}^{\dag} & I
\end{bmatrix}\geq0,\ M_{H_{\theta}}=\left(\partial_{\theta}V_{\theta}\right)-i\left(H_{\theta}\otimes I_{B}\right)V_{\theta},\ V_{\theta}^{\dag}M_{H_{\theta}}=0\right\} ,
\end{multline}
where we appealed to~\eqref{eq:SDP-inf-norm-squared-general} for
$a_{\theta}=\left\Vert M_{H_{\theta}}\right\Vert ^{2}$.
\begin{rem}
Similar to the approaches outlined in Section~\ref{subsec:lower-bounds-binary-search},
lower bounds on the query complexities of channel estimation can be
efficiently computed by a combination of semi-definite optimization
and grid search. Indeed, from the bounds in Corollary~\ref{cor:minimax-error-ch-est}
and algorithms similar to Algorithms~\ref{alg:lower-bnd-query-comp-parallel}
and~\ref{alg:lower-bnd-query-comp-adaptive}, combined with an outer
search over $\theta,\theta'\in\Theta$ satisfying $\left|\theta-\theta'\right|>2\delta$,
one can efficiently compute lower bounds on the query complexities
of channel estimation. One can obtain looser bounds at lower computational
cost by performing an outer search over just $\theta\in\Theta$ and
setting $\theta'=\theta+2\delta$. For sufficiently small $\delta$,
one could alternatively obtain lower bounds on the query complexities
of channel estimation using the bounds from Corollary~\ref{cor:minimax-error-query-comp-SLD-Fisher}
and algorithms similar to Algorithms~\ref{alg:lower-bnd-query-comp-parallel}
and~\ref{alg:lower-bnd-query-comp-adaptive}, combined with an outer
search over $\theta\in\Theta$.
\end{rem}

\end{rem}

\section{Optimization of channel discrimination and estimation measures}\label{sec:Optimization-of-channel-disc-est}

In this section, we show how to compute all of the channel measures
in our paper in terms of semi-definite optimization alone or semi-definite
optimization combined with grid search. One of the main tools that
we employ for this purpose is the Schur complement lemma.

\subsection{Schur complement lemma and implications for optimization}

The Schur complement lemma states that the following equivalence holds:
\begin{equation}
\begin{bmatrix}K & M\\
M^{\dag} & L
\end{bmatrix}>0\qquad\Longleftrightarrow\qquad L>0,\quad K>M^{\dag}L^{-1}M.
\end{equation}
It can be extended to positive semi-definite constraints in the following
way:
\begin{align}
\begin{bmatrix}K & M\\
M^{\dag} & L
\end{bmatrix}\geq0\qquad & \Longleftrightarrow\qquad\begin{bmatrix}K+\delta I & M\\
M^{\dag} & L+\varepsilon I
\end{bmatrix}>0\quad\forall\delta,\varepsilon>0\\
\qquad & \Longleftrightarrow\qquad L+\varepsilon I>0,\quad K+\delta I>M^{\dag}\left(L+\varepsilon I\right)^{-1}M\quad\forall\delta,\varepsilon>0\\
\qquad & \Longleftrightarrow\qquad L\geq0,\quad K\geq M^{\dag}\left(L+\varepsilon I\right)^{-1}M\quad\forall\varepsilon>0.
\end{align}
If $\imagem(M)\perp\ker(L)$, then
\begin{equation}
\lim_{\varepsilon\to0^{+}}M^{\dag}\left(L+\varepsilon I\right)^{-1}M=M^{\dag}L^{-1}M,
\end{equation}
where the inverse on the right-hand side is understood to be a generalized
inverse, taken on the support of $L$, so that
\begin{equation}
\begin{bmatrix}K & M\\
M^{\dag} & L
\end{bmatrix}\geq0\qquad\Longleftrightarrow\qquad L\geq0,\quad K\geq M^{\dag}L^{-1}M,
\end{equation}
whenever $\imagem(M)\perp\ker(L)$.

The Schur complement lemma leads to the following conclusions for
$A\in\mathbb{L}$:
\begin{align}
\begin{bmatrix}\lambda I & A\\
A^{\dag} & I
\end{bmatrix} & \geq0\qquad\Longleftrightarrow\qquad\lambda I\geq A^{\dag}A,\\
\begin{bmatrix}I & A\\
A^{\dag} & I
\end{bmatrix} & \geq0\qquad\Longleftrightarrow\qquad1\geq\left\Vert A\right\Vert ,\label{eq:contraction-as-SDP-constraint}\\
\begin{bmatrix}\lambda I & A\\
A^{\dag} & \lambda I
\end{bmatrix} & \geq0\qquad\Longleftrightarrow\qquad\lambda\geq0,\quad\lambda^{2}I\geq A^{\dag}A,
\end{align}
with the latter condition following because
\begin{align}
\begin{bmatrix}\lambda I & A\\
A^{\dag} & \lambda I
\end{bmatrix}\geq0\qquad\Longleftrightarrow & \qquad\lambda I\geq0,\quad\lambda I\geq A^{\dag}\left(\lambda I+\varepsilon I\right)^{-1}A\quad\forall\varepsilon>0,\\
\Longleftrightarrow & \qquad\lambda\geq0,\quad\lambda(\lambda+\varepsilon)I\geq A^{\dag}A\quad\forall\varepsilon>0,\\
\Longleftrightarrow & \qquad\lambda\geq0,\quad\lambda^{2}I\geq A^{\dag}A.
\end{align}
Given that $\left\Vert A\right\Vert ^{2}=\lambda_{\max}(A^{\dag}A)$
and $\left\Vert A\right\Vert =\sqrt{\lambda_{\max}(A^{\dag}A)}$,
we arrive at the following semi-definite programs (SDPs):
\begin{align}
\left\Vert A\right\Vert ^{2} & =\inf_{\lambda\geq0}\left\{ \lambda:\begin{bmatrix}\lambda I & A\\
A^{\dag} & I
\end{bmatrix}\geq0\right\} ,\label{eq:SDP-inf-norm-squared-general}\\
\left\Vert A\right\Vert  & =\inf_{\lambda\geq0}\left\{ \lambda:\begin{bmatrix}\lambda I & A\\
A^{\dag} & \lambda I
\end{bmatrix}\geq0\right\} .\label{eq:SDP-inf-norm-general}
\end{align}
If the matrix of interest, now denoted by $B$, is either positive
semi-definite or Hermitian, SDP characterizations of $\left\Vert B\right\Vert $
can be simplified along the following lines:
\begin{align}
B & \geq0\quad\implies\quad\left\Vert B\right\Vert =\inf_{\lambda\geq0}\left\{ \lambda:B\leq\lambda I\right\} ,\\
B & \in\mathbb{H}\quad\implies\quad\left\Vert B\right\Vert =\inf_{\lambda\geq0}\left\{ \lambda:-\lambda I\leq B\leq\lambda I\right\} .\label{eq:spectral-norm-herm-op}
\end{align}

\subsection{Optimization of quantities based on Bures distance of channels}

In this section, we delineate methods for optimizing various channel
distinguishability measures related to the Bures distance.

SDPs for the root fidelity of channels were established by~\cite[Eq.~(7)]{Yuan2017}
and~\cite[Proposition~55]{Katariya2021}. Proposition~\ref{prop:root-fidelity-alt-SDP}
gives an alternative semi-definite program for computing the root fidelity
of channels.
\begin{prop}
\label{prop:root-fidelity-alt-SDP}The root fidelity of channels $\mathcal{N}_{1}$
and $\mathcal{N}_{2}$ can be expressed as the following semi-definite
program:
\begin{equation}
\sqrt{F}(\mathcal{N}_{1},\mathcal{N}_{2})=\sup_{\lambda\geq0,W\in\mathbb{L}}\left\{ \lambda:\Re\!\left[V_{1}^{\dag}\left(W_{E}\otimes I_{B}\right)V_{2}\right]\geq\lambda I,\ \begin{bmatrix}I & W\\
W^{\dag} & I
\end{bmatrix}\geq0\right\} .
\end{equation}
\end{prop}

\begin{proof}
Beginning with the developments in~\eqref{eq:root-fid-chs-bures-distance}
and~\eqref{eq:reduction-fidelity-chs-1}--\eqref{eq:reduction-fidelity-chs-last},
consider that
\begin{align}
\sqrt{F}(\mathcal{N}_{1},\mathcal{N}_{2}) & =\sup_{W\in\mathbb{B}}\inf_{\rho}\Re\!\left[\Tr\!\left[V_{1}^{\dag}\left(W_{E}\otimes I_{B}\right)V_{2}\rho\right]\right]\\
 & =\sup_{W\in\mathbb{B}}\inf_{\rho}\Tr\!\left[\Re\!\left[V_{1}^{\dag}\left(W_{E}\otimes I_{B}\right)V_{2}\right]\rho\right]\\
 & =\sup_{W\in\mathbb{B}}\sup_{\lambda\in\mathbb{R}}\left\{ \lambda:\Re\!\left[V_{1}^{\dag}\left(W_{E}\otimes I_{B}\right)V_{2}\right]\geq\lambda I\right\} \\
 & =\sup_{W\in\mathbb{B}}\sup_{\lambda\geq0}\left\{ \lambda:\Re\!\left[V_{1}^{\dag}\left(W_{E}\otimes I_{B}\right)V_{2}\right]\geq\lambda I\right\} \\
 & =\sup_{\lambda\geq0,W\in\mathbb{L}}\left\{ \lambda:\Re\!\left[V_{1}^{\dag}\left(W_{E}\otimes I_{B}\right)V_{2}\right]\geq\lambda I,\ \begin{bmatrix}I & W\\
W^{\dag} & I
\end{bmatrix}\geq0\right\} ,
\end{align}
where we employed~\eqref{eq:contraction-as-SDP-constraint}.
\end{proof}
\begin{prop}
The squared Bures distance of channels $\mathcal{N}_{1}$ and $\mathcal{N}_{2}$
can be expressed as the following SDP:
\begin{equation}
d_{B}^{2}(\mathcal{N}_{1},\mathcal{N}_{2})=\inf_{\substack{\lambda\in\mathbb{R},
W\in\mathbb{L}
}
}\left\{ \begin{array}{c}
\lambda:\lambda I\geq2\left(I-\Re\!\left[M\right]\right)\geq-\lambda I,\ \begin{bmatrix}I & W\\
W^{\dag} & I
\end{bmatrix}\geq0,\\
M=V_{1}^{\dag}\left(W_{E}\otimes I_{B}\right)V_{2}
\end{array}\right\} .
\end{equation}
\end{prop}

\begin{proof}
Starting from~\eqref{eq:bures-distance-chs-reduced} and applying
\eqref{eq:spectral-norm-herm-op} for $2\left\Vert I-\Re\!\left[M\right]\right\Vert $
and~\eqref{eq:contraction-as-SDP-constraint} for $\left\Vert W\right\Vert \leq1$,
consider that
\begin{align}
d_{B}^{2}(\mathcal{N}_{1},\mathcal{N}_{2}) & =\inf_{W\in\mathbb{B}}\left\Vert 2\left(I-\Re\!\left[V_{1}^{\dag}\left(W_{E}\otimes I_{B}\right)V_{2}\right]\right)\right\Vert \\
 & =\inf_{\substack{\lambda\in\mathbb{R},\\
W\in\mathbb{L}
}
}\left\{ \begin{array}{c}
\lambda:\lambda I\geq2\left(I-\Re\!\left[V_{1}^{\dag}\left(W_{E}\otimes I_{B}\right)V_{2}\right]\right)\geq-\lambda I,\\
\begin{bmatrix}I & W\\
W^{\dag} & I
\end{bmatrix}\geq0
\end{array}\right\} ,
\end{align}
thus concluding the proof.
\end{proof}
\begin{prop}
\label{prop:SDP-Bures-channel-parallel}The upper bound in Theorem
\ref{thm:parallel-bures-dist-channels} for parallel discrimination
of channels $\mathcal{N}_{1}$ and $\mathcal{N}_{2}$ can be expressed
as the following SDP:
\begin{multline}
\inf_{W,M\in\mathbb{L}:\left\Vert W\right\Vert \leq1}\left\{ 2\left\Vert I-\Re\!\left[M\right]\right\Vert +\left(n-1\right)\left\Vert I-M\right\Vert ^{2}:M=V_{1}^{\dag}\left(W_{E}\otimes I_{B}\right)V_{2}\right\} \\
=\inf_{\substack{\lambda,\mu\geq0,\\
W,M\in\mathbb{L}
}
}\left\{ \begin{array}{c}
\lambda+\left(n-1\right)\mu:\lambda I\geq2\left(I-\Re\!\left[M\right]\right)\geq-\lambda I,\ \begin{bmatrix}\mu I & I-M\\
I-M^{\dag} & I
\end{bmatrix}\geq0,\\
\begin{bmatrix}I & W\\
W^{\dag} & I
\end{bmatrix}\geq0,\ M=V_{1}^{\dag}\left(W_{E}\otimes I_{B}\right)V_{2}
\end{array}\right\} ,
\end{multline}
\end{prop}

\begin{proof}
Applying~\eqref{eq:spectral-norm-herm-op} for $2\left\Vert I-\Re\!\left[M\right]\right\Vert $,
\eqref{eq:SDP-inf-norm-squared-general} for $\left\Vert I-M\right\Vert ^{2}$,
and~\eqref{eq:contraction-as-SDP-constraint} for $\left\Vert W\right\Vert \leq1$,
consider that
\begin{multline}
\inf_{W\in\mathbb{B}}\left\{ \left\Vert 2\left(I-\Re\!\left[V_{1}^{\dag}\left(W_{E}\otimes I_{B}\right)V_{2}\right]\right)\right\Vert +\left(n-1\right)\left\Vert I-V_{1}^{\dag}\left(W_{E}\otimes I_{B}\right)V_{2}\right\Vert ^{2}\right\} \\
=\inf_{\substack{\lambda,\mu\geq0,\\
W\in\mathbb{L}
}
}\left\{ \begin{array}{c}
\lambda+\left(n-1\right)\mu:\\
\lambda I\geq2\left(I-\Re\!\left[V_{1}^{\dag}\left(W_{E}\otimes I_{B}\right)V_{2}\right]\right)\geq-\lambda I,\\
\begin{bmatrix}\mu I & I-V_{1}^{\dag}\left(W_{E}\otimes I_{B}\right)V_{2}\\
I-\left(V_{1}^{\dag}\left(W_{E}\otimes I_{B}\right)V_{2}\right)^{\dag} & I
\end{bmatrix}\geq0,\\
\begin{bmatrix}I & W\\
W^{\dag} & I
\end{bmatrix}\geq0
\end{array}\right\} ,
\end{multline}
thus concluding the proof.
\end{proof}
\begin{prop}
\label{prop:upper-bound-adaptive-SDP-1d-search}The upper bound in
Theorem~\ref{thm:seq-bures-distance-channels} for adaptive discrimination
of channels $\mathcal{N}_{1}$ and $\mathcal{N}_{2}$ can be expressed
in the following way:
\begin{multline}
\inf_{\substack{W,M\in\mathbb{L}:\\
\left\Vert W\right\Vert \leq1
}
}\left\{ 2\left\Vert I-\Re\!\left[M\right]\right\Vert +\left(n-1\right)\left(2\left\Vert I-\Re\!\left[M\right]\right\Vert \right)^{\frac{1}{2}}\left\Vert I-M\right\Vert :M=V_{1}^{\dag}\left(W_{E}\otimes I_{B}\right)V_{2}\right\} \\
=\inf_{\nu\in\left(0,1\right]}\inf_{\substack{\lambda,\mu\geq0,\\
W\in\mathbb{L}
}
}\left\{ \begin{array}{c}
\left(1+\frac{n-1}{2}\nu\right)\lambda+\left(\frac{n-1}{2\nu}\right)\mu:\\
\lambda I\geq2\left(I-\Re\!\left[M\right]\right)\geq-\lambda I,\ \begin{bmatrix}\mu I & I-M\\
I-M^{\dag} & I
\end{bmatrix}\geq0,\\
\begin{bmatrix}I & W\\
W^{\dag} & I
\end{bmatrix}\geq0,\ M=V_{1}^{\dag}\left(W_{E}\otimes I_{B}\right)V_{2},
\end{array}\right\} .\label{eq:SDP-with-outer-loop}
\end{multline}
\end{prop}

\begin{proof}
Consider that, for $a,b\geq0$, such that $b\leq a$, the following
equality holds: 
\begin{equation}
\sqrt{ab}=\frac{1}{2}\inf_{\nu\in\left(0,1\right]}\left\{ \nu a+\frac{b}{\nu}\right\} ,\label{eq:geo-mean-opt}
\end{equation}
while noting that the optimal value of $\nu$ is $\sqrt{\frac{b}{a}}$.
Keeping the constraint $M=V_{1}^{\dag}\left(W_{E}\otimes I_{B}\right)V_{2}$
implicit, applying~\eqref{eq:geo-mean-opt}, while noting Remark~\ref{rem:parallel-adaptive-disc-compare-up-bnds},
it follows that
\begin{align}
 & \inf_{\substack{W,M\in\mathbb{L}:\\
\left\Vert W\right\Vert \leq1
}
}\left\{ 2\left\Vert I-\Re\!\left[M\right]\right\Vert +\left(n-1\right)\left(2\left\Vert I-\Re\!\left[M\right]\right\Vert \right)^{\frac{1}{2}}\left\Vert I-M\right\Vert \right\} \nonumber \\
 & =\inf_{\substack{W,M\in\mathbb{L}:\\
\left\Vert W\right\Vert \leq1
}
}\left\{ 2\left\Vert I-\Re\!\left[M\right]\right\Vert +\left(n-1\right)\sqrt{2\left\Vert I-\Re\!\left[M\right]\right\Vert \left\Vert I-M\right\Vert ^{2}}\right\} \\
 & =\inf_{\nu\in\left(0,1\right]}\inf_{\substack{W,M\in\mathbb{L}:\\
\left\Vert W\right\Vert \leq1
}
}\left\{ 2\left\Vert I-\Re\!\left[M\right]\right\Vert +\left(\frac{n-1}{2}\right)\left(\nu\left(2\left\Vert I-\Re\!\left[M\right]\right\Vert \right)+\frac{\left\Vert I-M\right\Vert ^{2}}{\nu}\right)\right\} \\
 & =\inf_{\nu\in\left(0,1\right]}\inf_{\substack{W,M\in\mathbb{L}:\\
\left\Vert W\right\Vert \leq1
}
}\left\{ 1+\frac{n-1}{2}\nu\left(2\left\Vert I-\Re\!\left[M\right]\right\Vert \right)+\left(\frac{n-1}{2}\right)\frac{\left\Vert I-M\right\Vert ^{2}}{\nu}\right\} .
\end{align}
Applying~\eqref{eq:spectral-norm-herm-op} for $2\left\Vert I-\Re\!\left[M\right]\right\Vert $
and~\eqref{eq:SDP-inf-norm-squared-general} for $\left\Vert I-M\right\Vert ^{2}$,
we conclude~\eqref{eq:SDP-with-outer-loop}.
\end{proof}
\begin{rem}
\label{rem:adaptive-ch-disc-opt}In practice, the optimization in
\eqref{eq:SDP-with-outer-loop} can be computed efficiently, even
though it is not an SDP. Indeed,~\eqref{eq:SDP-with-outer-loop} can
be rewritten in the following way:
\begin{equation}
\inf_{\nu\in\left(0,1\right]}J(\nu),
\end{equation}
where
\begin{equation}
J(\nu)\coloneqq\inf_{\substack{\lambda,\mu\geq0,\\
W\in\mathbb{L}
}
}\left\{ \begin{array}{c}
\left(1+\frac{n-1}{2}\nu\right)\lambda+\left(\frac{n-1}{2\nu}\right)\mu:\\
\lambda I\geq2\left(I-\Re\!\left[M\right]\right)\geq-\lambda I,\ \begin{bmatrix}\mu I & I-M\\
I-M^{\dag} & I
\end{bmatrix}\geq0,\\
\begin{bmatrix}I & W\\
W^{\dag} & I
\end{bmatrix}\geq0,\ M=V_{1}^{\dag}\left(W_{E}\otimes I_{B}\right)V_{2},
\end{array}\right\} .
\end{equation}
For each fixed $\nu\in\left(0,1\right]$, the function $J(\nu)$ can
be computed efficiently because it is an SDP. As such, one can perform
a global one-dimensional line search for the value of $\nu\in\left(0,1\right]$
that minimizes $J(\nu)$. The global search is inexpensive in practice
because $\nu$ is one-dimensional, and there are well known methods
for performing this step~\cite{Press2007}. In more detail, the outer
one-dimensional optimization over $\nu$ is nonconvex in general.
However, since the problem is scalar, a global optimum can be obtained
efficiently by a coarse-to-fine search strategy: we can first sample
$J(\nu)$ on a logarithmic grid, and then refine the best candidates
using Brent’s derivative-free minimization algorithm~\cite{Brent1973}.
\end{rem}

\subsection{Optimization of quantities based on SLD Fisher information}

In this section, we provide various methods for optimizing quantities
involving the SLD Fisher information. We note that, recently, a numerical
package is available for this purpose~\cite{Dulian2025}. In some
cases, the approaches we outline below differ from those given in
prior works~\cite{Kurdzialek2023,Dulian2025}.

Let us begin with the SLD Fisher information of channels, while noting
that a similar SDP for it was proposed in~\cite[Eqs.~(21)--(22)]{DemkowiczDobrzanski2012}.
\begin{prop}
The SLD Fisher information of the smooth channel family $\left(\mathcal{N}_{\theta}\right)_{\theta\in\Theta}$
can be expressed as the following SDP:
\begin{align}
\frac{1}{4}I_{F}\!\left(\theta;\left(\mathcal{N}_{\theta}\right)_{\theta\in\Theta}\right) & =\inf_{H_{\theta}\in\mathbb{H}}\left\Vert M_{H_{\theta}}\right\Vert ^{2}\\
 & =\inf_{\lambda\geq0,H_{\theta}\in\mathbb{H}}\left\{ \lambda:\begin{bmatrix}\lambda I & M_{H_{\theta}}\\
M_{H_{\theta}}^{\dag} & I
\end{bmatrix}\geq0\right\} ,
\end{align}
where
\begin{equation}
M_{H_{\theta}}\equiv\partial_{\theta}V_{\theta}-i\left(H_{\theta}\otimes I_{B}\right)V_{\theta}.
\end{equation}
\end{prop}

\begin{proof}
This is a direct consequence of~\eqref{eq:SLD-channels-final-exp-thm}
and~\eqref{eq:SDP-inf-norm-squared-general}, while incorporating
the optimization over $H_{\theta}\in\mathbb{H}$.
\end{proof}
The following SDP is similar to that given in~\cite[Appendix~E]{Kolodynski2013}:
\begin{prop}
For all $n\in\mathbb{N}$, the upper bound in Theorem~\ref{thm:upper-bound-SLD-Fisher-parallel}
for parallel estimation of the smooth channel family $\left(\mathcal{N}_{\theta}\right)_{\theta\in\Theta}$
can be expressed as the following SDP:
\begin{multline}
\inf_{H_{\theta}\in\mathbb{H}}\left\{ \left\Vert M_{H_{\theta}}\right\Vert ^{2}+\left(n-1\right)\left\Vert V_{\theta}^{\dag}M_{H_{\theta}}\right\Vert ^{2}\right\} \\
=\inf_{\substack{\lambda,\mu\geq0,\\
H_{\theta}\in\mathbb{H}
}
}\left\{ \begin{array}{c}
\lambda+\left(n-1\right)\mu:\begin{bmatrix}\lambda I & M_{H_{\theta}}\\
M_{H_{\theta}}^{\dag} & I
\end{bmatrix}\geq0,\ \begin{bmatrix}\mu I & V_{\theta}^{\dag}M_{H_{\theta}}\\
M_{H_{\theta}}^{\dag}V_{\theta} & I
\end{bmatrix}\geq0,\end{array}\right\} ,
\end{multline}
where
\begin{equation}
M_{H_{\theta}}\equiv\partial_{\theta}V_{\theta}-i\left(H_{\theta}\otimes I_{B}\right)V_{\theta}.
\end{equation}
\end{prop}

\begin{proof}
This is a direct consequence of~\eqref{eq:SDP-inf-norm-squared-general},
while incorporating the optimization over $H_{\theta}\in\mathbb{H}$.
\end{proof}
\begin{prop}
For all $n\in\mathbb{N}$, the upper bound in Theorem~\ref{thm:upper-bound-SLD-Fisher-adaptive}
for adaptive estimation of the smooth channel family $\left(\mathcal{N}_{\theta}\right)_{\theta\in\Theta}$
can be expressed in the following way:
\begin{multline}
\inf_{H_{\theta}\in\mathbb{H}}\left\{ \left\Vert M_{H_{\theta}}\right\Vert ^{2}+\left(n-1\right)\left\Vert M_{H_{\theta}}\right\Vert \left\Vert V_{\theta}^{\dag}M_{H_{\theta}}\right\Vert \right\} =\\
\inf_{\nu\in\left(0,1\right]}\inf_{\substack{\lambda,\mu\geq0,\\
H_{\theta}\in\mathbb{H}
}
}\left\{ \left(1+\frac{n-1}{2}\nu\right)\lambda+\left(\frac{n-1}{2\nu}\right)\mu:\begin{bmatrix}\lambda I & M_{H_{\theta}}\\
M_{H_{\theta}} & I
\end{bmatrix}\geq0,\begin{bmatrix}\mu I & V_{\theta}^{\dag}M_{H_{\theta}}\\
M_{H_{\theta}}^{\dag}V_{\theta} & I
\end{bmatrix}\geq0\right\} ,\label{eq:adaptive-ch-est-up-bnd-efficient}
\end{multline}
where
\begin{equation}
M_{H_{\theta}}\equiv\partial_{\theta}V_{\theta}-i\left(H_{\theta}\otimes I_{B}\right)V_{\theta}.
\end{equation}
\end{prop}

\begin{proof}
The proof follows a similar line of reasoning as given for Proposition
\ref{prop:upper-bound-adaptive-SDP-1d-search}.
\end{proof}
By the same line of reasoning given in Remark~\ref{rem:adaptive-ch-disc-opt},
the optimization in~\eqref{eq:adaptive-ch-est-up-bnd-efficient} can
be computed efficiently by a combination of semi-definite optimization
and one-dimensional grid search. This approach differs from that outlined
in~\cite[Appendix~D]{Kurdzialek2023}.

\section{Conclusion}\label{sec:Conclusion}

In summary, we provided a unified treatment for quantum channel discrimination
and estimation, and we established lower bounds on the error probabilities
and query complexities for these tasks, in both the parallel and adaptive
access models. As highlighted throughout, an advantage of our approach
is that we formulate all statements and proofs in terms of isometric
extensions of quantum channels, so that many proofs rely on basic
properties of isometries and the spectral norm. Along the way, we
also revisited many foundational statements in the quantum estimation
theory literature, providing alternative proofs that may enhance general
understanding of them.

Going forward, there are a number of questions that can be addressed
using our framework. First, can one easily generalize the various
bounds presented here to the energy-constrained case? There has already
been some work on this topic~\cite{Chen2025}, and we suspect that
the operator $E$-norm of~\cite{Shirokov2020}, combined with our
approach here, could be useful for this purpose. Additionally, we
wonder about extending the results to obtain lower bounds on the query
complexity of multiple channel discrimination. Bounds for this task
are available in~\cite[Theorem~19]{Nuradha2025}, but the approach
of our paper, possibly combined with the notion of multivariate fidelities
\cite{Nuradha2025a}, would presumably lead to different bounds.

\medskip{}

\textit{Acknowledgements}---ZH is supported by an ARC DECRA Fellowship
(DE230100144) ``Quantum-enabled super-resolution imaging'' and an
RMIT Vice Chancellor's Senior Research Fellowship. She is also
grateful to the Cornell School of Electrical and Computer Engineering for hospitality during an October 2025 research visit, as well as the
Cornell Lab of Ornithology for the Merlin app, which has enabled many hours of fruitful birding. 
JJM acknowledges
support from QuantERA (HQCC).
TN acknowledges support from the
Department of Mathematics and the IQUIST Postdoctoral Fellowship from
the Illinois Quantum Information Science and Technology Center at
the University of Illinois Urbana-Champaign. MMW acknowledges support
from the National Science Foundation under grant no.~2329662 and
from the Cornell School of Electrical and Computer Engineering.

\bibliographystyle{alphaurl}
\phantomsection\addcontentsline{toc}{section}{\refname}\bibliography{Ref}

\newcommand{\etalchar}[1]{$^{#1}$}
\begin{thebibliography}{ACMT{\etalchar{+}}07}

\bibitem[Ac{\'i}01]{Acin2001}
Antonio Ac{\'i}n.
\newblock Statistical distinguishability between unitary operations.
\newblock {\em Physical Review Letters}, 87:177901, October 2001.
\newblock URL: \url{https://link.aps.org/doi/10.1103/PhysRevLett.87.177901}, \href {https://doi.org/10.1103/PhysRevLett.87.177901} {\path{doi:10.1103/PhysRevLett.87.177901}}.

\bibitem[ACMT{\etalchar{+}}07]{Audenaert2007}
K.~M.~R. Audenaert, J.~Calsamiglia, R.~Mu{\~{n}}oz-Tapia, E.~Bagan, Ll. Masanes, A.~Acin, and F.~Verstraete.
\newblock Discriminating states: The quantum {Chernoff} bound.
\newblock {\em Physical Review Letters}, 98:160501, April 2007.
\newblock \href {https://doi.org/10.1103/physrevlett.98.160501} {\path{doi:10.1103/physrevlett.98.160501}}.

\bibitem[BB84]{Bennett1984}
Charles~H. Bennett and Gilles Brassard.
\newblock Quantum cryptography: Public key distribution and coin tossing.
\newblock In {\em Proceedings of the IEEE International Conference on Computers, Systems and Signal Processing}, pages 175--179, 1984.
\newblock \href {https://arxiv.org/abs/2003.06557} {\path{arXiv:2003.06557}}.

\bibitem[BBC{\etalchar{+}}93]{Bennett1993}
Charles~H. Bennett, Gilles Brassard, Claude Cr\'epeau, Richard Jozsa, Asher Peres, and William~K. Wootters.
\newblock Teleporting an unknown quantum state via dual classical and {Einstein-Podolsky-Rosen} channels.
\newblock {\em Physical Review Letters}, 70:1895--1899, March 1993.
\newblock URL: \url{https://link.aps.org/doi/10.1103/PhysRevLett.70.1895}, \href {https://doi.org/10.1103/PhysRevLett.70.1895} {\path{doi:10.1103/PhysRevLett.70.1895}}.

\bibitem[Bre73]{Brent1973}
Richard~P. Brent.
\newblock {\em Algorithms for Minimization without Derivatives}.
\newblock Prentice–Hall, Englewood Cliffs, NJ, 1973.
\newblock Chapters 5--6: Brent’s method for one-dimensional minimization.

\bibitem[Bur69]{Bures1969}
Donald Bures.
\newblock An extension of {Kakutani's} theorem on infinite product measures to the tensor product of semifinite {$\omega^*$}-algebras.
\newblock {\em Transactions of the American Mathematical Society}, 135(0):199--212, 1969.
\newblock \href {https://doi.org/10.1090/s0002-9947-1969-0236719-2} {\path{doi:10.1090/s0002-9947-1969-0236719-2}}.

\bibitem[BW92]{Bennett1992}
Charles~H. Bennett and Stephen~J. Wiesner.
\newblock Communication via one- and two-particle operators on {Einstein-Podolsky-Rosen} states.
\newblock {\em Physical Review Letters}, 69:2881--2884, November 1992.
\newblock URL: \url{https://link.aps.org/doi/10.1103/PhysRevLett.69.2881}, \href {https://doi.org/10.1103/PhysRevLett.69.2881} {\path{doi:10.1103/PhysRevLett.69.2881}}.

\bibitem[Cam73]{LeCam1973}
Lucien M.~Le Cam.
\newblock Convergence of estimates under dimensionality restrictions.
\newblock {\em The Annals of Statistics}, 1(1):38--53, 1973.
\newblock Institute of Mathematical Statistics.
\newblock \href {https://doi.org/10.1214/aos/1193342380} {\path{doi:10.1214/aos/1193342380}}.

\bibitem[CDL{\etalchar{+}}25]{Cheng2025}
Hao-Chung Cheng, Nilanjana Datta, Nana Liu, Theshani Nuradha, Robert Salzmann, and Mark~M. Wilde.
\newblock An invitation to the sample complexity of quantum hypothesis testing.
\newblock {\em npj Quantum Information}, 11(1):94, 2025.
\newblock \href {https://doi.org/10.1038/s41534-025-00980-8} {\path{doi:10.1038/s41534-025-00980-8}}.

\bibitem[CDP08]{Chiribella2008}
Giulio Chiribella, Giacomo~M. D'Ariano, and Paolo Perinotti.
\newblock Memory effects in quantum channel discrimination.
\newblock {\em Physical Review Letters}, 101:180501, October 2008.
\newblock URL: \url{https://link.aps.org/doi/10.1103/PhysRevLett.101.180501}, \href {https://doi.org/10.1103/PhysRevLett.101.180501} {\path{doi:10.1103/PhysRevLett.101.180501}}.

\bibitem[CDR13]{Chiribella2013}
Giulio Chiribella, Giacomo~Mauro D'Ariano, and Martin Roetteler.
\newblock On the query complexity of perfect gate discrimination.
\newblock In {\em 8th Conference on the Theory of Quantum Computation, Communication and Cryptography (TQC 2013)}, pages 178--191. Schloss Dagstuhl--Leibniz-Zentrum f{\"u}r Informatik, 2013.
\newblock \href {https://doi.org/10.4230/LIPIcs.TQC.2013.178} {\path{doi:10.4230/LIPIcs.TQC.2013.178}}.

\bibitem[CMW16]{Cooney2016}
Tom Cooney, Milán Mosonyi, and Mark~M. Wilde.
\newblock Strong converse exponents for a quantum channel discrimination problem and quantum-feedback-assisted communication.
\newblock {\em Communications in Mathematical Physics}, 344(3):797--829, 2016.
\newblock \href {https://doi.org/10.1007/s00220-016-2645-4} {\path{doi:10.1007/s00220-016-2645-4}}.

\bibitem[CY25]{Chen2025}
Longyun Chen and Yuxiang Yang.
\newblock Optimal quantum metrology under energy constraints, 2025.
\newblock URL: \url{https://arxiv.org/abs/2506.09436}, \href {https://arxiv.org/abs/2506.09436} {\path{arXiv:2506.09436}}.

\bibitem[DDKG12]{DemkowiczDobrzanski2012}
Rafał Demkowicz-Dobrzański, Jan Kołodyński, and Mădălin Guţă.
\newblock The elusive {H}eisenberg limit in quantum-enhanced metrology.
\newblock {\em Nature Communications}, 3(1):1063, 2012.
\newblock \href {https://doi.org/10.1038/ncomms2067} {\path{doi:10.1038/ncomms2067}}.

\bibitem[DDM14]{Demkowicz2014}
Rafał Demkowicz-Dobrzański and Lorenzo Maccone.
\newblock Using entanglement against noise in quantum metrology.
\newblock {\em Physical Review Letters}, 113:250801, December 2014.
\newblock URL: \url{https://link.aps.org/doi/10.1103/PhysRevLett.113.250801}, \href {https://doi.org/10.1103/PhysRevLett.113.250801} {\path{doi:10.1103/PhysRevLett.113.250801}}.

\bibitem[DFY07]{Duan2007}
Runyao Duan, Yuan Feng, and Mingsheng Ying.
\newblock Entanglement is not necessary for perfect discrimination between unitary operations.
\newblock {\em Physical Review Letters}, 98:100503, March 2007.
\newblock URL: \url{https://link.aps.org/doi/10.1103/PhysRevLett.98.100503}, \href {https://doi.org/10.1103/PhysRevLett.98.100503} {\path{doi:10.1103/PhysRevLett.98.100503}}.

\bibitem[DFY09]{Duan2009}
Runyao Duan, Yuan Feng, and Mingsheng Ying.
\newblock Perfect distinguishability of quantum operations.
\newblock {\em Physical Review Letters}, 103(21):210501, 2009.
\newblock \href {https://doi.org/10.1103/PhysRevLett.103.210501} {\path{doi:10.1103/PhysRevLett.103.210501}}.

\bibitem[DKDD25]{Dulian2025}
Piotr Dulian, Stanisław Kurdziałek, and Rafał Demkowicz-Dobrzański.
\newblock {QM}etro++ --- {P}ython optimization package for large scale quantum metrology with customized strategy structures, 2025.
\newblock URL: \url{https://arxiv.org/abs/2506.16524}, \href {https://arxiv.org/abs/2506.16524} {\path{arXiv:2506.16524}}.

\bibitem[EdMFD11]{Escher2011}
B.~M. Escher, R.~L. de~Matos~Filho, and L.~Davidovich.
\newblock General framework for estimating the ultimate precision limit in noisy quantum-enhanced metrology.
\newblock {\em Nature Physics}, 7:406--411, 2011.
\newblock \href {https://doi.org/10.1038/nphys1958} {\path{doi:10.1038/nphys1958}}.

\bibitem[Eke91]{Ekert1991}
Artur~K. Ekert.
\newblock Quantum cryptography based on {Bell}'s theorem.
\newblock {\em Physical Review Letters}, 67:661--663, August 1991.
\newblock URL: \url{https://link.aps.org/doi/10.1103/PhysRevLett.67.661}, \href {https://doi.org/10.1103/PhysRevLett.67.661} {\path{doi:10.1103/PhysRevLett.67.661}}.

\bibitem[FI08]{Fujiwara2008}
Akio Fujiwara and Hiroshi Imai.
\newblock A fibre bundle over manifolds of quantum channels and its application to quantum statistics.
\newblock {\em Journal of Physics A: Mathematical and Theoretical}, 41(25):255304, May 2008.
\newblock \href {https://doi.org/10.1088/1751-8113/41/25/255304} {\path{doi:10.1088/1751-8113/41/25/255304}}.

\bibitem[Fuj01]{Fujiwara2001}
Akio Fujiwara.
\newblock Quantum channel identification problem.
\newblock {\em Physical Review A}, 63:042304, March 2001.
\newblock URL: \url{https://link.aps.org/doi/10.1103/PhysRevA.63.042304}, \href {https://doi.org/10.1103/PhysRevA.63.042304} {\path{doi:10.1103/PhysRevA.63.042304}}.

\bibitem[GLM04]{Giovannetti2004}
Vittorio Giovannetti, Seth Lloyd, and Lorenzo Maccone.
\newblock Quantum-enhanced measurements: Beating the standard quantum limit.
\newblock {\em Science}, 306(5700):1330--1336, 2004.
\newblock \href {https://doi.org/10.1126/science.1104149} {\path{doi:10.1126/science.1104149}}.

\bibitem[GLN05]{Gilchrist2005}
Alexei Gilchrist, Nathan~K. Langford, and Michael~A. Nielsen.
\newblock Distance measures to compare real and ideal quantum processes.
\newblock {\em Physical Review A}, 71:062310, June 2005.
\newblock URL: \url{https://link.aps.org/doi/10.1103/PhysRevA.71.062310}, \href {https://doi.org/10.1103/PhysRevA.71.062310} {\path{doi:10.1103/PhysRevA.71.062310}}.

\bibitem[GRS18]{Gutoski2018}
Gus Gutoski, Ansis Rosmanis, and Jamie Sikora.
\newblock Fidelity of quantum strategies with applications to cryptography.
\newblock {\em {Quantum}}, 2:89, September 2018.
\newblock \href {https://doi.org/10.22331/q-2018-09-03-89} {\path{doi:10.22331/q-2018-09-03-89}}.

\bibitem[Hal50]{Halmos1950}
P.~R. Halmos.
\newblock Normal dilations and extensions of operators.
\newblock {\em Summa Brasiliensis Mathematicae}, II(VI):125--134, 1950.

\bibitem[Hay09]{Hayashi2009}
Masahito Hayashi.
\newblock Discrimination of two channels by adaptive methods and its application to quantum system.
\newblock {\em IEEE Transactions on Information Theory}, 55(8):3807--3820, 2009.
\newblock URL: \url{https://arxiv.org/abs/0804.0686}, \href {https://doi.org/10.1109/TIT.2009.2023726} {\path{doi:10.1109/TIT.2009.2023726}}.

\bibitem[Hay11]{Hayashi2011}
Masahito Hayashi.
\newblock Comparison between the {C}ramér--{R}ao and the mini-max approaches in quantum channel estimation.
\newblock {\em Communications in Mathematical Physics}, 304(3):689--709, 2011.
\newblock \href {https://doi.org/10.1007/s00220-011-1239-4} {\path{doi:10.1007/s00220-011-1239-4}}.

\bibitem[Hay17]{Hayashi2017}
Masahito Hayashi.
\newblock {\em Quantum Information Theory: Mathematical Foundation}.
\newblock Springer, second edition, 2017.
\newblock \href {https://doi.org/10.1007/978-3-662-49725-8} {\path{doi:10.1007/978-3-662-49725-8}}.

\bibitem[HHLW10]{Harrow2010}
Aram~W. Harrow, Andy Hassidim, Debbie~W. Leung, and John Watrous.
\newblock Adaptive versus nonadaptive strategies for quantum channel discrimination.
\newblock {\em Physical Review A}, 81(3):032339, 2010.
\newblock \href {https://doi.org/10.1103/PhysRevA.81.032339} {\path{doi:10.1103/PhysRevA.81.032339}}.

\bibitem[HKOT23]{Haah2023}
Jeongwan Haah, Robin Kothari, Ryan O’Donnell, and Ewin Tang.
\newblock Query-optimal estimation of unitary channels in diamond distance.
\newblock In {\em 2023 IEEE 64th Annual Symposium on Foundations of Computer Science (FOCS)}, pages 363--390, 2023.
\newblock \href {https://doi.org/10.1109/FOCS57990.2023.00028} {\path{doi:10.1109/FOCS57990.2023.00028}}.

\bibitem[HL22]{Huang2022}
Xiaowei Huang and Lvzhou Li.
\newblock Query complexity of unitary operation discrimination.
\newblock {\em Physica A: Statistical Mechanics and its Applications}, 604:127863, 2022.
\newblock \href {https://doi.org/10.1016/j.physa.2022.127863} {\path{doi:10.1016/j.physa.2022.127863}}.

\bibitem[Hol19]{Holevo2019}
Alexander~S. Holevo.
\newblock {\em Quantum Systems, Channels, Information: A Mathematical Introduction}, volume~16.
\newblock Walter de Gruyter, second edition, 2019.
\newblock \href {https://doi.org/10.1515/9783110642490} {\path{doi:10.1515/9783110642490}}.

\bibitem[Hü92]{Huebner1992}
Matthias Hübner.
\newblock Explicit computation of the {B}ures distance for density matrices.
\newblock {\em Physics Letters A}, 163(4):239--242, 1992.
\newblock URL: \url{https://www.sciencedirect.com/science/article/pii/037596019291004B}, \href {https://doi.org/10.1016/0375-9601(92)91004-B} {\path{doi:10.1016/0375-9601(92)91004-B}}.

\bibitem[IM21]{Ito2021}
Ryo Ito and Ryuhei Mori.
\newblock Lower bounds on the error probability of multiple quantum channel discrimination by the {B}ures angle and the trace distance, 2021.
\newblock \href {https://arxiv.org/abs/2107.03948} {\path{arXiv:2107.03948}}.

\bibitem[JWD{\etalchar{+}}08]{Ji2008}
Zhengfeng Ji, Guoming Wang, Runyao Duan, Yuan Feng, and Mingsheng Ying.
\newblock Parameter estimation of quantum channels.
\newblock {\em IEEE Transactions on Information Theory}, 54(11):5172--5185, 2008.
\newblock \href {https://doi.org/10.1109/TIT.2008.929940} {\path{doi:10.1109/TIT.2008.929940}}.

\bibitem[KDD13]{Kolodynski2013}
Jan Kołodyński and Rafał Demkowicz-Dobrzański.
\newblock Efficient tools for quantum metrology with uncorrelated noise.
\newblock {\em New Journal of Physics}, 15(7):073043, July 2013.
\newblock \href {https://doi.org/10.1088/1367-2630/15/7/073043} {\path{doi:10.1088/1367-2630/15/7/073043}}.

\bibitem[KGADD23]{Kurdzialek2023}
Stanisław Kurdziałek, Wojciech Górecki, Francesco Albarelli, and Rafał Demkowicz-Dobrzański.
\newblock Using adaptiveness and causal superpositions against noise in quantum metrology.
\newblock {\em Physical Review Letters}, 131:090801, August 2023.
\newblock URL: \url{https://link.aps.org/doi/10.1103/PhysRevLett.131.090801}, \href {https://doi.org/10.1103/PhysRevLett.131.090801} {\path{doi:10.1103/PhysRevLett.131.090801}}.

\bibitem[KKLGT19]{Kawachi2019}
Akinori Kawachi, Kenichi Kawano, Fran{\c{c}}ois Le~Gall, and Suguru Tamaki.
\newblock Quantum query complexity of unitary operator discrimination.
\newblock {\em IEICE Transactions on Information and Systems}, 102(3):483--491, 2019.
\newblock \href {https://doi.org/10.1587/transinf.2018FCP0012} {\path{doi:10.1587/transinf.2018FCP0012}}.

\bibitem[KSW08]{Kretschmann2008}
Dennis Kretschmann, Dirk Schlingemann, and Reinhard~F. Werner.
\newblock The information-disturbance tradeoff and the continuity of {S}tinespring's representation.
\newblock {\em IEEE Transactions on Information Theory}, 54(4):1708--1717, 2008.
\newblock \href {https://doi.org/10.1109/TIT.2008.917696} {\path{doi:10.1109/TIT.2008.917696}}.

\bibitem[KW21a]{Katariya2021a}
Vishal Katariya and Mark~M. Wilde.
\newblock Evaluating the advantage of adaptive strategies for quantum channel distinguishability.
\newblock {\em Physical Review A}, 104:052406, November 2021.
\newblock URL: \url{https://link.aps.org/doi/10.1103/PhysRevA.104.052406}, \href {https://doi.org/10.1103/PhysRevA.104.052406} {\path{doi:10.1103/PhysRevA.104.052406}}.

\bibitem[KW21b]{Katariya2021}
Vishal Katariya and Mark~M. Wilde.
\newblock Geometric distinguishability measures limit quantum channel estimation and discrimination.
\newblock {\em Quantum Information Processing}, 20(2):78, 2021.
\newblock \href {https://doi.org/10.1007/s11128-021-02992-7} {\path{doi:10.1007/s11128-021-02992-7}}.

\bibitem[KW24]{Khatri2024}
Sumeet Khatri and Mark~M. Wilde.
\newblock Principles of quantum communication theory: A modern approach, 2024.
\newblock URL: \url{https://arxiv.org/abs/2011.04672v2}, \href {https://arxiv.org/abs/2011.04672v2} {\path{arXiv:2011.04672v2}}.

\bibitem[Lan70]{Lancaster1970}
Peter Lancaster.
\newblock Explicit solutions of linear matrix equations.
\newblock {\em SIAM Review}, 12(4):544--566, October 1970.
\newblock \href {https://doi.org/10.1137/1012104} {\path{doi:10.1137/1012104}}.

\bibitem[LJW{\etalchar{+}}25]{Li2025}
Lei Li, Zhe Ji, Qing-Wen Wang, Shu-Qian Shen, and Ming Li.
\newblock Sampling complexity of quantum channel discrimination.
\newblock {\em Communications in Theoretical Physics}, 77(10):105101, June 2025.
\newblock URL: \url{https://dx.doi.org/10.1088/1572-9494/adcb9e}, \href {https://doi.org/10.1088/1572-9494/adcb9e} {\path{doi:10.1088/1572-9494/adcb9e}}.

\bibitem[MKSF{\etalchar{+}}25]{Meyer2025}
Johannes~Jakob Meyer, Sumeet Khatri, Daniel Stilck~Franca, Jens Eisert, and Philippe Faist.
\newblock Quantum metrology in the finite-sample regime.
\newblock {\em PRX Quantum}, 6:030336, August 2025.
\newblock URL: \url{https://link.aps.org/doi/10.1103/qbn1-p6bq}, \href {https://doi.org/10.1103/qbn1-p6bq} {\path{doi:10.1103/qbn1-p6bq}}.

\bibitem[MPW10]{Matthews2010}
William Matthews, Marco Piani, and John Watrous.
\newblock Entanglement in channel discrimination with restricted measurements.
\newblock {\em Physical Review A}, 82(3):032302, 2010.
\newblock \href {https://doi.org/10.1103/PhysRevA.82.032302} {\path{doi:10.1103/PhysRevA.82.032302}}.

\bibitem[NMLW25]{Nuradha2025a}
Theshani Nuradha, Hemant~K. Mishra, Felix Leditzky, and Mark~M. Wilde.
\newblock Multivariate fidelities.
\newblock {\em Journal of Physics A: Mathematical and Theoretical}, 58(16):165304, April 2025.
\newblock \href {https://doi.org/10.1088/1751-8121/adc645} {\path{doi:10.1088/1751-8121/adc645}}.

\bibitem[NW25]{Nuradha2025}
Theshani Nuradha and Mark~M. Wilde.
\newblock Query complexity of classical and quantum channel discrimination.
\newblock {\em Quantum Science and Technology}, 10(4):045075, October 2025.
\newblock \href {https://doi.org/10.1088/2058-9565/ae0a79} {\path{doi:10.1088/2058-9565/ae0a79}}.

\bibitem[PTVF07]{Press2007}
William~H. Press, Saul~A. Teukolsky, William~T. Vetterling, and Brian~P. Flannery.
\newblock {\em Numerical Recipes: The Art of Scientific Computing}.
\newblock Cambridge University Press, Cambridge, 3rd edition, 2007.
\newblock Section 10.2: Brent’s method and golden-section search.

\bibitem[PW09]{Piani2009}
Marco Piani and John Watrous.
\newblock All entangled states are useful for channel discrimination.
\newblock {\em Physical Review Letters}, 102:250501, June 2009.
\newblock URL: \url{https://link.aps.org/doi/10.1103/PhysRevLett.102.250501}, \href {https://doi.org/10.1103/PhysRevLett.102.250501} {\path{doi:10.1103/PhysRevLett.102.250501}}.

\bibitem[RS75]{Reed1975}
Michael Reed and Barry Simon.
\newblock {\em Methods of Modern Mathematical Physics {II}: {F}ourier Analysis, Self-Adjointness}.
\newblock Academic Press, 1975.
\newblock Section X.12 Time-dependent Hamiltonians.

\bibitem[RYCS22]{Rossi2022}
Zane~M. Rossi, Jeffery Yu, Isaac~L. Chuang, and Sho Sugiura.
\newblock Quantum advantage for noisy channel discrimination.
\newblock {\em Physical Review A}, 105:032401, March 2022.
\newblock URL: \url{https://link.aps.org/doi/10.1103/PhysRevA.105.032401}, \href {https://doi.org/10.1103/PhysRevA.105.032401} {\path{doi:10.1103/PhysRevA.105.032401}}.

\bibitem[Sac05]{Sacchi2005}
Massimiliano~F. Sacchi.
\newblock Optimal discrimination of quantum operations.
\newblock {\em Physical Review A}, 71:062340, June 2005.
\newblock URL: \url{https://link.aps.org/doi/10.1103/PhysRevA.71.062340}, \href {https://doi.org/10.1103/PhysRevA.71.062340} {\path{doi:10.1103/PhysRevA.71.062340}}.

\bibitem[SDD26]{Sieniawski2025}
Stanisław Sieniawski and Rafał Demkowicz-Dobrzański.
\newblock Adaptive quantum channel discrimination using methods of quantum metrology.
\newblock {\em New Journal of Physics}, 28(2):024502, 2026.
\newblock URL: \url{https://iopscience.iop.org/article/10.1088/1367-2630/ae3edd}, \href {https://arxiv.org/abs/2510.15506} {\path{arXiv:2510.15506}}, \href {https://doi.org/10.1088/1367-2630/ae3edd} {\path{doi:10.1088/1367-2630/ae3edd}}.

\bibitem[Shi20]{Shirokov2020}
Maksim~E. Shirokov.
\newblock Operator {$E$}-norms and their use.
\newblock {\em Sbornik: Mathematics}, 211(9):1309--1338, 2020.
\newblock \href {https://doi.org/10.1070/sm9336} {\path{doi:10.1070/sm9336}}.

\bibitem[SHW22]{Salek2022}
Farzin Salek, Masahito Hayashi, and Andreas Winter.
\newblock Usefulness of adaptive strategies in asymptotic quantum channel discrimination.
\newblock {\em Physical Review A}, 105:022419, 2022.
\newblock URL: \url{https://arxiv.org/abs/2011.06569}, \href {https://arxiv.org/abs/arXiv:2011.06569} {\path{arXiv:arXiv:2011.06569}}, \href {https://doi.org/10.1103/PhysRevA.105.022419} {\path{doi:10.1103/PhysRevA.105.022419}}.

\bibitem[SK20]{Sidhu2020}
Jasminder~S. Sidhu and Pieter Kok.
\newblock Geometric perspective on quantum parameter estimation.
\newblock {\em AVS Quantum Science}, 2(1):014701, February 2020.
\newblock \href {https://doi.org/10.1116/1.5119961} {\path{doi:10.1116/1.5119961}}.

\bibitem[SM06]{Sarovar2006}
Mohan Sarovar and G.~J. Milburn.
\newblock Optimal estimation of one-parameter quantum channels.
\newblock {\em Journal of Physics A: Mathematical and General}, 39(26):8487, June 2006.
\newblock \href {https://doi.org/10.1088/0305-4470/39/26/015} {\path{doi:10.1088/0305-4470/39/26/015}}.

\bibitem[Uhl76]{Uhlmann1976}
Armin Uhlmann.
\newblock The ``transition probability'' in the state space of a {$*$}-algebra.
\newblock {\em Reports on Mathematical Physics}, 9(2):273--279, 1976.
\newblock \href {https://doi.org/10.1016/0034-4877(76)90060-4} {\path{doi:10.1016/0034-4877(76)90060-4}}.

\bibitem[Wat18]{Watrous2018}
John Watrous.
\newblock {\em The Theory of Quantum Information}.
\newblock Cambridge University Press, 2018.
\newblock \href {https://doi.org/10.1017/9781316848142} {\path{doi:10.1017/9781316848142}}.

\bibitem[WBHK20]{Wilde2020}
Mark~M. Wilde, Mario Berta, Christoph Hirche, and Eneet Kaur.
\newblock Amortized channel divergence for asymptotic quantum channel discrimination.
\newblock {\em Letters in Mathematical Physics}, 110:2277--2336, August 2020.
\newblock \href {https://arxiv.org/abs/1808.01498} {\path{arXiv:1808.01498}}, \href {https://doi.org/10.1007/s11005-020-01297-7} {\path{doi:10.1007/s11005-020-01297-7}}.

\bibitem[Wil17]{Wilde2017}
Mark~M. Wilde.
\newblock {\em Quantum Information Theory}.
\newblock Cambridge University Press, February 2017.
\newblock \href {https://doi.org/10.1017/9781316809976} {\path{doi:10.1017/9781316809976}}.

\bibitem[Wil25]{Wilde2025}
Mark~M. Wilde.
\newblock Quantum {F}isher information matrices from {R}\'enyi relative entropies, 2025.
\newblock URL: \url{https://arxiv.org/abs/2510.02218}, \href {https://arxiv.org/abs/2510.02218} {\path{arXiv:2510.02218}}.

\bibitem[WW19]{Wang2019}
Xin Wang and Mark~M. Wilde.
\newblock Resource theory of asymmetric distinguishability for quantum channels.
\newblock {\em Physical Review Research}, 1:033169, December 2019.
\newblock URL: \url{https://link.aps.org/doi/10.1103/PhysRevResearch.1.033169}, \href {https://doi.org/10.1103/PhysRevResearch.1.033169} {\path{doi:10.1103/PhysRevResearch.1.033169}}.

\bibitem[YF17]{Yuan2017}
Haidong Yuan and Chi-Hang~Fred Fung.
\newblock Fidelity and {F}isher information on quantum channels.
\newblock {\em New Journal of Physics}, 19(11):113039, November 2017.
\newblock \href {https://doi.org/10.1088/1367-2630/aa874c} {\path{doi:10.1088/1367-2630/aa874c}}.

\bibitem[Yu97]{Yu1997}
Bin Yu.
\newblock {A}ssouad, {F}ano, and {L}e {C}am.
\newblock In David Pollard, Erik Torgersen, and Grace~L. Yang, editors, {\em Festschrift for Lucien Le Cam: Research Papers in Probability and Statistics}, pages 423--435. Springer New York, 1997.
\newblock \href {https://doi.org/10.1007/978-1-4612-1880-7_29} {\path{doi:10.1007/978-1-4612-1880-7_29}}.

\bibitem[ZJ21]{Zhou2021}
Sisi Zhou and Liang Jiang.
\newblock Asymptotic theory of quantum channel estimation.
\newblock {\em PRX Quantum}, 2:010343, March 2021.
\newblock URL: \url{https://link.aps.org/doi/10.1103/PRXQuantum.2.010343}, \href {https://doi.org/10.1103/PRXQuantum.2.010343} {\path{doi:10.1103/PRXQuantum.2.010343}}.

\end{thebibliography}

\end{document}